\documentclass{amsart}

\title[Triangulations and solitons]
{Triangulations and soliton graphs for \\
totally positive Grassmannian}

\author{Rachel Karpman and Yuji Kodama} 
\date{\today}
\thanks{The second author was partially
supported by NSF grant DMS-1714770.}
\address{Department of Mathematics, Ohio State University, Columbus, OH 43210}
\email{karpman.2@osu.edu, kodama.1@osu.edu}
\subjclass[2000]{}

\usepackage{url}
\usepackage{color}
\usepackage{rotating}% provides a rotate environment

\def\vblack(#1, #2)#3{\cnode*[linecolor=black](#1, #2){3}{#3}}
\def\vwhite(#1,#2)#3{\cnode[linecolor=black,fillcolor=white,fillstyle=solid](#1,#2){3}{#3}}
\countdef\x=23
\countdef\y=24
\countdef\z=25
\countdef\t=26

\def\tbox(#1,#2)#3{
\x=#1 \y=#2 
\multiply\x by 12 
\multiply\y by 12 
\z=\x \t=\y
\advance\z by 12 
\advance\t by 12 
\psline(\x,\y)(\x,\t)(\z,\t)(\z,\y)(\x,\y)
\advance\x by 6
\advance\y by 6 
\rput(\x,\y){{\bf #3}}}

\usepackage{amssymb}
\usepackage{graphicx}
\usepackage{epstopdf}
\usepackage{amscd}
\usepackage{graphics}

\def\proof{\par{\it Proof}. \ignorespaces}
\def\endproof{{\ \vbox{\hrule\hbox{%
     \vrule height1.3ex\hskip0.8ex\vrule}\hrule }}\par}

\theoremstyle{definition}

\theoremstyle{remark}

\numberwithin{equation}{section}

\let\trueint=\int
\let\truesum=\sum
\def\int{\mathop{\textstyle\trueint}\limits}
\def\sum{\mathop{\textstyle\truesum}\limits}
\let\<=\langle
\let\>=\rangle

\def\sech{\mathop{\rm sech}\nolimits}
\def\half{{\textstyle\frac12}}

\def\Gr{{\rm Gr}}
\def\u{\mathbf{u}}
\def\v{\mathbf{v}}
\def\pp{\hat{\mathbf{p}}}

\def\t{{\bf t}}

\def\r{{\bf r}}
\newcommand{\p}{\mathbf{p}}

\def\P{{\mathsf P}}
\def\Q{{\mathsf Q}}

\usepackage{amsmath, amsthm, amssymb, amsbsy}
\usepackage{amsfonts, latexsym, stmaryrd, amscd, xy}
\usepackage[mathscr]{eucal}
\usepackage{epsfig}
\usepackage{young}
\newtheorem{theorem}{Theorem}[section]

\newtheorem{definition}[theorem]{Definition}
\newtheorem{proposition}[theorem]{Proposition}
\newtheorem{lemma}[theorem]{Lemma}
\newtheorem{example}[theorem]{Example}

\newtheorem{remark}[theorem]{Remark}
\newtheorem{algorithm}[theorem]{Algorithm}
\usepackage{amsfonts}
\usepackage{xy}
\usepackage{amssymb}
\usepackage{amsmath}

\newcommand{\R}{\mathbb R}

\newcommand{\PP}{{\sf P}}

\newcommand{\A}{\mathcal{A}}

\newcommand{\B}{\mathcal{B}}

\newcommand{\FF}{\mathcal{F}}

\newcommand{\back}{\backslash}
\let\k\kappa

\newcommand{\thmrefer}[1]{\renewcommand\thetheorem
  {\protect\ref{#1}}\addtocounter{theorem}{-1}}

\xyoption{all}
\CompileMatrices

\begin{document}

\begin{abstract}
The \emph{KP equation} is a nonlinear dispersive wave equation which provides an excellent model for resonant interactions of shallow-water waves. It is well known that regular \emph{soliton solutions} of the KP equation may be constructed from points in the \emph{totally nonnegative Grassmannian} $\Gr(N,M)_{\geq 0}$.  Kodama and Williams studied the asymptotic patterns (tropical limit) of KP solitons, called \emph{soliton graphs}, and showed that they correspond to Postnikov's \emph{Le-diagrams}. In this paper, we consider soliton graphs for the \emph{KP hierarchy}, a family of commuting flows which are compatible with the KP equation.  For the \emph{positive Grassmannian} $\Gr(2,M)_{>0}$, Kodama and Williams showed that soliton graphs are in bijection with triangulations of the $M$-gon.  We extend this result to  $\Gr(N,M)_{>0}$ when $N=3$ and
$M=6,7$ and $8$.  In each case, we show that soliton graphs are in bijection with Postnikov's \emph{plabic graphs}, which generalize Le-diagrams.
\end{abstract}

\maketitle
\setcounter{tocdepth}{1}
\tableofcontents

\section{Preliminaries}
\label{sec: prelim} In this section, we start to review the KP equation and soliton graphs.  Then we define the \emph{duality map}, which gives a correspondence between soliton graphs and \emph{soliton subdivisions}, and sketch some combinatorial background. At the end of this section, we describe the structure of the paper and state our main theorems.

\subsection{The KP equation and soliton solutions}
\label{subsec: KP}
The KP equation is a two-dimensional nonlinear dispersive wave equation given by
\begin{equation}\label{e:KP}
\frac{\partial}{\partial x}\left(-4\frac{\partial u}{\partial t}+6u\frac{\partial u}{\partial x}+\frac{\partial^3u}{\partial x^3}\right)+3\frac{\partial^2u}{\partial y^2}=0,
\end{equation}
where $u=u(x,y,t)$ represents the wave amplitude at the point $(x,y)$ for a fixed time $t$.
The KP equation was first proposed to study the stability of one-soliton solutions of the KdV equation
under the influence of weak transverse perturbations \cite{KP:70}.
The equation can also be derived from the three-dimensional Euler equation for an irrotational and incompressible fluid under the assumptions that it describes wave propagation of small amplitude,
long wavelength and quasi-two dimensional approximations (see e.g. \cite{K:10, KY:16}). 

We are interested in solutions of the KP equation that are regular in the entire $xy$-plane, where they are localized along certain line segments and rays.  We call such solutions \emph{line-soliton} solutions, or \emph{KP solitons} (see \cite{K:17} for a survey of the KP soliton).  To construct a KP soliton, it suffices to find a function $\tau(x,y,t)$ which satisfies the bilinear equation
\begin{equation}\label{e:HirotaKP}
P(D_x,D_y,D_t)\tau\cdot\tau:=\left(-4D_xD_t+D_x^4+3D_y^2\right)\tau\cdot \tau=0,
\end{equation}
where $D_z^n$ is the Hirota derivative defined by
\[
D_z^nf\cdot g:=\frac{\partial^n}{\partial s^n}f(z+s)g(z-s)\Big|_{s=0}.
\]
See for example \cite{H:04}.
The function
\begin{equation}\label{e:u-tau}
u(x,y,t)=2\frac{\partial^2}{\partial x^2}\ln\tau(x,y,t)
\end{equation}
then gives a solution.

\begin{remark}\label{rem:dispersion}
The dispersion relation of the KP equation is given by
\[
P(p,q,\omega)=-4p\omega+p^4+3q^2=0,
\]
which has the following parametrization,
\begin{equation}\label{e:parameterization}
p=\k_i-\k_j,\qquad q=\k_i^2-\k_j^2,\qquad \omega=\k_i^3-\k_j^3,
\end{equation}
with arbitrary constants $\k_i$ and $\k_j$.
\end{remark}

\begin{example}\label{Example:N = 1}
As a simple example, we may take the $\tau$-function
\[
\tau(x,y,t)=\sum_{i=1}^M a_ie^{\theta_i(x,y,t)}\qquad\text{with}\qquad \theta_i(x,y,t)=\kappa_i x+\kappa_i^2y+\kappa_i^3 t.
\]
where the parameters $a_i$ and $\kappa_i$ are real.
(Note that $P(D_x,D_y,D_t)e^{\theta_i}\cdot e^{\theta_j}=0$ from Remark \ref{rem:dispersion}.)
If the $a_i$ are nonnegative and not all zero, then $\tau>0$ everywhere; in other words, the solution is regular.  

Suppose $\tau$ has only two nonzero exponential terms, so that
$\tau=a_ie^{\theta_i}+a_je^{\theta_j}$
for some $i < j$, and we have
\begin{equation}\label{e:1soliton}
u(x,y,t)=\frac{(\k_i-\k_j)^2}{2}\sech^2\left(\frac{1}{2}\tilde\theta_{[i,j]}(x,y,t)\right)\qquad\text{with}\qquad
\tilde\theta_{[i,j]}=\theta_i-\theta_j+\ln\frac{a_i}{a_j}.
\end{equation}
Then $u(x,y,t)$ describes a wave with one peak, localized along the line $\tilde\theta_{[i,j]}(x,y,t)=0$.
We say this line is an $[i,j]$-soliton, or an $\{i,j\}$-soliton if we do not with to specify the order of $i$ and $j$.
\end{example}

We now describe a generalization of Example \ref{Example:N = 1}.  Fix real parameters $\kappa_1,\dots,\kappa_M$, and let $A = (a_{i,j})$ be a full-rank $N \times M$ matrix for some $N < M$.  We construct a $\tau$-function $\tau(x,y,t)$, and hence a KP soliton $u(x,y,t)$, from the matrix $A$. The function $\tau(x,y,t)$ 
is given by
the Wronskian determinant with respect to $x$ (see e.g. \cite{H:04, K:17} and the references listed therein)
\[\tau=\text{Wr}(f_1,\ldots,f_N),\]
where the scalar function $f_i(x,y,t)$ are given by 
\begin{equation}\label{e:f}
(f_1,\ldots,f_N)=(E_1,\ldots,E_M)A^T\qquad\text{with}\qquad E_j=\exp(\kappa_jx+\kappa_j^2y+\kappa_j^3t)\,.
\end{equation}
Throughout the paper, we assume the following order,
\begin{equation}\label{e:order}
\kappa_1~<~\kappa_2~<~\cdots~<~\kappa_M.
\end{equation}
Then using the Cauchy-Binet lemma for the determinant, the $\tau$-function with (\ref{e:f}) can be expressed as the following sum of exponential terms,
\begin{equation}\label{e:KPtau}
\tau(x,y,t)=\text{Wr}(f_1,\ldots,f_N)=\sum_{I\in\binom{[M]}{N}}\Delta_I(A)K_I\exp(\Theta_I(x,y,t)).
\end{equation}
Here $I=\{i_1<i_2<\ldots<i_N\}$, and $\binom{[M]}{N}$ denotes the set of all $N$-index subsets of $[M]:=\{1,\ldots,M\}.$
The term $K_I$ is defined by $K_I=\prod_{j>l}(\kappa_{i_j}-\kappa_{i_l})$, and 
the order \eqref{e:order} implies $K_I>0$.  The coefficient $\Delta_I(A)$ is the $N\times N$ minor of the matrix $A$ 
with the columns labeled by the index set $I$, and  the exponent $\Theta_I(x,y,t)$
is given by
\begin{equation}\label{e:Theta}
 \Theta_I(x,y,t)=\sum_{i\in I}\theta_i(x,y,t)=p_Ix+q_Iy+\omega_It,
 \end{equation}
where $p_I=\sum_{i\in I}\kappa_i, q_I=\sum_{i\in I}\kappa_i^2$ and $\omega_I=\sum_{i\in I}\kappa_i^3$.  It was then shown in \cite{KW:13} that  the $\tau$-function
is positive for all $(x,y,t)$ (i.e. the solution is regular) if and only if $\Delta_I(A)\ge 0$
for all the $N$-element subset $I$.  In this case, the matrix $A$ is called a \emph{totally nonnegative} (TNN) matrix \cite{postnikov}.

\begin{remark}\label{rem:KPhierarchy}
The KP equation admits an infinite number of commuting flows, and these flows all together
define the \emph{KP hierarchy}. Let $\{t_n:n=1,2,\ldots,\}$ denote the flow parameters.
Then the $\tau$-function for the KP hierarchy is also given in the same form as
\eqref{e:KPtau} where each $f_i$ now satisfies the linear equations
\[
\frac{\partial f_i}{\partial t_n}=\frac{\partial^n f_i}{\partial x^n}\qquad\text{for}\qquad n=1,2,\ldots,
\]
where $t_1=x, t_2=y, t_3=t$, and the rest of the $t_n$'s give the higher flows. (See e.g. \cite{MJD:00, K:17}.) 

We may extend the construction of the $\tau$-function to include some of these additional flow parameters. Let $A$ be an $N \times M$ matrix, and let $\hat{\mathbf{t}} = (t_1,\ldots,t_{M-1}).$  Let $x = t_1$ and $y = t_2$, and let 
$\mathbf{t} = (t_3,\ldots,t_{M-1})$, i.e. $\hat{\bf t}=(x,y,\mathbf{t})$.
Then we may substitute 
\[E_j(\hat{\mathbf{t}}) = \exp(\theta_j(\hat{\mathbf{t}})) :=\sum_{i = 1}^{M-1} \kappa_j^i t_i\]
for the $E_j$ in \eqref{e:f}.
Taking the Wronskian as above, we obtain the $\tau$-function for the KP hierarchy,
\begin{equation}\label{e:tau}
\tau(x,y,\mathbf{t}) = \sum_{I\in\binom{[M]}{N}}\Delta_I(A)K_I\exp(\Theta_I(x,y,\mathbf{t})),
\end{equation}
where the definition of $\Theta_I(x,y,\mathbf{t})$ is analogous to \eqref{e:Theta}, i.e.
$\Theta_I=\sum_{j\in I}\theta_j$ with
\[
\theta_j(x,y, \mathbf{t}) = \kappa_j x + \kappa_j^2 y + \sum_{i= 3}^{M-1} \kappa_j^i t_i.
\]
Setting $t = t_3$, and treating the remaining $t_i$-parameters as constants, we obtain a soliton solution to the KP equation.  
\end{remark}

\begin{remark}\label{Grassman}
The \emph{Grassmannian} $\Gr(N,M)$ is the parameter space of $N$-planes in $\mathbb{R}^M$.  Concretely, $\Gr(N,M)$ is the space of full-rank $N \times M$ matrices, modulo row operations.  A matrix $A$ corresponds to the span of its rows, and the map (the Pl\"ucker embedding)
\[
A~ \longmapsto~ \left\{\Delta_I(A):I \in \binom{[M]}{N}\right\},
\]
gives a system of homogeneous coordinates on $\Gr(N,M)$, known as \emph{Pl\"{u}cker coordinates}. Hence the construction $A \mapsto \tau(x,y,\mathbf{t})$ gives a soliton solution for each point in $\Gr(N,M)$, which is unique up to multiplication by a scalar.
Regular soliton solutions correspond to points in the \emph{totally nonnegative Grassmannian} $\Gr(M,N)_{\ge0}$, which is of considerable interest in its own right  \cite{postnikov, kodama}.
\end{remark}

%%%%%%%%%%%%%%%%%%%%%%%%%%%

\subsection{Soliton graphs}
\label{subsec: soliton graphs}

We are interested in the two-dimensional wave patterns generated by the regular KP solitons $u(x,y,\mathbf{t})$ constructed in the previous section. We may represent the wave pattern at a given time by a \emph{contour plot} showing the wave peaks in the $xy$-plane.  Figure \ref{fig:Sol36} shows the time evolution of the solution for a $3\times 6$ matrix $A$, with $\kappa$-parameters
\[(\kappa_1,\ldots,\kappa_6)=(-5/2, -5/4,-1/2,1/2,3/2,5/2).\] 
Here all $3\times 3$ minors of $A$ are nonzero (this type of matrix is called a \emph{totally positive} matrix). 
%%%%%%%%%%%%%%%%%%%%%%%%%%
\begin{figure}[h!]
\begin{center}
\includegraphics[width=3.5cm]{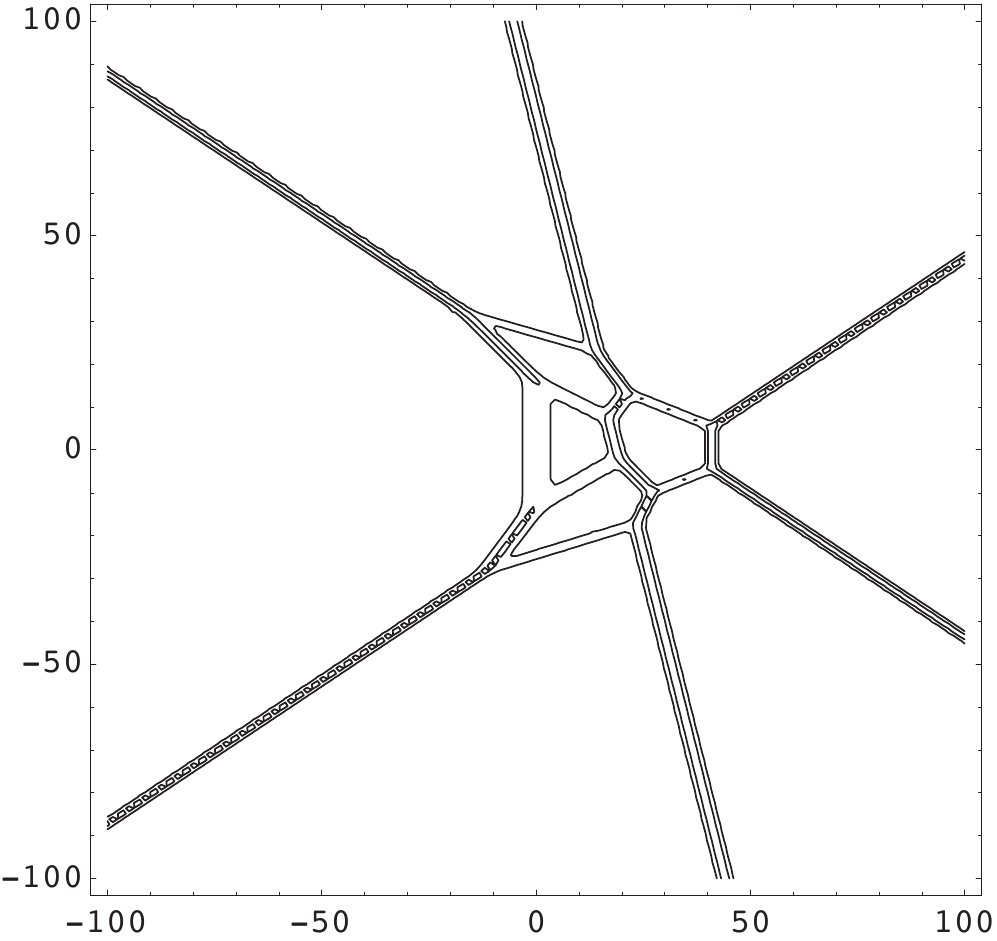}
\hskip0.3cm{\includegraphics[width=3.5cm]{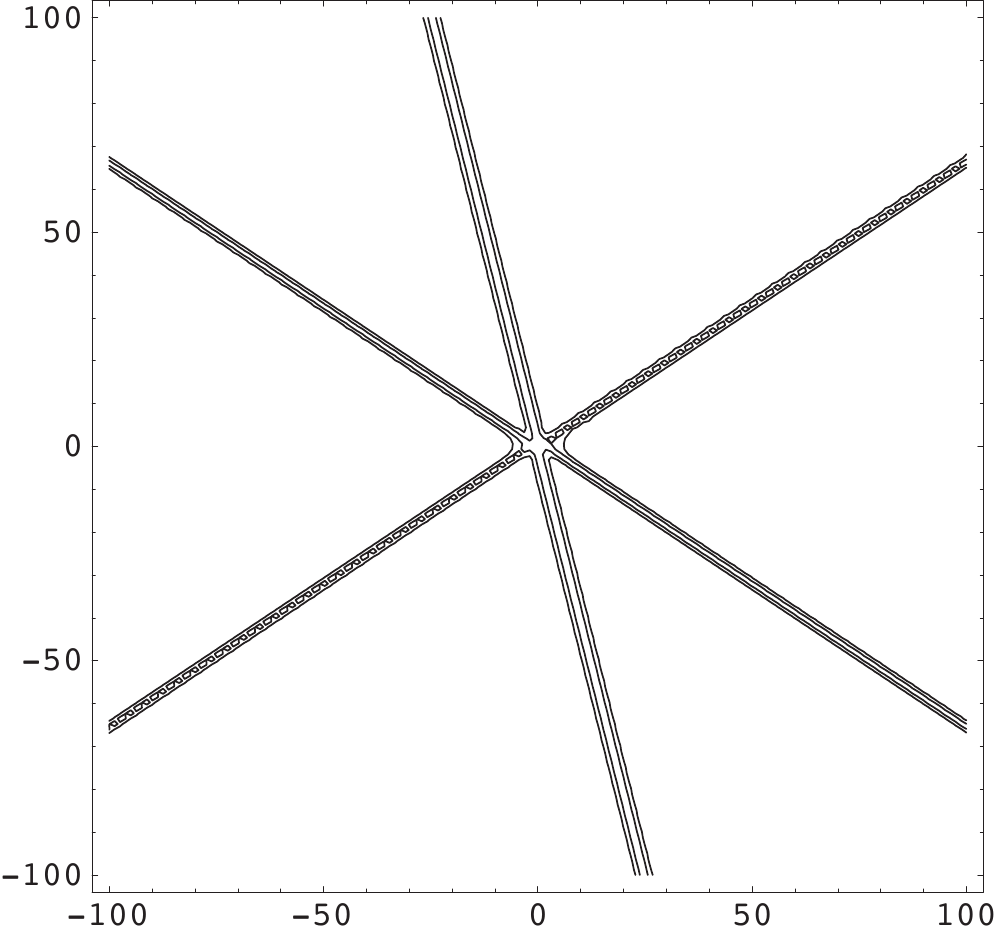}}
\hskip0.3cm\includegraphics[width=3.5cm]{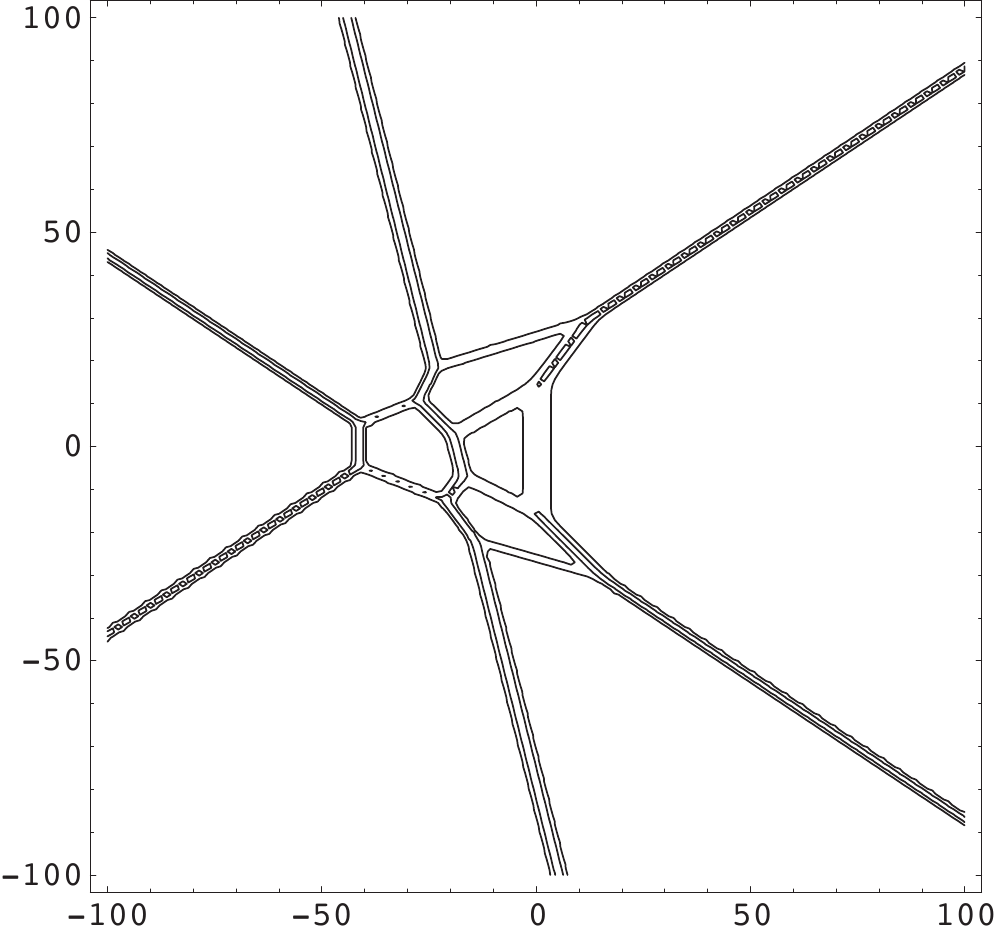} 
\hskip0.3cm{\includegraphics[width=3.5cm]{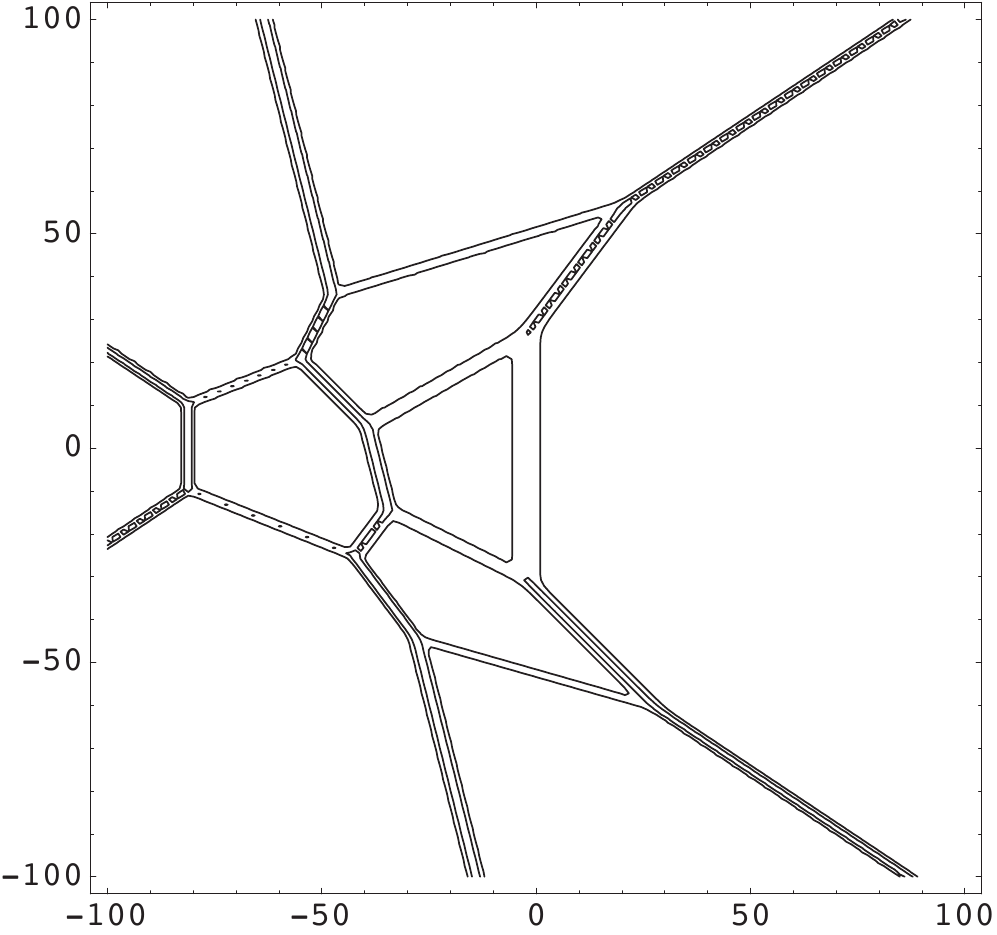}}
\end{center}
\vskip-0.4cm
\caption{The contour plots corresponding to a KP soliton for $\text{Gr}(3,6)_{>0}$.
 The panels show the time evolution of the solution $u(x,y,t)$ for $t=-10, 0, 10, 20$ from the left.}
\label{fig:Sol36}
\end{figure}

Each region in a contour plot represents the portion of the $xy$-plane where one of the exponential terms $\Delta_I(A)e^{\Theta_I}$ in the $\tau$-function \eqref{e:tau}
is \emph{dominant} over the others.
Hence to characterize the contour plot, we must determine which exponential term $\Delta_I(A)e^{\Theta_I}$
dominates at each point in the $xy$-plane. 
Equivalently, we may ask which of the linear terms
\begin{equation}\label{e:linear terms}
\ln(\Delta_J(A)K_J) + \Theta_J(x,y,\mathbf{t}) \quad\text{for}\quad J\in\binom{[M]}{N}
\end{equation}
dominates at each point.

Since the KP equation is nonlinear, arbitrary contour plots can be rather complicated.  To make the problem more tractable, we focus on the asymptotic behavior of these plots for large values of the variables.  We rescale the variables $(x,y,\mathbf{t})$, so that the constant terms $\ln(\Delta_J(A)K_J)$ are negligible. More precisely, we perform a change of variables
$x\to x/\epsilon, y\to y/\epsilon$ and $\mathbf{t}\to \mathbf{t}/\epsilon$ with a small positive number $0<\epsilon\ll 1$.  Then the $\tau$-function becomes
\begin{eqnarray*}
&{}&\tau\left(\frac{x}{\epsilon},\frac{y}{\epsilon},\frac{\mathbf{t}}{\epsilon}\right)=\sum_{I\in\mathcal{M}(A)}\exp\left(\frac{1}{{\epsilon}}{{\Theta}}_I(x,y,\mathbf{t})+\ln(K_I\Delta_J(A))\right)
\end{eqnarray*}
where $\mathcal{M}(A)$ is the \emph{matroid} associated to the matrix $A$, defined by
\[
\mathcal{M}(A):=\left\{I\in\binom{[M]}{N}:\Delta_I(A)>0\right\}.
\]

Then we define a piecewise linear function which is given by the \emph{tropical limit}
\begin{equation}\label{e:tropical-f}
f_{\mathcal{M}(A)}(x,y,\mathbf{t}):=\lim_{\epsilon\to 0}\left(\epsilon\ln\tau\right)=\underset{I\in\mathcal{M}(A)}{\text{max}}\left\{\Theta_I(x,y,\mathbf{t})\right\}.
\end{equation}
That is, $f_{\mathcal{M}(A)}(x,y,\mathbf{t})$ represents a dominant plane $z=\Theta_I(x,y,\mathbf{t})$ in $\R^3$  for fixed $\mathbf{t}$.  We define the \emph{soliton graph} for each $\mathbf{t}$
by
\[
\mathcal{C}_\mathbf{t}(\mathcal{M}(A)):=\{\text{the locus of the $xy$-plane where $f_{\mathcal{M}(A)}(x,y,\mathbf{t})$
is \emph{not} linear}\}.
\]
The soliton graph is hence a collection of bounded and unbounded line segments, which we call \emph{line solitons}.
Note that  each region of the complement of 
$\mathcal{C}_t(\mathcal{M}(A))$ 
is a domain of linearity for $f_{\mathcal{M}(A)}(x,y,\mathbf{t})$,  hence each region is  associated to a {\it dominant plane}  $z=\Theta_I(x,y,\mathbf{t})$ for a certain $I\in\mathcal{M}(A)$. We label this region $\Theta_I$ or simply $I$. 

Suppose a line-soliton separates two regions, labeled $I$ and $J$.  
Then we have
\begin{equation}\label{e:IJindices}
J = I \backslash \{i\} \cup \{j\}\quad\text{for some}\quad i,j\in[M],
\end{equation}
that is, their labels differ only by a single index for a generic choice of the $\k$-parameters,
i.e. $\k_i+\k_j\ne\k_n+\k_m$ if $\{i,j\}\ne\{n,m\}$ (see \cite{CK:09}).
  As in Example \ref{Example:N = 1}, 
we call this segment an $[i,j]$-soliton for $i<j$; if we do not wish to specify that $i < j$, we use the notation $\{i,j\}$-soliton instead.

Our goal is to understand the combinatorial structure of soliton graphs; that is, we want to classify the possible sets of region labels, and the adjacency relations among the regions.  Hence we may forget the original $xy$-coordinates, and represent a soliton graph as an abstract network with labeled faces. Edges represent line solitons, and vertices represent points where multiple solitons meet in a common endpoint.  (If multiple solitons cross at a point which is \emph{not} a common endpoint, we do not consider that a vertex.) We review some facts about the resulting networks, due to \cite{KW:13, kodama} (see also \cite{K:17} for a survey of these results).

Generically, a point where several solitons share an endpoint will have degree three.  So a generic soliton graph is a trivalent network, with regions labeled by elements of $\binom{[M]}{N}$ for some $N < M$. Let $I,J,L \in \binom{[M]}{N}$ be the labels of three regions which meet at a trivalent vertex $\v$ of a soliton graph.
Recall \eqref{e:IJindices}, that is, whenever two regions of the $xy$-plane are separated by a line soliton, their labels differ by a single index.  
Hence there are two possibilities for the labels $I,J$ and $L$:
\begin{enumerate}
\item 
$I=I_0\cup\{i\}, J=I_0\cup\{j\}$ and $L=I_0\cup\{l\}$ for some common $(N-1)$-index set $I_0$
\item 
 $I=K_0\setminus\{i\}, J=K_0\setminus\{j\}$ and $L=K_0\setminus\{l\}$ for some common $(N+1)$-index set $K_0$. 
\end{enumerate}
We color the vertex $\v$ white in the first case, and black in the second.  See Figure \ref{fig:Gr12-3} for an example.

In the previous works \cite{CK:08,CK:09, KW:13, kodama},
it was shown that the KP soliton \eqref{e:u-tau} with
the $\tau$-function \eqref{e:KPtau} consists of $N$ line-solitons 
as $y \gg 0$ and $M-N$ line-solitons as $y \ll 0$. 
Each of those asymptotic solitons is uniquely parametrized by a map 
$\pi$ such that $\pi(i)=j$ if the $[i,j]$-soliton appears
at $y \gg 0$, and $\pi(j)=i$ if the $[i,j]$-soliton appears
at $y \ll 0$. The map $\pi$ is well-defined, and is in fact a fixed-point
free permutation or {\it derangement} of the index set $\{1,\ldots,M\}$.
Moreover, the derangement 
is completely determined by the \emph{matroid} $\mathcal{M}(A)$ of the totally nonnegative matrix $A$, and vice versa (see \cite{K:17} for a survey of these results.)

A totally nonnegative matrix $A$ is \emph{totally positive} if $\mathcal{M}(A) = \binom{[M]}{N}$.  The corresponding derangement is given by $i \mapsto i-N$, where all values are taken modulo $M$. The space of totally positive matrices, modulo row operations, is the \emph{totally positive Grassmannian} $\Gr(N,M)_{>0}$.  
Soliton graphs for $\Gr(N,M)_{>0}$ have nice combinatorial properties, which make them easier to classify.  See Section \ref{subsec: weaksep} for details.
In what follows, we restrict our attention to soliton graphs for $\Gr(N,M)_{>0}$.

%%%%%%%%%%%%%%%%%%%%%%%%%%

\subsection{Duality and soliton subdivisions}
\label{subsec: duality}

In order to study the soliton graphs for $\Gr(N,M)_{>0}$, we first define a bijection,
called the \emph{duality map}, which maps a plane in $\mathbb{R}^3$ to a point in $\mathbb{R}^3$, 
\begin{equation}\label{e:Dmap}
\mu: (p,q, \omega) ~\longleftrightarrow ~ \{(x,y,z): z=px+qy+\omega\},
\end{equation}
where $p=\k_i, q=\k_i^2$ and $\omega$ is some constant (we may take $\omega=\k_i^3t$ for a KP soliton).
The vector $\langle p,q,-1 \rangle$ is the normal vector of the plane,
and the vector $\langle p,q \rangle$ gives the increasing direction of the plane, i.e.
$\nabla z=\langle p,q \rangle$.   See Figure \ref{fig:Duality-line}. Using the map, we can classify the soliton graphs $\mathcal{C}_\mathbf{t}(\mathcal{M}(A))$ via the \emph{triangulations} of a polygon inscribed in a parabola as described below.
%%%%%%%%%%%%%%%%%%%%%%%%%%%
\begin{figure}[h]
\begin{center}
\includegraphics[height=3.5cm]{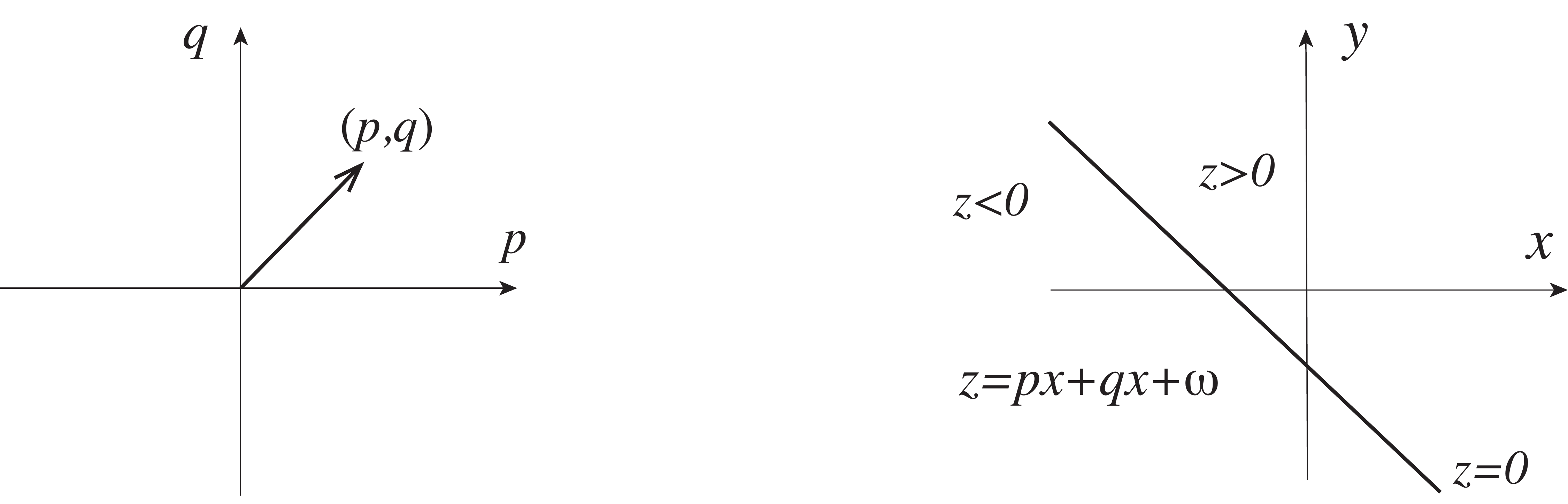}
\end{center}
\caption{Duality map. The vector $\langle p,q,-1 \rangle$ is the normal vector of the plane
$z = px + qy + \omega$ and the vector $\nabla z=\langle p,q \rangle$ gives the increasing direction of $z$.}\label{fig:Duality-line}
\end{figure}
%%%%%%%%%%%%%%%%%%%%%%%%%

As a simplest example, consider the case with three points $\hat\p_i=(p_i,q_i,\omega_i)\in\R^3$ with
$p_i=\kappa_i, q_i=\kappa_i^2$ and $\omega_i=\kappa_i^3t$ for $i=1,2,3$.  Then we have 
a triangle inscribed in the parabola $q=p^2$ whose vertices are $\{\p_1,\p_2,\p_3\}$,
and each vertex $\p_i=(p_i,q_i)$ has a \emph{weight} $\omega_i$.  Again for simplicity, take all $\omega_i=0$.
Then all the planes 
\[z=\theta_i(x,y)=p_ix+q_iy\]
intersect at the origin, and at each point $(x,y)$,
one of the planes becomes dominant.  Figure \ref{fig:Gr12-3} below shows the \emph{duality} between the triangles in the $pq$-plane and the soliton graphs in the $xy$-plane at $t=0$. The dynamics of the intersection point are linear in time $t$ as given by $\theta_1=\theta_2=\theta_3$.
%%%%%%%%%%%%%%%%%%%%%%
\begin{figure}[h]
\begin{center}
\includegraphics[height=3.5cm]{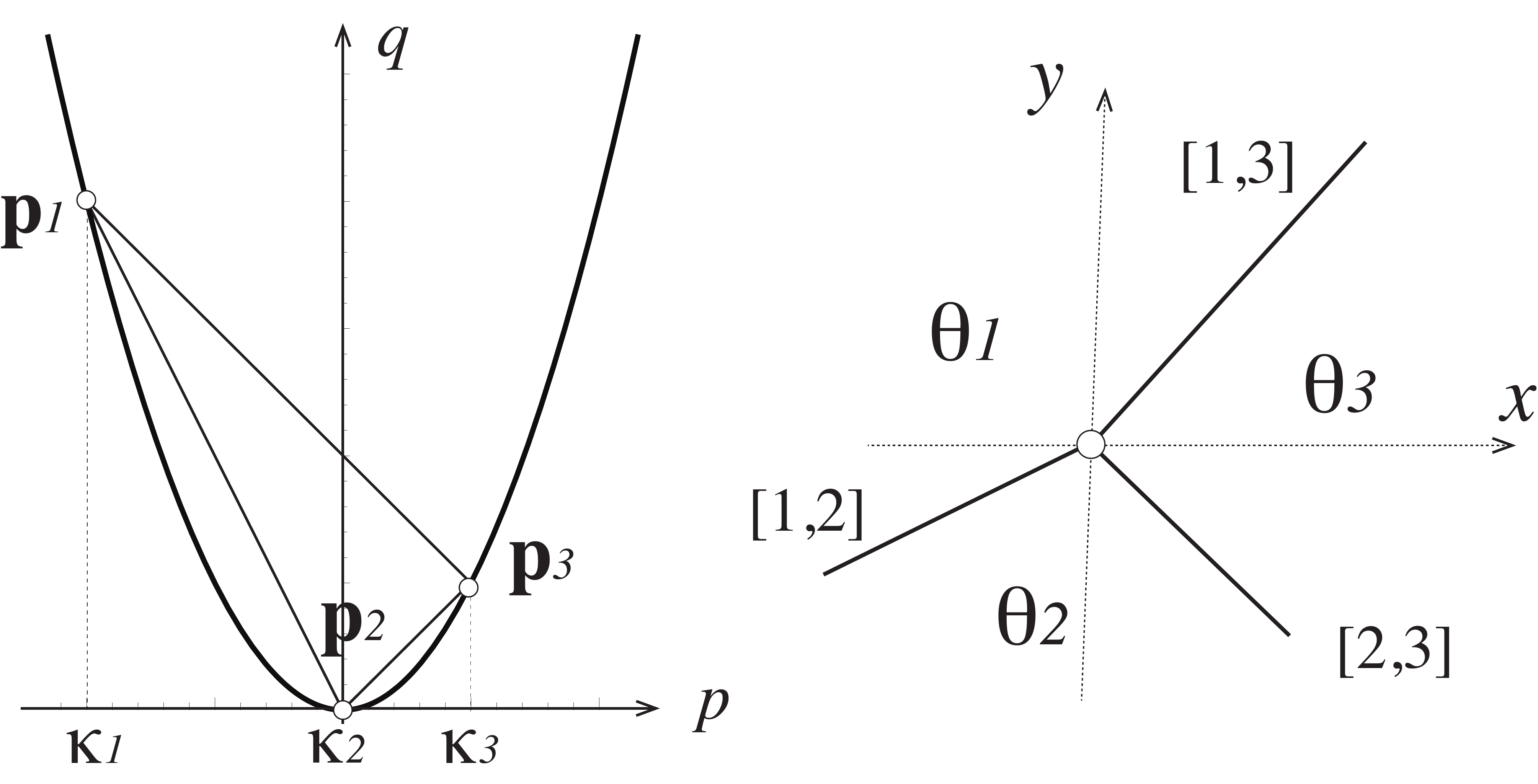}\hskip1.5cm
\includegraphics[height=3.5cm]{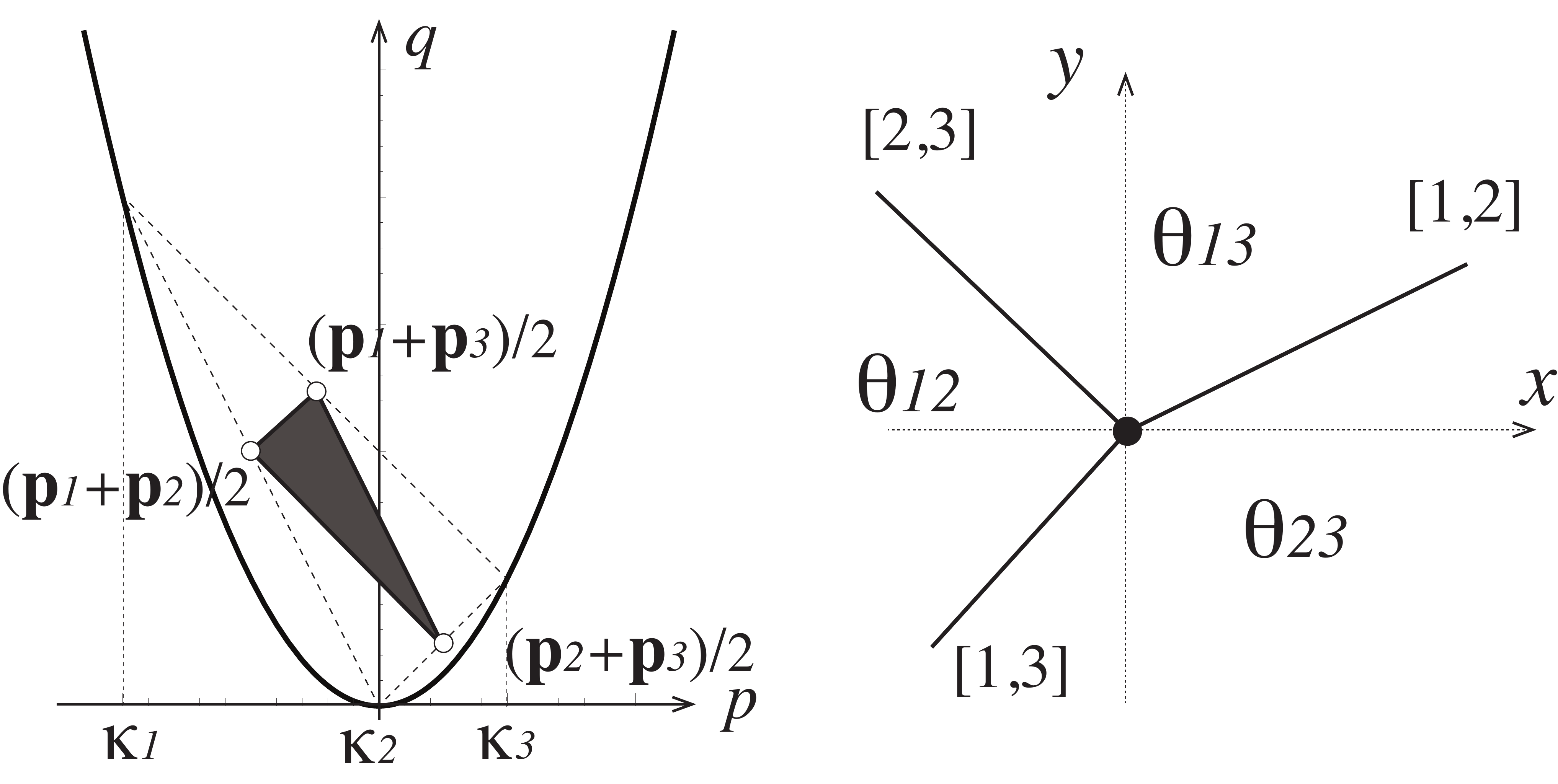}
\end{center}
\caption{Duality between the triangles prescribed in a parabola in the $pq$-plane and the soliton graphs in the $xy$-plane.
Trivalent vertices in the soliton graphs are colored white for
 $\Gr(1,3)_{>0}$ (left) and black for $\Gr(2,3)_{>0}$ (right).}\label{fig:Gr12-3}
\end{figure}
%%%%%%%%%%%%%%%%
The left two panels show the case for $\text{Gr}(1,3)_{>0}$, that  is, 
 we have $\mathcal{M}(A)=\{1,2,3\}$ and
\[
f_{\mathcal{M}(A)}(x,y,0)=\text{max}\{\theta_i(x,y,0):i=1,2,3\}.
\]
In the soliton graph (the second figure from left), each region is labeled by the dominant plane
$z=\theta_i(x,y,0)$.  Since the trivalent vertex in the soliton graph is colored white, 
we define the triangle inscribed in the parabola as a \emph{white} triangle.  
Notice that each edge of the triangle, say $\p_{[i,j]}=\p_i-\p_j$, is perpendicular
to the line given by $\theta_i=\theta_j$ which corresponds to the $[i,j]$-soliton. 

The right two panels show the case for $\text{Gr}(2,3)_{>0}$, that is, 
we have $\mathcal{M}(A)=\{12,13,23\}$ and
\[
f_{\mathcal{M}(A)}(x,y,0)=\text{max}\{(\theta_i+\theta_j)(x,y,0):1\le i<j\le 3\},
\]
This triangle is defined as a \emph{black} triangle, which is dual to the black vertex in
the soliton graph.  In the figure, the black triangle is the convex hull of the vertices $\{\frac{1}{2}(\p_i+\p_j):1\le i<j\le 3\}$, i.e. the vertices are the mid points of the edges of the white triangle in the left figure.

In general,  the soliton graph for $\text{Gr}(N,M)_{>0}$ has only trivalent vertices
which are colored either white or black \cite{kodama}.  Hence for a generic choice of weights, the 
image of the duality map for the soliton graph is a triangulation with colored triangles.  We will only consider the soliton graphs for $\text{Gr}(N,M)_{>0}$, and their corresponding triangulations.

In the case of Gr$(2,4)_{>0}$, we have
\[
f_{\mathcal{M}(A)}(x,y,t)=\text{max}\{(\theta_i+\theta_j)(x,y,t):1\le i<j\le 4\},
\]
Figure \ref{fig:Gr24} illustrates the soliton graphs for $t<0$ (left),  and  for $t>0$ (right).  
%%%%%%%%%%%%%%%%%%%%%%%%%%%
\begin{figure}[h]
\begin{center}
\includegraphics[height=3.5cm]{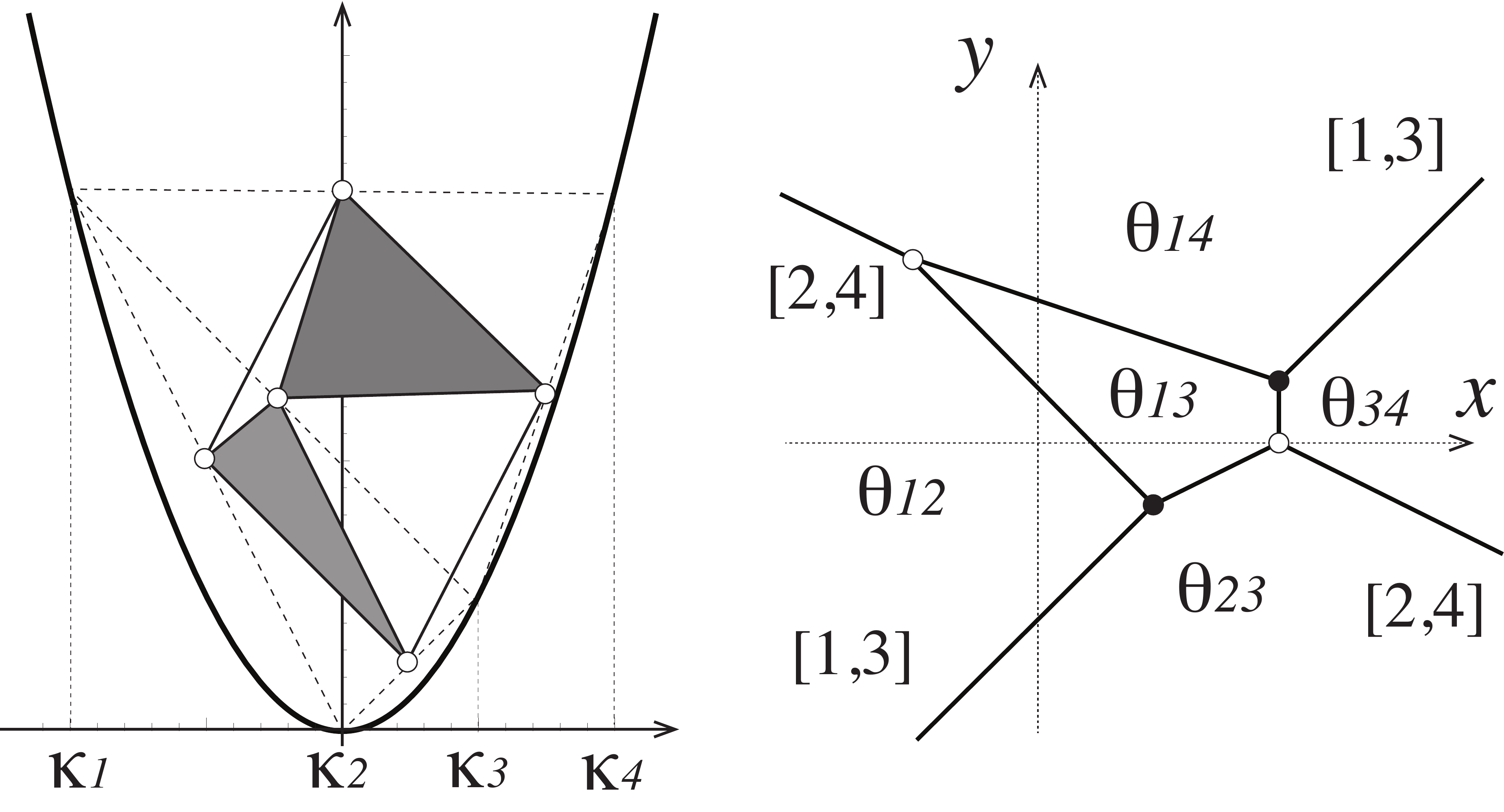}\hskip1.5cm
\includegraphics[height=3.5cm]{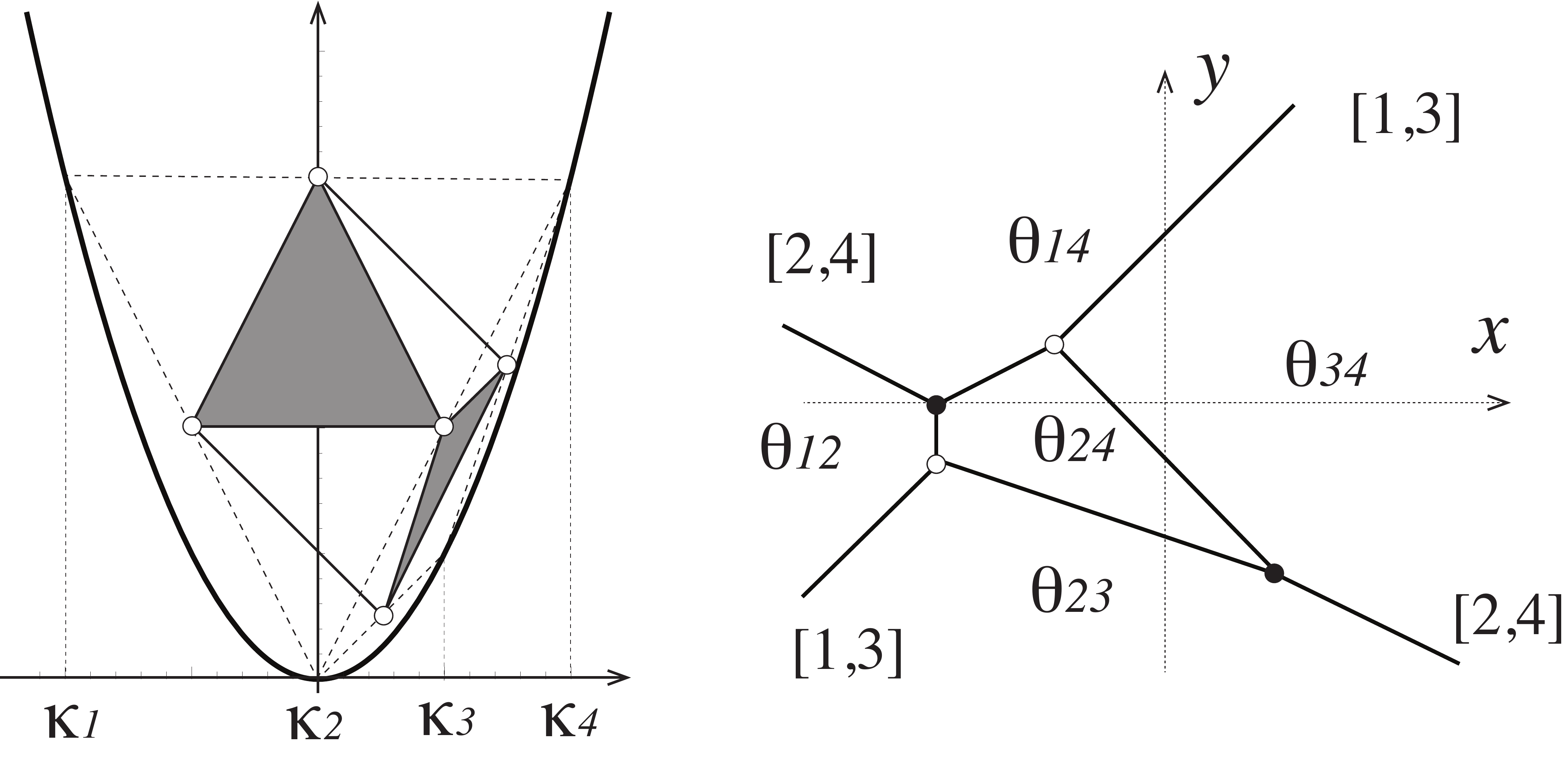}
\end{center}
\caption{Triangulations and the soliton graphs for $\Gr(2,4)$: Black-white flip.
The left two figures show the triangulation of the point set $\frac{1}{2}\{\p_{12},\p_{23},\p_{34},\p_{14},\p_{13},\p_{24}\}$ (each point is shown as an open circle) and the corresponding soliton graph for $t<0$.  The right two figures are for $t>0$.
The $\kappa$-parameters are $(-2,0,1,2)$.}\label{fig:Gr24}
\end{figure}
%%%%%%%%%%%%%%%%%%%%%%%%%
In the figures, the change of the graphs can be considered as a flip in the triangulation of the quadrilateral given by
the convex hull of the set of six points $\{\p_i+\p_j:  1\le i<j \le 4\}$.
The flip corresponds to the \emph{mutation} in the cluster algebra structure on $\Gr(N,M)$ \cite{fomin, FZ:03, scott, KW:11, kodama}, and  we call  the flip ``black-white flip'' (i.e. the colors of the vertices exchange).
Note that the quadrilateral (parallelogram) in the figures are given by the convex hull of $\{\frac{1}{2}(\p_i+\p_j):1\le i<j\le 4\}$.

%%%%%%%%%%%%%%%%%%%%%%%%%%%%%

\subsubsection{Definitions and notation}
We now give some definitions and notations that we use in the rest of this paper. 
%Recall that  $[M] : = \{ 1,2,\cdots, M\},$ and ${[M] \choose N}$ is the collection of all the $N$-element subsets of $[M]$.  
Let  $I=i_1\cdots i_N$ denote the $N$-element subset $I=\{i_1,\ldots,i_N\}\in\binom{[M]}{N}$, and $Ii_{N+1}$ denote the $(N+1)$-element subset $I\cup\{i_{N+1}\}$.  Also let $I\backslash i_k$ denote the $(N-1)$-element subset $I\setminus\{i_k\}$ for $k\in[N]$.

For $A \in \Gr(N,M)_{>0}$, we have $\mathcal{M}(A) = {[M] \choose N}$.  We denote the corresponding point 
configuration by
\[
\mathcal{A}_{N,M}:=\left\{\p_I=(p_I,q_I):  I\in \binom{[M]}{N}\right\}, 
\]
where $p_I=\sum_{i\in I}\k_i, q_I=\sum_{i\in I}\k^2_i$ with the order $\k_1<\k_2<\cdots<\k_M$.
Also note that  the convex hull  
\[{\sf P}_{N,M} := {\sf P}_{N,M}^0 := \text{conv} (\mathcal{A}_{N,M})\]
is an $M$-gon.  This follows by considering the behavior of KP solitons for $|y| \gg 0$ and applying the duality map (see \cite{CK:09, kodama} for the asymptotic behavior of the KP solitons).

Given a weight vector $\omega=(\omega_1,\cdots,\omega_M)$, we assign each point $\p_i$ a weight $\omega_i$, and write $\hat \p_i= (\p_i,\omega_i)$.  For $I\in\binom{[M]}{N}$,
we have the weighted point $\hat\p_I=(\p_I,\omega_I)$ where $\p_I = \sum_{i \in I} \p_i$ and $\omega_I=\sum_{k \in I}\omega_{k}$. Then we consider the weighted (or lifted) point configuration
\[
\mathcal{A}^\omega_{N,M}:=\left\{\hat\p_I=(p_I,q_I,\omega_I)\in \R^3:I\in \binom{[M]}{N}\right\},
\]
and the convex hull of the lifted point configuration
\[
{\sf P}^\omega_{N,M} = \text{conv} (\mathcal{A}^\omega_{N,M}),
\]
which is a three-dimensional convex polytope.  Note here that ${\sf P}^0_{N,M}$ is an $M$-gon in the $pq$-plane. The  vertices of the $M$-gon are given by
\[
\left\{\p_{I_j}: I_j=\{j,j+1,\ldots,j+N-1\},~j=1,\ldots,M~(\text{cyclic order})\right\}
\]
For example, in the case of Gr$(2,6)_{>0}$, the convex hull of 15 points $\{\p_{i,j}=\p_i+\p_j:1\le i<j\le 6\}$ is a \emph{hexagon} with the vertices 
\[
\{\p_{1,2}, ~\p_{2,3},~\p_{3,4},~\p_{4,5},~\p_{5,6},~\p_{1,6}\}.
\]
Note here that all other $9$ points $\p_{i,j}$ with $|i-j|>1$ (cyclic sense) are \emph{inner} points of the hexagon.  

We also define the $m$-faces of the polygon ${\sf P}^{\omega}_{N,M}$ with $m$ being the dimension of the face, and in our case, $m=0, 1$ or $2$.
\begin{definition}\index{soliton subdivision}
A nonempty set ${\sf S}\subset {\sf P}^{\omega}_{N,M}$ is an \emph{upper $m$-face} of the convex polytope ${\sf P}^{\omega}_{N,M}$, if there exists a plane $\mathcal{P}:=\{z=ax+by+c\}$
such that 
\begin{itemize}
\item[(a)] ${\sf S}={\sf P}^\omega_{N,M}\cap \mathcal{P}$ with $\text{dim}\,{\sf S}=m$ and
\item[(b)] any point in the region $z > ax+by+c$ has no intersection with ${\sf P}^{\omega}_{N,M}$.
\end{itemize}
Similarly, a nonempty set ${\sf S}\subset {\sf P}^{\omega}_{N,M}$ is \emph{a lower $m$-face} of ${\sf P}^{\omega}_{N,M}$, if there exists a plane $\mathcal{P}:=\{z=ax+by+c\}$
with property (a), and the region $z < ax+by+c$ has no intersection with  ${\sf P}^{\omega}_{N,M}$.
\end{definition}

 Then projecting the \emph{upper hull}, the collections of all upper faces,
of ${\sf P}^{\omega}_{N,M}$ back on to the $pq$-plane
induces a \emph{regular subdivision} of the polygon ${\sf P}^0_{N,M}$ in $\R^2$.
Here the notions of \emph{subdivision} and \emph{regular} (or coherent) are defined in general as follows (see e.g. \cite{DRS:10, T:06}):
\begin{definition}\label{d:subdivision}
A set $\Q$ is a subdivision of the $M$-gon ${\sf P}^0$, if there are sets of indices
$\{\sigma_1,\ldots,\sigma_m\}$ with $\sigma_i\subset \binom{[M]}{N}$ such that
${\sf P}_{\sigma_i}:={\rm conv}\{\p_j: j\in\sigma_i\}$ satisfy
\begin{itemize}
\item[(i)] ${\sf P}_{\sigma_i}$ is a $k$-gon with $k\ge 3$,
\item[(ii)] $\Q=\bigcup_{i=1}^m{\sf P}_{\sigma_i}$,
\item[(iii)] ${\sf P}_{\sigma_i}\cap{\sf P}_{\sigma_j}$ is either empty or a common edge of those
polygons.
\end{itemize}
In particular, if all ${\sf P}_{\sigma_i}$ are triangles, then the subdivision is called a \emph{triangulation}.
We also say that a subdivision ${\sf Q}$ is \emph{regular}, if it is obtained by the projection of the upper hull of a
polytope ${\sf P}^\omega={\rm conv}\{(\p_I,\omega_I)\in\R^3\}$ for some weight $\omega$.
\end{definition}

We then define a \emph{soliton subdivision} to be a regular subdivision, denoted by $\Q_{N,M}(\omega)$, which is given by
the projection of the upper hull of ${\sf P}^\omega_{N,M}$, 
where each polygon $\P_{\sigma_i}$ in the subdivision is the projection of an upper face of $\P^\omega_{N,M}$.  We sometimes refer to 
``a regular subdivision (or regular triangulation) $\Q_{N,M}(\omega(\t))$ of the polygon $\P_{N,M}$ associated with the weight function $\omega(\t)$'' as simply  ``subdivision (or triangulation) $\Q_{N,M}$ of $\A_{N,M}^\omega$''.

For a polygon in $\Q_{N,M}(\omega)$,  its vertices are given by the set $\{\p_{I_1}, \cdots, \p_{I_k} \}$ when the polygon is a $k$-gon.
Each vertex $\p_I$ can be represented by its index set $I\in\binom{[M]}{N}$, and we may denote the polygon $\text{conv}\{\p_{I_1}, \cdots, \p_{I_k} \}$ by $\{\p_{I_1}, \cdots, \p_{I_k} \}$, or simply its index set $\{I_1, I_2, \cdots, I_k\}$ for short. For the corresponding face of the polytope in ${\sf P}^\omega_{N,M}$, we sometimes denote it as $\{\hat\p_{I_1}, \hat\p_{I_2}, \cdots, \hat\p_{I_k}\}$ or
$\{\hat{I_1}, \hat{I_2}, \cdots, \hat{I_N}\}$. 

We also define the following notions for the polygons appearing in the subdivision $\Q_{N,M}(\omega)$, which is the generalization of \emph{white}-\emph{black} triangles:
\begin{definition}\label{def:black-white}
For a convex polygon ($k$-gon) in the subdivision $\Q_{N,M}(\omega)$, we say that 
\begin{itemize}
\item[{\rm (a)}] the polygon is \emph{white} if the vertices of the polygon  are expressed by
 \[
 \{Ii_1,Ii_2,\cdots, Ii_k\}\qquad\text{for some}\quad I \in {[M] \choose N-1},
 \]
 and 
 \item[{\rm (b)}] the polygon is \emph{black} if the vertices are expressed by 
 \[
 \{J\back i_1,J\back i_2,\cdots,J \back i_k \}\qquad\text{for some}\quad J \in { [M] \choose N+1}. 
 \]
 \end{itemize}
 Since the index sets of two adjacent points differ only by a single index \cite{CK:09},  there are only these types of polygons in the subdivisions (recall that each edge in the subdivision corresponds to a line-soliton).
 \end{definition} 

We are interested in using soliton triangulations to study the \emph{combinatorial} structure of soliton graphs, forgetting the $xy$ coordinates.  Hence, we may forget the $pq$-coordinates of a soliton triangulation, and remember only the adjacency relations between the tiles.  For convenience, we often draw the vertices $\{\pp_{I_1},\ldots,\pp_{I_M}\}$ of the convex $M$-gon ${\sf P}_{N,M}$ as points on a circle, rather than on a parabola.
%%%%%%%%%%%%%%%%%%%%%%%%%%%%%%%

\subsection{Plabic graph, weakly separated collections and realizability}
\label{subsec: weaksep}

The main objective of this paper is to classify soliton graphs for $\Gr(N,M)_{>0}$. By results of \cite{kodama}, these graphs are planar and trivalent.
For convenience, we may embed a soliton graph in a bounding disk whose interior contains all vertices of the graph.  We place a \emph{boundary vertex} at the point where each $\{i,\pi(i)\}$-soliton intersects the disk, and label the boundary vertex $\pi(i)$. We forget the metric structure on the graph, and maintain only the combinatorial structure. 
As in the previous section, we color each internal vertex black or white, depending on the labels of the surrounding faces.  See figure \ref{plabic} for an example.

\begin{figure}[ht]
\centering
\includegraphics[width=5cm]{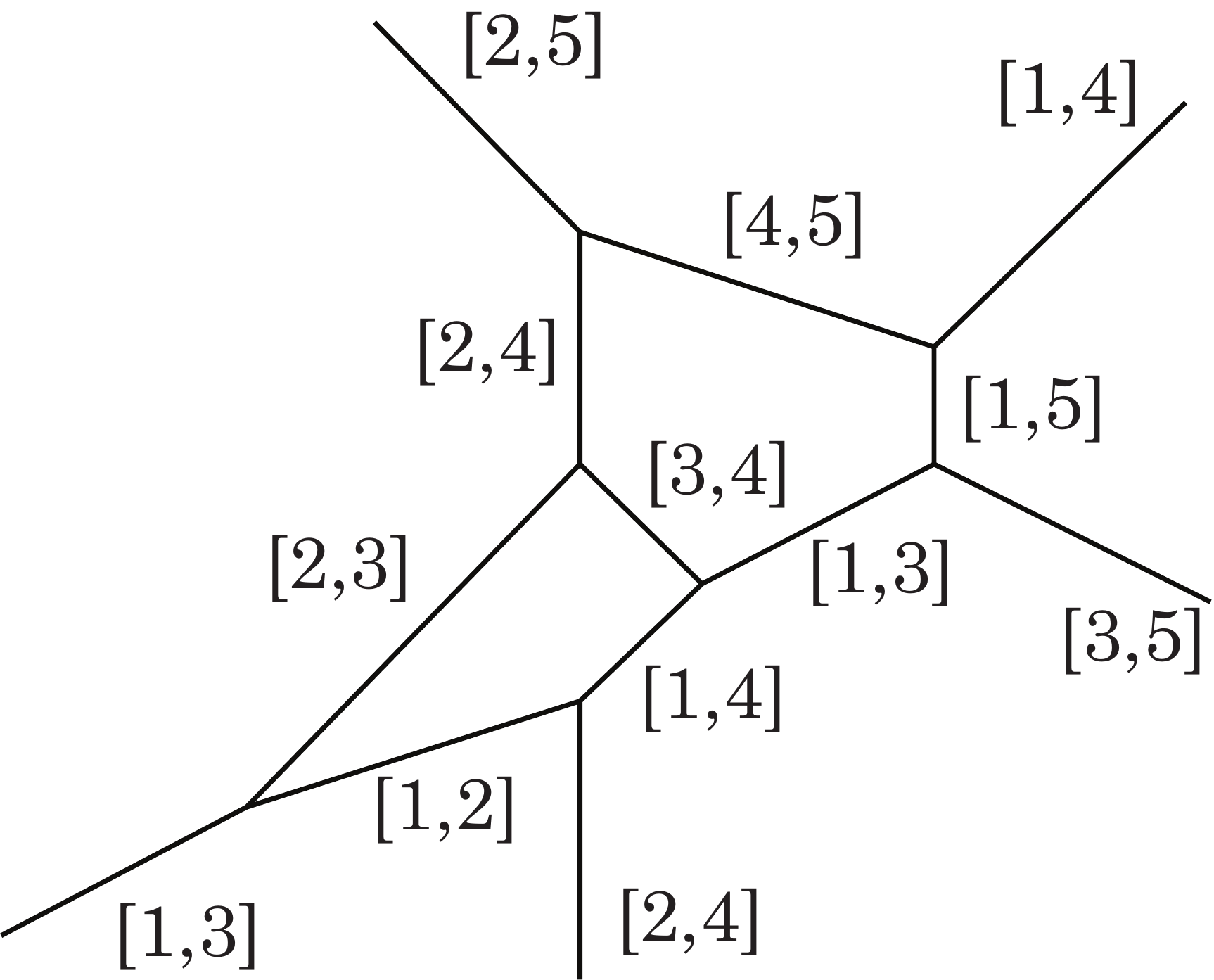}\hskip1cm
\includegraphics[trim = {2.5in 8.3in 2.3in 1in}, clip]{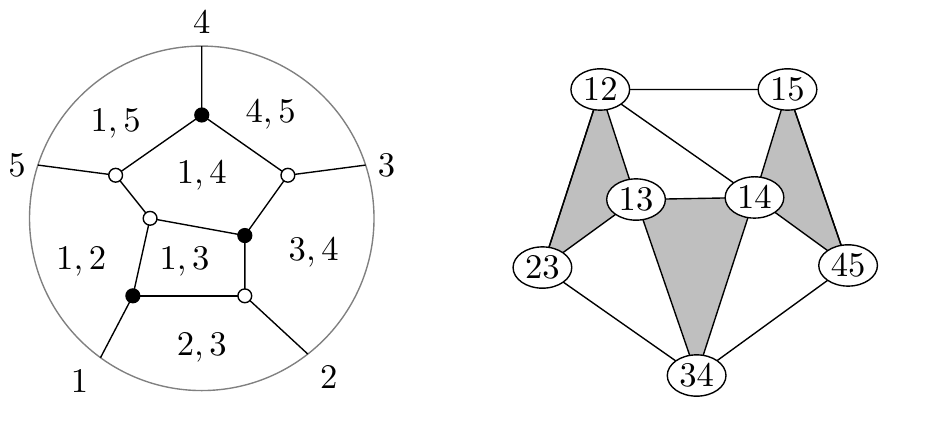}
\caption{A soliton graph for $\Gr(2,5)_{>0}$, and corresponding plabic graph and triangulation.}
\label{plabic}
\end{figure} 

With these conventions, every soliton graph for $\Gr(N,M)_{>0}$ is combinatorially (or topologically) equivalent to a \emph{reduced plabic graph} \cite{kodama}. First introduced by Postnikov, reduced plabic graphs play a key role in the combinatorial theory of $\Gr(M,N)_{>0}$ \cite{postnikov}.  We give a precise characterization of plabic graphs in Section \ref{subsec: combo background}.  For now, it suffices to remark that a plabic graph is a planar, bicolored network which satisfies some technical conditions; and whose faces have a natural labeling by elements of $\binom{[M]}{N}$ for some $N < M$.

Each plabic graph has an associated permutation $\pi$.  For soliton graphs, $\pi$ is the derangement defined by the soliton asymptotics \cite{CK:09, kodama}.  We say $G$ is a plabic graph for the \emph{totally positive Grassmannian} $\Gr(N,M)_{>0}$ if $\pi$ is the permutation corresponding to $\Gr(N,M)_{>0}$; that is, if $\pi$ is defined by $i \mapsto i - N$ with indices taken modulo $M$. Plabic graphs for $\Gr(M,N)_{>0}$ have an easy classification in terms of \emph{weakly separated collections}, as explained below.
    
For $G$ a reduced plabic graph, contracting an edge whose endpoints are vertices of the same color gives a reduced plabic graph $G'$ with the same face labels, and this operation is reversible.  We say that two plabic graphs are \emph{contraction equivalent} if we can transform one into another by contracting and un-contracting unicolored edges. Reduced plabic graphs, up to contraction equivalence, are determined uniquely by their face labels.  The possible collections of face labels can be easily classified, using the notion of \emph{weak separation} defined in \cite{leclerc}.

\begin{definition}
For $I, J \subseteq \binom{[M]}{N}$, we say $I$ and $J$ are \emph{weakly separated} if there \emph{do not} exist $a,b \in I \backslash J$ and $c,d \in J \backslash I$ such that if $M$ points $1, 2, \ldots, M$ are arranged counter-clockwise around a circle, the points $a, c, b$ and $d$ occur in cyclic order.  
\end{definition}
 
\begin{definition}
A \emph{weakly separated collection} is a collection of elements of $\binom{[M]}{N}$ whose members are pairwise weakly separated.  A weakly separated collection is \emph{maximal} if it is maximal by inclusion.
\end{definition}

\begin{theorem}\cite{oh}
A collection of elements of $\binom{[M]}{N}$ is the set of face labels of a plabic graph for $\Gr(N,M)_{>0}$ if and only it is a maximal weakly separated collection.
\end{theorem}

In \cite{oh}, the authors introduce planar diagrams called \emph{plabic tilings}, which correspond to weakly separated collections. We refer to \cite[Section 9]{oh} for the precise definition.  For our purposes, it suffices to describe \emph{triangulated} plabic tilings as the \emph{duals} of trivalent reduced plabic graphs.  That is, we can obtain a triangulated plabic tiling from a plabic graph by applying a purely combinatorial analogue of the duality map from Section \ref{subsec: duality}.  Deleting edges that separate triangles of the same color, we obtain a subdivision of the $M$-gon into black and white polygons, which we call a \emph{plabic tiling}.  Since soliton graphs for $\Gr(N,M)_{>0}$ are plabic graphs \cite{kodama}, soliton triangulations are  triangulated plabic tilings
(via the duality map).  
%In this paper, we often refer to triangulated plabic tiling  as simply \emph{triangulation} (or soliton triangulation), and plabic tiling as \emph{subdivision} (or soliton subdivision).

\begin{definition}
A plabic graph for $\Gr(N,M)_{>0}$ is \emph{realizable} if it is a soliton graph, up to contraction equivalence; a triangulated plabic tiling is realizable if it is a soliton triangulation.  A weakly separated collection for $\Gr(N,M)_{>0}$ is realizable if it is the set of face labels of a soliton graph, or equivalently, the set of vertex labels of a soliton subdivision.
\end{definition}

Kodama and Williams showed that every plabic graph for $\Gr(2,M)_{>0}$ is realizable, up to contraction equivalence \cite[Theorem 12.1]{kodama}.  In the language of tilings, their result says that every weakly separated collection for $\Gr(2,M)_{>0}$ is realizable.  We recover this result below, as a consequence of Algorithm \ref{InductiveAlgorithm}.  See Section \ref{subsec: zonotope} for details. 

In his PhD thesis, Huang showed that every weakly separated collection (or plabic tiling) for $\Gr(3,6)_{>0}$ is realizable \cite{H:15}. However, some collections are only realizable for certain choices of $\kappa$-parameters.  Huang then conjectured that \emph{every} weakly separated collection for any $\Gr(N,M)_{>0}$ is realizable for some choice of parameters, a conjecture we disprove in Section \ref{sec: non-realizable}.

\subsection{Summary of results}
\label{subsec: sum}

The structure of the rest of the paper is as follows.
In Section \ref{sec: constructing}, we describe an inductive algorithm from \cite{H:15} for constructing soliton subdivisions, which will be used in Sections \ref{sec: polyhedral} and \ref{sec: realizability}.  As a consequence, 
in Section \ref{subsec: zonotope}, we recover Kodama and Williams' classification of soliton graphs for $\Gr(2,M)_{>0}$ \cite{kodama}, by proving that every weakly separated collection for $\Gr(2,M)_{>0}$ is realizable.

In Section \ref{sec: polyhedral}, we construct a \emph{polyhedral fan} in the space of \emph{multi-time parameters} of the KP hierarchy, which can be used to check whether a given subdivision comes from a soliton graph.  In Section \ref{sec: Gr36}, we use the polyhedral fan to classify soliton graphs for $\Gr(3,6)_{>0}$, by showing that every possible soliton subdivision comes from a soliton graph.  In addition, we specify the subdivisions which are realizable for each choice of $\kappa$-parameters in the KP soliton
(Theorem \ref{realizable-checking-NM}). The main results of Sections \ref{sec: polyhedral} and \ref{sec: Gr36} first appeared in \cite{H:15}, but are presented here in greater detail.

In Section \ref{sec: realizability}, we show that every possible soliton subdivision for $\Gr(3,7)_{>0}$ or $\Gr(3,8)_{>0}$ occurs for \emph{some} choice of $\kappa$-parameters. For $\Gr(3,7)_{>0}$, we give a more precise classification in terms of the $\kappa$-parameters, just as we did for $\Gr(3,6)_{>0}$
(Theorem \ref{classify37}); we do not yet have a classification for each choice of the $\kappa$-parameters. 

Finally, in Section \ref{sec: non-realizable}, we give a subdivision that does \emph{not} come from a soliton graph, disproving a conjecture made in \cite{H:15}.
More precisely, we prove Theorem \ref{NotRealizable}, which states for some $\Gr(N,M)_{>0}$, there is a weakly separated collection which is not realizable for \emph{any} choice of the parameters.

%%%%%%%%%%%%%%%%%%%%%%%%%%%%%%%%%
\section{Constructing soliton graphs}
\label{sec: constructing}

In this section, we describe an explicit algorithm to construct soliton triangulations. We give the algorithm for $\Gr(1,M)_{>0}$ in Section \ref{subsec: A1M}, and for general $\Gr(N,M)_{>0}$ in Section \ref{subsec:induction}.  In Section \ref{subsec: zonotope} we present some useful consequences of the algorithm.

\subsection{Triangulations of the point configuration $\A_{1,M}^\omega$}
\label{subsec: A1M}
Let us start with the subdivisions ${\sf Q}_{1,4}^{\omega}$.
The polytope $\P_{1,4}^\omega$ is a tetrahedron, and the subdivision $\Q_{1,4}(\omega)$
given by the projection of $\P_{1,4}^\omega$  depends on the following determinant,
\begin{equation}\label{e:4-gonD}
D_{1,2,3,4} =
\begin{vmatrix}
1  &  p_{1}  & q_{1}  & \omega_{1} \\
1  &  p_{2}  & q_{2}  & \omega_{2} \\
1  &  p_{3}  & q_{3}  & \omega_{3} \\
1  &  p_{4}  & q_{4}  & \omega_{4}
\end{vmatrix}\qquad \text{with}\quad (p_i,q_i)=(\k_i,\k^2_i).
\end{equation}
That is, we have the following lemma, which we call the \emph{quadrilateral checking lemma}:
\begin{lemma}\label{lem:QCL}\index{quadrilateral checking lemma}
The subdivision $\Q_{1,4}(\omega)$ has the diagonal $\{1,3\}$ if the determinant $D_{1,2,3,4} < 0$; and $\Q_{1,4}(\omega)$ has the diagonal $\{2,4\}$ if $D_{1,2,3,4} > 0$.
\end{lemma}

\begin{proof}
Using the vector notation for the points, i.e. $\hat\p_i:=(p_i,q_i,\omega_i)$,
the determinant \eqref{e:4-gonD} is expressed by a triple scalar product:
\[
D_{1,2,3,4}=-[(\hat\p_2-\hat\p_1)\times(\hat\p_4-\hat\p_1)]\cdot (\hat\p_3-\hat\p_1).
\]
Then $D_{1,2,3,4}<0$ implies that the edge $\{\hat 1,\hat 3\}$ given by the vector $\hat\p_3-\hat\p_1$ is in the upper face of the tetrahedron $\P_{1,4}^\omega$.
That is, the diagonal $\{1,3\}$ in the subdivision $\Q_{1,4}(\omega)$ is the projection
of the upper 1-face of $\P_{1,4}^\omega$.  The case $D_{1,2,3,4}>0$ implies that the edge $\{\hat 1,\hat 3\}$ is in the lower 
face of $\P_{1,4}^\omega$, and then the edge $\{2,4\}$ is the diagonal of $\Q_{1,4}(\omega)$.
\end{proof}

\begin{remark}
Since the formula $D_{1,2,3,4}$ is dealing with the relative position of two diagonals in $\mathbb{R}^3$, we may also state that the edge $\{\hat1,\hat3\}$ is vertically above $\{\hat2,\hat4\}$ when $D_{1,2,3,4} <0$. This means, when looking down from above at the intersection point of the diagonals $\{1,3\}$ and $\{2,4\}$ in $\mathbb{R}^2$, $\{\hat1, \hat3\}$ is vertically above $\{\hat2, \hat4\}$ in $\R^3$. 
\end{remark}

Thus, the sign of the determinant $D_{i_1,i_2,i_3,i_4}$  for a quadrilateral with vertices $\{\hat\p_{i_k}:k=1,\ldots,4\}$ determines the triangulation of the point configuration $\mathcal{A}^\omega_{1,4}$.  Repeatedly apply Lemma \ref{lem:QCL} leads to the following algorithm to construct a subdivision
$\Q_{1,M}(\omega)$ for given weights $\omega=(\omega_1,\ldots,\omega_M)$, for arbitrary $M$:

\begin{algorithm}[Soliton Triangulation for the point configuration $\A_{1,M}^\omega$] 
\label{algorithmGr1M} 

%\smallskip
\noindent
\begin{itemize}
	\item[(1)] Starting with the triangle $\{1,2,3\}$, we add next vertex $\p_4$.  Then 
the original boundary edge $\{1, 3\}$ becomes an internal edge of the 4-gon $\{1,2,3,4\}$,
and we use Lemma \ref{lem:QCL} to check whether or not the edge $\{\hat{1},\hat{3}\}$ is an upper face of the tetrahedron $\PP_{1,4}^\omega$.  If $D_{1,2,3,4} < 0$, then the edge $\{\hat1, \hat3\}$ is an upper face of $\PP_{1,4}^\omega$, and it is a  diagonal for the triangulation  ${\sf Q}_{1,4}(\omega)$;  if $D_{1,2,3,4} > 0$, the edge $\{2, 4\}$ now becomes the diagonal instead of $\{1,3\}$.
	\item[(2)] Suppose that we have a triangulation of the polygon $\{1,2, \cdots, m\}$
	having a triangle $\{1,j,m\}$ for some $\p_j$ with $1<j<m$.  Then we add the next vertex $\p_{m+1}$, and
	consider the following process:
	\begin{itemize}
		\item[(i)] We consider the 4-gon $\{1,j,m,m+1\}$, and mark the original boundary edge $\{1, m\}$ as  a ``dashed'' line, meaning that this edge should be checked as a diagonal for the triangulation $\Q_{1,m+1}(\omega)$.  See Fig.~\ref{fig:Gr(1,5)adding}.
	\item[(ii)] Use Lemma \ref{lem:QCL} to check whether or not this dashed line $\{1,m\}$ remains to be the diagonal for $\{1,j,m,m+1\}$.  If $D_{1,j,m,{m+1}} < 0$, the line $\{1, m\}$ is the diagonal, and we move to the step (3). If $D_{1,j,m,{m+1}} > 0$, we have a new diagonal $\{j, {m+1}\}$ in $\Q_{1,m+1}(\omega)$, which breaks the whole polygon $\{1,2, \cdots, m+1 \}$ into two polygons, and we mark $\{1, j\}$ and $\{j, m\}$ as dashed lines. See Lemma \ref{Not1face} below.
	We then mark the original edge $\{1,m\}$ as a ``dotted'' line, meaning that it is below the line $\{j,m+1\}$.  See Fig. \ref{fig:Gr(1,5)adding}.
	\item[(iii)] For these polygons $\{1,2,\cdots,j,m+1 \}$ and $\{j,j+1 \cdots,m+1 \}$, we repeat the process in (ii). 
\end{itemize}	
	\item[(3)] We repeat the step (2), and finally obtain the soliton triangulation $\Q_{1,M}(\omega)$ after adding $\p_M$.
	\end{itemize}
\end{algorithm}

\begin{remark}
It should be noted that one can choose any order of adding process of the vertices in the algorithm. In particular, we may choose a suitable order from the weights. For example, we choose the order in the adding process $\{i_1,i_2,i_3,\ldots\}$, if the weights are in the order $\omega_{i_1}>\omega_{i_2}>\omega_{i_3}>\cdots.$
\end{remark}

\begin{example}
Consider $\mathcal{A}^\omega_{1,5}$.  We demonstrate the construction of the triangulation
of $\A_{1,5}^\omega$ for an arbitrary choice of the weights as shown in Figure \ref{fig:Gr(1,5)adding}. 
We start from the quarirateral $\{1,2,3,4\}$, and use the determinant $D_{1,2,3,4}$ to construct a triangulation.  Depending on the sign of the determinant, we have two triangulations.
Then we add a vertex $\p_{5}$.  We now check whether $\{1,4\}$ is a diagonal for the quadrirateral 
$\{1,3,4,5\}$ or $\{1,2,4,5\}$.  If the edge $\{1,4\}$ remains the diagonal for 
the quadrilateral, then we have the triangulation of the pentagon. If not, then we have a new edge depending on the sign of the determinant $D_{i,j,k,l}$.  See Fig.~\ref{fig:Gr(1,5)adding}.
%%%%%%%%%%%%%%%%%%
\begin{figure}[h!]
  \centering
   \setlength\fboxsep{0pt}
   \setlength\fboxrule{0pt}
   \fbox{\includegraphics[width=12.5cm]{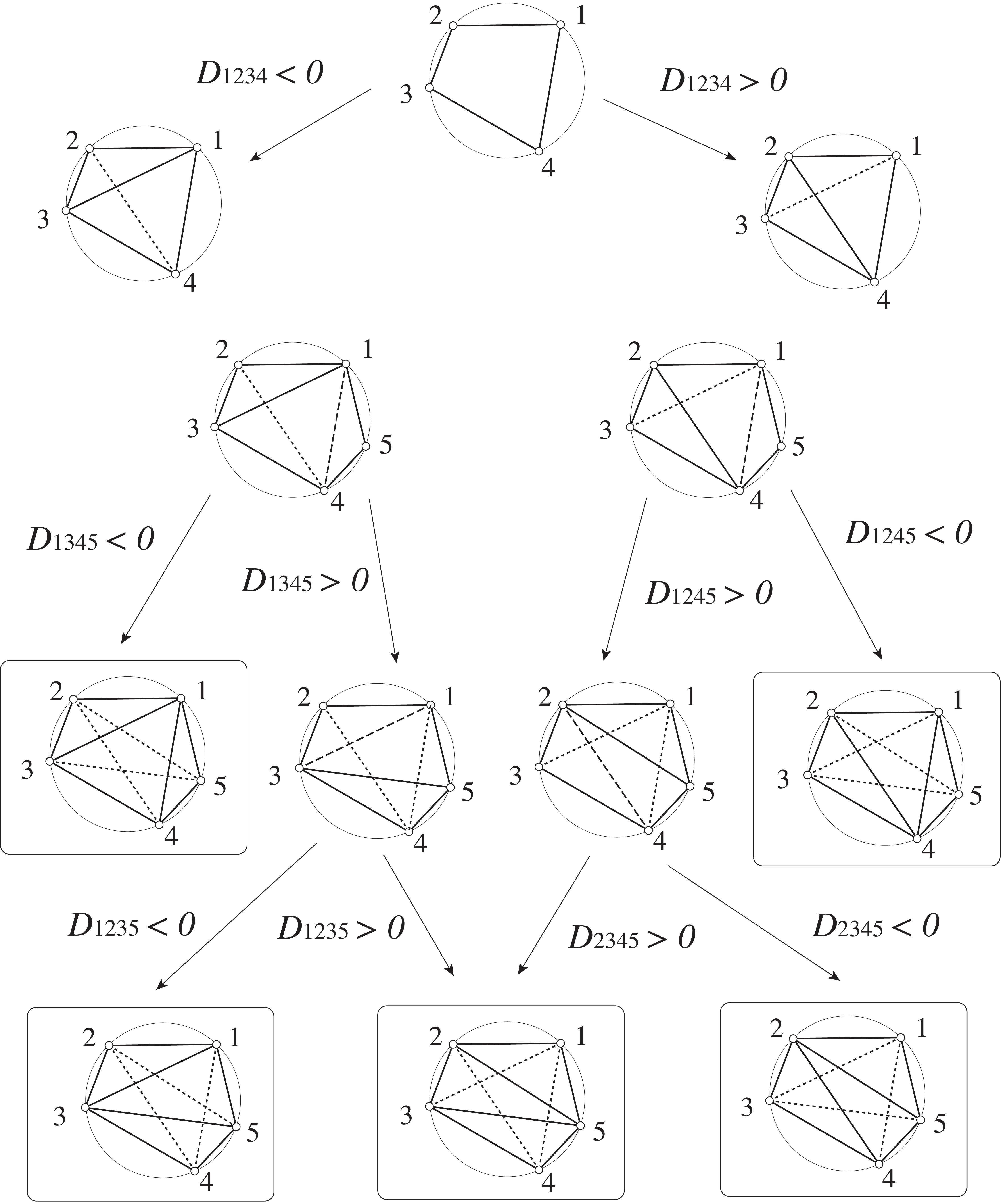} }
     \caption{Algorithm to construct the soliton subdivision $\Q_{1,5}$.
     Adding a new vertex $\p_{5}$, we mark the edge $\{1,4\}$ as a dashed line and  check whether or not it gives a diagonal for the new polygon $\P_{1,5}$.
     Each step is determined by the sign of the determinant $D_{i,j,k,l}$ for the quadrilateral $\{i,j,k,l\}$. The dotted lines are invisible edges obtained in the steps. } \label{fig:Gr(1,5)adding} 
\end{figure}
%%%%%%%%%%%%%%%%%%%%
\end{example}

Let us now state some lemmas to verify the algorithm.
\begin{lemma}\label{Not1face}
Let $\{\p_{i}:i=1,\ldots,M\}$ be the vertices of the $M$-gon ${\sf P}^0_{1,M}$.
Then the following two statements are equivalent for fixed weights $\omega$ and fixed $a,c \in [M]$:
\begin{itemize}
\item[(1)] The edge $\{\hat a,\hat c\}$ is \emph{not} an upper $1$-face of ${\sf P}^\omega_{1,M}$.
\item[(2)] There exists $b,d$ with $a<b<c<d$ (in the cyclic order) such that $\{\hat{b}, \hat{d}\}$ is an upper 1-face of $\P_{1,M}^\omega$.
\end{itemize}
\end{lemma}
\begin{proof}
If the edge $\{\hat a,\hat c\}$ is not an upper $1$-face, then $\{a,c\}$ is not a boundary edge of ${\sf P}^0_{1,M}$. Thus the edge $\{{a}, {c}\}$ breaks the polygon ${\sf P}_{1,M}^0$ into two parts. Consider one side, say $\PP_1$, which is also a polygon  having $\{{a}, {c}\}$ with $a<c$ as a boundary. We consider
the corresponding polytope ${\sf P}^\omega_1$, and we can find an upper $2$-face of ${\sf P}^\omega_1$ whose
vertices have indices including $a,c,b$ for some $b$ with $a<b<c$.  Consider a plane $\mathcal{P}$ spanned by those
three vertices $\{\hat a,\hat b,\hat c\}$. Then there exists a vertex $\hat \p_d$ on the other half of the polytope such that $\hat \p_d$ is above the plane $\mathcal{P}$.  This means that the edge $\{\hat b,\hat d\}$ is vertically above $\{\hat a,\hat c\}$.  There may be several vertices $\hat \p_d$, and
one can find at least one such vertex such that $\{\hat b,\hat d\}$ is a 1-face of $\P_{1,M}^\omega$.  If not, then $\{\hat a,\hat c\}$ should be 1-face of $\P_{1,M}^\omega$. But this contradicts.

The other direction, $(2)\to(1)$, is obvious.
\end{proof}

Lemma \ref{Not1face}  implies that if the edge $\{\hat{a}, \hat{c}\}$ with $a<c$ is an upper $1$-face, then for any $\{{b},{d}\}$ intersecting $\{{a},{c}\}$ in the $pq$-plane, $\{\hat{a}, \hat{c}\}$ is vertically above $\{\hat{b},\hat{d}\}$.

We also have the following lemma:

\begin{lemma}\label{separatelinelemma}
Let $\{a,b\} $ with $a<b$ be a diagonal in the subdivision $\Q_{1,M}(\omega)$, which divides the $M$-gon
$\P_{1,M}^0$ into two polygons, say $\P_1$ and $\P_2$. If $L$ is a diagonal of the 
triangulation for $\P_1$, then $L$ is also a diagonal in $\Q_{1,M}(\omega)$.
\end{lemma}

\begin{proof}
Denote such $L$ by $\{i,j\}$, and assume $a\leq i<j \leq b$.
Suppose  $L$ is not a diagonal for the subdivision $\Q_{1,M}(\omega)$. Then by Lemma \ref{Not1face}, we can find $c,d$ such that $\{\hat{c}, \hat{d}\}$ is above $\{\hat i,\hat j\}$
and $i<c<j$ and $b<d<a$ (in the cyclic order), i.e. $\p_c$ is a vertex in $\P_1$ and $\p_d$ is in $\P_2$.  Then the edge $\{\hat i,\hat j\}$ cannot be above the plane spanned by $\{\hat a,\hat b,\hat c\}$. Since the plane contains the edges $\{\hat a,\hat c\}$ and $\{\hat b,\hat c\}$,
at least one of these edges is vertically above $\{\hat i,\hat j\}$.  But this contradicts $L$ being
a diagonal of the subdivision for $\P_1$.
\end{proof}

Now we can give the proof of Algorithm \ref{algorithmGr1M}:
\begin{proof} 
Assume we have a soliton triangulation ${\sf Q}_{1,m}(\omega)$ of the polygon,
\[
{\sf P}_{\{i_1 \cdots i_m\}}:=\text{conv}\{\p_{i_j}:j=1,\ldots,m\},
\]
 where $i_1<i_2<\cdots<i_m$ in cyclic order.
We then add one more vertex $\p_{i_{m+1}}$ with $i_1<i_{m+1}<i_m$ in cyclic order.
 Let $\{i_1,i_l,i_m\}$ be a triangle in the triangulation $\Q_{1,m}(\omega)$ for some vertex $\p_l$. 
If  $\{\hat i_1,\hat i_m\}$ is vertically above $\{\hat i_{m+1}, \hat i_l\}$, then we have the vertex $\hat\p_{i_{m+1}}$ is below the plane containing the vertices $\{\hat\p_{i_1},\hat\p_{i_m},\hat\p_{i_l}\}$, hence $\{{i_1},{i_m}\}$ a diagonal in ${\Q}_{1,m+1}(\omega)$.  If $\{\hat i_{m+1}, \hat i_l\}$ is vertically above $\{\hat i_1,\hat i_m\}$,  Lemma \ref{separatelinelemma}
implies that $\{i_{m+1}, {i_l}\}$ breaks the polygon ${\sf P}_{\{i_1,\cdots, i_{m+1}\}}$ into two sub-polygons ${\sf P}_1$, ${\sf P}_2$, and $\{\hat i_{m+1}, \hat i_l\}$ is an upper $1$-face of  the polytope ${\sf P}_{\{i_1, \cdots,i_{m+1}\}}^\omega$. Inductively we can consider the triangulations of these sub-polygons ${\sf P}_1$ and ${\sf P}_2$.
\end{proof}

%%%%%%%%%%%%%%%%%%%%%%%%

\subsection{Inductive construction of the triangulation $\Q_{N,M}(\omega)$}\label{subsec:induction}
We now develop an inductive algorithm to construct the triangulation $\Q_{N+1,M}(\omega)$ from $\Q_{N,M}(\omega)$ for given weights $\omega$. For the case $N=1$, the triangulation $\Q_{1,M}(\omega)$ may be constructed by Algorithm \ref{algorithmGr1M} from the
previous section. We show how to construct $\Q_{2,M}(\omega)$ from $\Q_{1,M}(\omega)$, then give the inductive step in general.

The subdivision $\Q_{2,M}(\omega)$ corresponds to the point configuration
\[
\A_{2,M}^\omega=\left\{\hat\p_{ij}:1\le i<j\le M\right\}.
\]
where $\hat\p_{ij} = \hat\p_i +\hat \p_j$.
We must identify the points $\A_{2,M}^\omega$ which give the upper vertices
of the polytope $\P_{2,M}^\omega$. Since $\half(\p_i+\p_j)$ is the midpoint of the edge $\{\p_i,\p_j\}$,
the vertices in $\Q_{2,M}$ are induced by the edges in $\Q_{1,M}$. Moreover, 
the subdivision $\Q_{2,M}$ has the following structure:
\begin{itemize}
\item[(a)] The boundary vertices of $\Q_{2,M}$ are given by $\{\p_i+\p_{i+1}:i=1,\ldots,M,~\text{mod}(M)\}$.
\item[(b)] If the degree of the (boundary) vertex $\p_i$ in $\Q_{1,M}$ is $d_i\ge 3$, we have a $d_i$-gon with the vertices $\{\p_{ij_{\ell}}=\p_i+\p_{j_{\ell}}:\ell=1, \ldots,d_i\}$, 
which corresponds to a subdivision $\P_{1,d_i}$ of the point configuration $\A_{1,d_i}=\{\p_{j_{\ell}}:\ell=1,\ldots,d_i\}$.
\item[(c)] The number of internal vertices in $\Q_{2,M}$ is given by $M-3$, and we have
\[
\frac{1}{2}\sum_{i=1}^M(d_i-2)=M-3.
\]
\end{itemize}
Figure \ref{1Mto2M} demonstrates the induction process to construct a subdivision $\Q_{2,M}$ from the triangulation $\Q_{1,M}$ by an example with $M=11$.
%%%%%%%%%%%%%%%%%%%%%%%%%%%
\begin{figure}[h]
\begin{center}
\includegraphics[height=5cm]{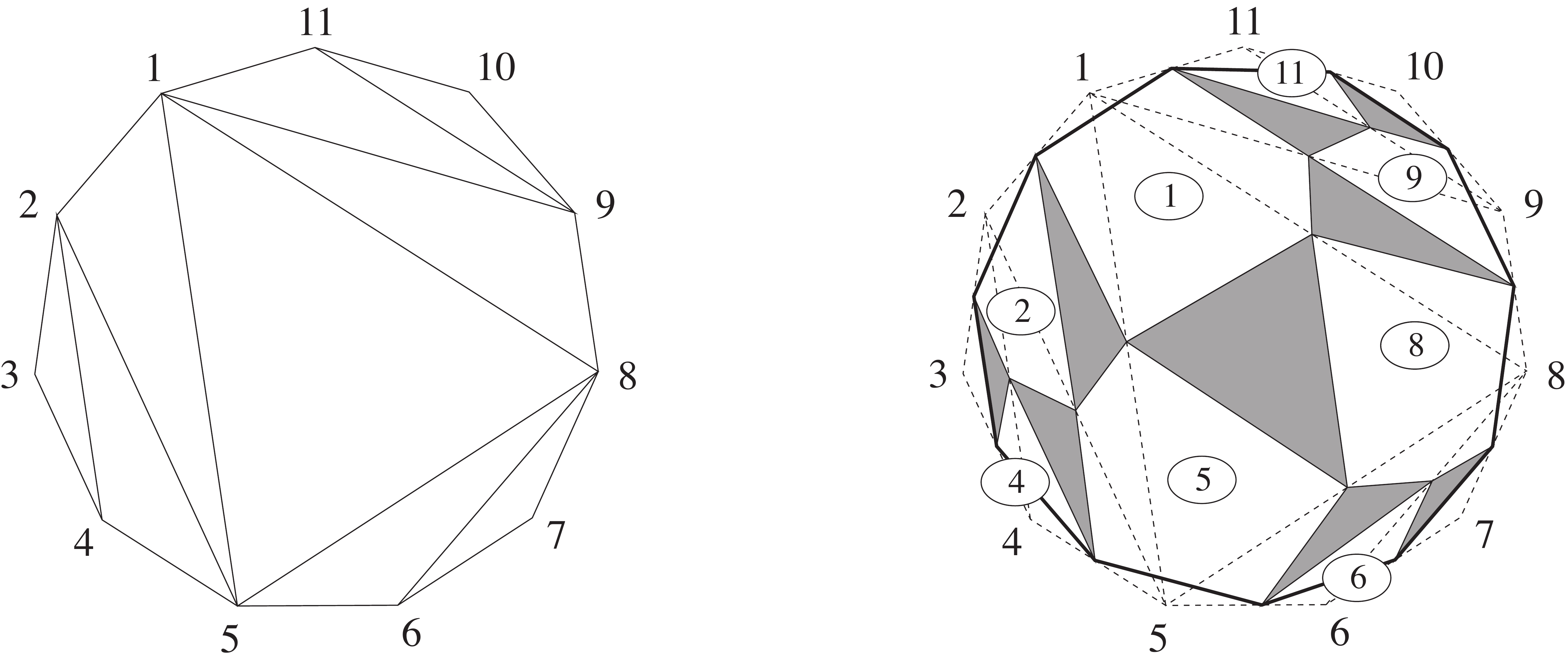}
\end{center}
\caption{A subdivision $\Q_{2,11}$ on the right is constructed from the triangulation $\Q_{1,11}$
in the left.  The vertices in $\Q_{2,11}$ are given by the midpoints of the edges in $\Q_{1,11}$ in this figure.
Each white polygon in $\Q_{2,11}$ corresponds to the vertex
$\p_i$ with the degree $d_i\ge3$ in $\Q_{1,11}$, and the number in each circle 
in the white polygon is the index of $\p_i$, which is the common index in Definition \ref{def:black-white}.
}\label{1Mto2M}
\end{figure}
%%%%%%%%%%%%%%%%%%%%%%%%%%%%%

We now describe the induction process from $\Q_{N,M}$ to $\Q_{N+1,M}$.  Let us first introduce a continuous process, called the \emph{$\epsilon$-blow up} or simply the \emph{blow-up}.
\begin{definition}
\label{def:e-blow-up}
Let $\Q_{N,M}$ be the triangulation for given weights $\omega$. 
For a number $0\le\epsilon\le1$, let $J+\epsilon a$ denote the vertex $\p_J+\epsilon\p_a$
for some index set $J$ and index $a \not\in J$.  Then
we define an \emph{$\epsilon$-blow up} of $\Q_{N,M}$
by the following procedure:
\begin{itemize}
\item[(a)] For each white triangle $\{{Ia},{Ib},{Ic}\}$ in $\Q_{N,M}$ for some $I\in\binom{[M]}{N-1}$, we replace it by the hexagon $\{{Ia}+\epsilon b, {Ia}+\epsilon c, {Ib}+\epsilon a, {Ib}+\epsilon c, {Ic}+\epsilon a,{Ic}+\epsilon b\}$. That is, when $\epsilon=0$, it is the original triangle,
and when $\epsilon=1$, it becomes a black triangle with the vertices $\{{Iab},{Ibc},{Iac}\}$.
\item[(b)] For each black triangle $\{{Kab},{Kbc},{Kac}\}$ in $\Q_{N,M}$ for some $K\in\binom{[M]}{N-2}$,  we replace it by the triangle 
 $\{{Kab} + \epsilon c,{Kbc}+\epsilon a, {Kac}+ \epsilon b\}$, which shrinks to the point $\p_{Kabc}$
 when $\epsilon=1$.
\end{itemize}
The $\epsilon$-blow up for $0<\epsilon<1$ is a $2M$-gon with the vertices $\{\p_{I_i}+\epsilon \p_{i+N},\p_{I_i}+\epsilon \p_{i-1}: i=1,\ldots, M ({\rm mod}~M)\}$ where $I_i=\{i,i+1,\ldots,i+N-1\}$ are the indices of the boundary vertices of $\Q_{N,M}$.
\end{definition}

To describe the structure of the $\epsilon$-blow up of the triangulation $\Q_{N,M}(\omega)$ for $\epsilon=1$, we first define the \emph{induced} degree for each vertex $\p$ in $\Q_{N,M}(\omega)$, denoted by I-deg$(\p)$, by 
\begin{equation}\label{def:Ideg}\index{induced degree (I-deg)}
\text{I-deg}(\p)=\{\text{\# of incoming edges to $\p$}\}~-~\{\text{\# of black triangles adjacent to $\p$}\}.
\end{equation}
Then one can see that the $\epsilon$-blow up of $\Q_{N,M}(\omega)$ consists of the following black and white polygons when $\epsilon=1$:
\begin{itemize}
\item[(a)] Each white triangle in $\Q_{N,M}(\omega)$ generates a black triangle.
\item[(b)] Each vertex $\p$ in $\Q_{N,M}(\omega)$ with I-deg$(\p)=m\ge 3$ generates a white $m$-gon.
\end{itemize}
Note that we can triangulate the $m$-gon in (b) using Algorithm \ref{algorithmGr1M}.

Now we have the following proposition:
\begin{proposition}\label{prop:Induction}
The $\epsilon$-blow up of the triangulation $\Q_{N,M}(\omega)$ generates a subdivision $\Q_{N+1,M}(\omega)$ for fixed $\omega$ when $\epsilon=1$.
\end{proposition}
\begin{proof}
We need to show that all the polygons generated in the $\epsilon$-blow up correspond
to the upper faces of $\P_{N+1,M}^\omega$. This can be shown as follows:
\begin{itemize}
\item[(a)] We first show that each black triangle $\{Iab,Ibc,Iac \}$ in $\Q_{N+1,M}(\omega)$ is obtained from the white triangle $\{Ia,Ib,Ic \}$ in $\Q_{N,M}(\omega)$ with $I \in {[M] \choose N-1}$. Since $\{\widehat{Ia},\widehat{Ib}, \widehat{Ic}\}$ is an upper $2$-face of ${\sf P}^\omega_{N,M}$, all the other vertices $\hat\p_J$, $J \in {[M] \choose N}$, are below the plane containing this face. Thus we can see all the vertices $\hat\p_K$, $K \in {[M] \choose N+1}$ are below the plane containing the points $\{\hat\p_{Iab},\hat\p_{Ibc}, \hat\p_{Iac}\}$, that is,  $\{\widehat{Iab},\widehat{Ibc}, \widehat{Iac}\}$ is an upper $2$-face of ${\sf P}^\omega_{N+1,M}$. 
\item[(b)] 
Now we show the each white triangle $\{Ii,Ij,Ik \}$ in $\Q_{N+1,M}(\omega)$ is also given by the projection of an upper $2$-face of the polytope ${\sf P}^\omega_{N+1,M}$. First, we have a plane $\mathcal{L}: \; z=ax+by+c$ containing $\{\hat\p_{Ii},\hat{\p}_{Ij},\hat\p_{Ik} \}$. Consider the plane $\mathcal{L}'$ parallel to $\mathcal{L}$ containing $\hat\p_I$ in ${\sf P}^\omega_{N,M}$. Then we can see that it is above all other vertices $\hat\p_J$ for all $J \in {[M] \choose N}$, and consequently, we have that $\mathcal{L}$ containing $\{\hat\p_{Ii},\hat\p_{Ij},\hat\p_{Ik}\}$ is above all the vertices $\hat\p_K$ for all $K \in {[M] \choose N+1}$. This implies that the white triangle $\{Ii,Ij,Ik\}$ is given by the projection of the upper 2-face $\{\widehat{Ii},\widehat{Ij},\widehat{Ik}\}$.
\end{itemize}
The items (a) and (b) complete the proof.
\end{proof}

This proposition gives an inductive algorithm to construct the triangulation $\Q_{N+1,M}(\omega)$ from $\Q_{N,M}(\omega)$.
\begin{algorithm}[Inductive construction of $\Q_{N+1,M}(\omega)$ from $\Q_{N,M}(\omega)$]\label{InductiveAlgorithm}

\noindent
\begin{itemize}
\item[(1)]  Apply the $\epsilon$-blow up to $\Q_{N,M}(\omega)$, and take $\epsilon=1$ to construct a subdivision $\Q_{N+1,M}(\omega)$.
\item[(2)] Use Algorithm \ref{algorithmGr1M} to triangulate the white polygons in the subdivision
$\Q_{N+1,M}(\omega)$ obtained in the previous step.
\end{itemize}
\end{algorithm}

In Fig.~\ref{fig:M11}, we illustrate the inductive construction of the triangulations
$\Q_{N,11}$ for $N=2,3$.  Here the triangulation $\Q_{2,11}(\omega)$ is obtained from the triangulations
of the white polygons in the subdivision $\Q_{2,11}(\omega)$ in Fig.~\ref{1Mto2M} (i.e. Step (2) in Algorithm \ref{InductiveAlgorithm}).
%%%%%%%%%%%%%%%%%%%%%%%%%%%
\begin{figure}[h]
\begin{center}
\includegraphics[width=14cm]{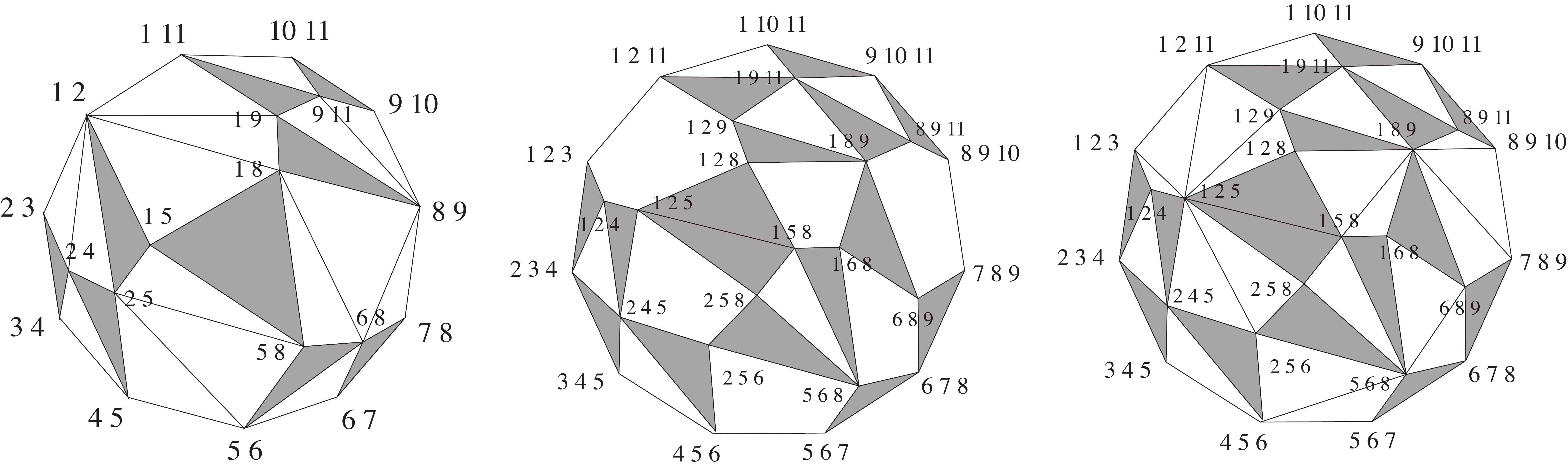}
\end{center}
\caption{Inductive construction of $\Q_{3,11}$ from $\Q_{2,11}$. Step (1) in Algorithm \ref{InductiveAlgorithm} shows the process from the triangulation $\Q_{2,11}(\omega)$ (left) to a subdivision $\Q_{3,11}(\omega)$ (middle). Then Step (2) provides the triangulation $\Q_{3,11}(\omega)$ (right).  }
\label{fig:M11}
\end{figure}
%%%%%%%%%%%%%%%%%%%%%%%%%%%%%

One can also show the following proposition about the topological structure of the triangulation $\Q_{N,M}(\omega)$:
\begin{proposition}
The triangulation $\Q_{N,M}(\omega)$ has
\begin{itemize}
\item[{\rm(1)}] $F_N^B:=N(M-N)-M+N$  black triangles, 
\item[{\rm(2)}] $F_N^W:=N(M-N)-N$  white triangles,
\item[{\rm(3)}] $V_N:=N(M-N)+1$ vertices, and
\item[{\rm(4)}] $E_N=3N(M-N)-M$  edges.
\end{itemize}
\end{proposition}

\begin{proof}
First note that the item (4) is a direct consequence of the items (1), (2) and (3) by the Euler characteristics.
That is, we have
\[
E_N=V_N+F_N-1=3N(M-N)-M,
\]
where $F_N=F_N^B+F_N^W$ is the total number of faces in $\Q_{N,M}(\omega)$.

We give an inductive proof based on the construction of the subdivision $\Q_{N+1,M}(\omega)$ via the $\epsilon$-blow up in Definition \ref{def:e-blow-up}.
When $N=1$, the subdivision $\Q_{1,M}$ is a triangulation of the $M$-gon with $M-3$ diagonals, thus it satisfies all items. 
\begin{itemize}
\item[(i)] From Definition \ref{def:e-blow-up}, it is clear that the number of black triangles in $\Q_{N+1,M}(\omega)$
is given by the number of white triangles in $\Q_{N,M}(\omega)$.  This means that we have
\begin{align*}
F_{N+1}^B&=F_N^W=N(M-N)-N\\
&=(N+1)(M-(N+1))-M+(N+1),
\end{align*}
which agrees with the formula in (1).
\item[(ii)]  Note that each edge in $\Q_{N,M}(\omega)$ becomes a vertex in $\Q_{N+1,M}(\omega)$. 
Also note that each black triangle in $\Q_{N,M}(\omega)$ shrinks to a vertex in $\Q_{N+1,M}(\omega)$.  This leads to
\begin{align*}
V_{N+1}&=E_N-2F_N^B=3N(M-N)-M-2(N(M-N)-M+N)\\
&=(N+1)(M-(N+1))+1,
\end{align*}
which gives the formula in (3).
\item[(iii)]  Since each vertex in $\Q_{N,M}(\omega)$ generates white triangles and the number of these triangles is related to the degree of the vertex, we first calculate the total degree of the vertices. Let $d_i$ be the degree of each vertex $\p_i$ in $\Q_{N,M}(\omega)$.  Then
the total degree of the vertices in $\Q_{N,M}(\omega)$ is given by
\[
\sum_{i=1}^{V_N} d_i=2E_N=2(3N(M-N)-M).
\]
Then note that each vertex $\p_i$ with degree $d_i\ge 2$ generates $d_i-2$ white triangles via the $\epsilon$-blow up at $\epsilon=1$.  However, three vertices of each black triangle shrink to a point at $\epsilon=1$, hence those vertices do not generate any triangles, and we have
\begin{align*}
F_{N+1}^W&=\sum_{i=1}^{V_N}(d_i-2)-3F_N^B=2E_N-2V_N-3F_N^B\\
&=(N+1)(M-(N+1))-(N+1).
\end{align*}
\end{itemize}
This completes the proof.
\end{proof}

\subsection{Connection to zonotopal tilings}
\label{subsec: zonotope}
By \cite[Corollary 10.9]{kodama}, every non-degenerate soliton graph for $\Gr(N,M)_{>0}$ is a reduced plabic graph.  Equivalently, every non-degenerate soliton subdivision ${\sf Q}_{N,M}$ is a triangulated plabic tiling.  In this section, we outline an alternate proof of this result, using Algorithm \ref{InductiveAlgorithm} and Galashin's results on Zonotopal tilings; see \cite{galashin}.  As a consequence, we obtain Kodama and Williams' classification of soliton graphs for $\Gr(2,M)_{>0}$; and derive a key lemma that we use in the proofs of Theorems \ref{realize37} and \ref{realize38}.

For polytopes $Z_1$ and $Z_2$, the \emph{Minkowski sum} of $Z_1$ and $Z_2$ is the set of points 
\[\{ \p_1 + \p_2 : \p_1 \in Z_1 \text{ and }\p_2 \in Z_2\}.\]
A \emph{zonotope} is a polytope which is a Minkowski sum of line segments.  The \emph{cyclic zonotope} $Z(3,M)$ is the Minkowski sum of $M$ segments $[0,\pp_i]$ in $\mathbb{R}^3$, where 
\[\pp_i = (\k_i,\k_i^2,1)\] 
and $\k_1 < \k_2 < \ldots < \k_M$.  A \emph{zonotopal tiling} of $Z(3,M)$ is a subdivision of $Z(3,M)$ into smaller zonotopes (called tiles), each of which is the Minkowski sum of 
\[\{ \pp_i : i \in I\} \cup \{ [0,\pp_j] : j \in J\}\] 
for some disjoint $I, J \subseteq [M]$; and such that the intersection of any two tiles is either empty, or a lower-dimensional tile.  We say the tiling is \emph{fine} if each top-dimensional tile is a translate of the Minkowski sum of at most three segments.

Let $\p_1,\ldots,\p_M$ be points on the parabola $q = p^2$, and let 
\[\omega = (\omega_1,\ldots,\omega_M)\]
be a weighting of the points.  Without loss of generality, assume $\omega_1,\ldots,\omega_M > 0$. Repeatedly applying the induction algorithm gives a subdivision of $\Q_{N,M}(\omega)$ for each $0 \leq N \leq M$. This collection of subdivisions induces a tiling of the cyclic zonotope, whose top-dimensional tiles are in one-to-one correspondence with the white triangles in the family of subdivisions. Figure \ref{fig:zonotope} illustrates this construction for $M = 4$, and $1 \leq N \leq 3$.
The horizontal direction shows the $\epsilon$-blow up of the vertices, and the number of blowing-up
$\epsilon\to1$ directions are given by the I-degree of the vertex.

%%%%%%%%%%%%%%%%%%%%%%%%%%%
\begin{figure}[h]
\begin{center}
\includegraphics[height=3.2cm]{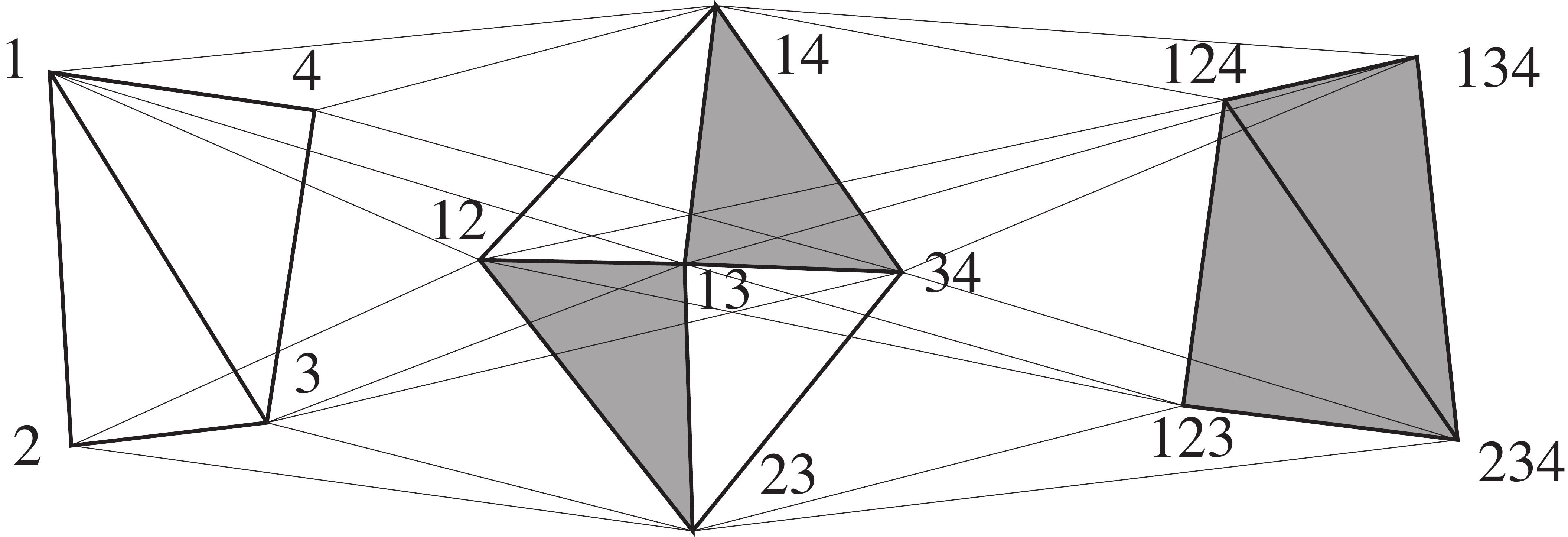}
\end{center}
\vskip-0.3cm
\caption{Zonotope structure in the $\epsilon$-blow up.  
Starting with the triangulation of the 4-gon $\Q_{1,4}$ at left, we have the triangulation $\Q_{2,4}$
in the middle by blowing-up each vertex of $\Q_{1,4}$. Note that each edge shrinks to a point
via the $\epsilon$-blow up when $\epsilon=1$.}
\label{fig:zonotope}
\end{figure}
%%%%%%%%%%%%%%%%%%%%%%%%%%%%%
 
Each vertex of a zonotopal tiling which lies in the plane $z = N$ has the form
 \[\pp_I = \sum_{i \in I} \pp_i\]
 for some $I \in {{[M]}\choose{N}}$, and the vertex labels $I$ form a maximal weakly separated collection $\mathcal{D}_N$ in ${{[M]}\choose{N}}$.  Intersecting a fine zonotopal tiling of $Z(3,M)$ with the plane $z = N$ gives a triangulated plabic tiling, i.e. soliton triangulation, with vertex set $\mathcal{D}_N$ \cite[Theorem 2.1]{galashin}. It follows from the above discussion that every soliton subdivision $\Q_{N,M}$ is combinatorially equivalent to a section of a zonotopal tiling, and is therefore a triangulated plabic tiling for $\Gr(N,M)_{>0}$.  Hence we recover the result of \cite{kodama} that soliton graphs for $\Gr(N,M)_{>0}$ are plabic graphs.
 
It follows from the discussion in \cite[Section 4]{galashin} that for $M \geq 2$, every plabic tiling for $\Gr(N,M)_{>0}$ may be obtained from some \emph{triangulated} plabic tiling of $\Gr(N-1,M)_{>0}$, by applying a purely combinatorial analog of the induction algorithm, Algorithm \ref{InductiveAlgorithm}.  The following observation, which we will use in the proofs of Theorem \ref{realize37} and Theorem \ref{realize38}, is immediate. 

\begin{lemma}
\label{tilings}
If every \emph{triangulated} plabic tiling for $\Gr(N-1,M)_{>0}$ is realizable, then every plabic tiling for $\Gr(N,M)_{>0}$ is realizable. 
\end{lemma}

For any $M$, it is easy to see that triangulated plabic tilings for $\Gr(1,M)_{>0}$ are realizable. Hence we obtain Kodama and Williams' result \cite[Theorem 12.1]{kodama}, which states that all plabic tilings for $\Gr(2,M)_{>0}$ are realizable.  Moreover, from Algorithm \ref{InductiveAlgorithm} it is clear that these plabic tilings are in one-to-one correspondence with the triangulations of the $M$-gon, just as described in \cite{kodama}.

%%%%%%%%%%%%%%%%%%%%%%%%%%%%%%%

\section{KP hierarchy and the polyhedral structure of multi-time space}
\label{sec: polyhedral}
In this section, we define a \emph{polyhedral fan} structure on the space of \emph{multi-time parameters} $\mathbf{t}$ in the KP hierarchy. The \emph{cones} in this fan structure correspond to realizable soliton subdivisions $\Q_{N,M}$ for a fixed choice of $\kappa$ parameters.  That is, one can construct the soliton graph dual to the subdivision $\Q_{N,M}(\omega(\t))$ by choosing multi-time parameters $\t\in\R^{M-3}$
in the corresponding cone.  See Theorem \ref{realizable-checking-NM} for a precise statement.  We will apply these results to classify soliton graphs for $\Gr(3,6)_{>0}$.

\subsection{The KP hierarchy and the multi-time space}
\label{subsec: multi-time}
In this section, we show that we can realize triangulations corresponding to arbitrary weight vectors $\omega$ simply by varying the multi-time parameters $\mathbf{t}$.  We also review \emph{polyhedral cones}.

Recall that the piecewise linear function $f_{N,M}(x,y,{\bf t})$ for a KP soliton is given by
\begin{align*}
f_{N,M}(x,y,{\bf t})&={\text{max}}\left\{\Theta_I(x,y,{\bf t})=p_Ix+q_Iy+\omega_I({\bf t}):~~ I\in\binom{[M]}{N}\right\}\\[0.5ex]
&\text{with}\quad p_I=\sum_{i\in I}\k_i,\quad q_I=\sum_{i\in I}\k_i^2,\quad\omega_I=\sum_{i\in I}\omega_i({\bf t}),
\end{align*}
where  $\omega_i({\bf t})$ with ${\bf t}=(t_3,\ldots,t_{{M-1}})\in\R^{M-3}$ is given by
\[
\omega_i({\bf t})=\sum_{k=3}^{M-1}\kappa_i^kt_k.
\]
The point configuration for this case is
\[
\A^{\omega({\bf t})}_{N,M} = \left\{ \hat\p_I=(p_I,q_I,\omega_I({\bf t}))\in\mathbb{R}^3 \; : \; I \in {[M] \choose N} \right\}.
\]

Let $\Omega = (\Omega_1,\ldots,\Omega_M)$ be a weight vector, and let $Q_{N,M}(\Omega)$ be the corresponding soliton triangulation.  We claim that there exists $\bf{t}$ such that $Q_{N,M}(\omega(\mathbf{t}))$ is combinatorially equivalent to $Q_{N,M}(\Omega)$.

To see this, we first define the plane $\ell_i(t_0,x,y,\t):=t_0+\theta_i(x,y,{\bf t})$, and
consider the system $\ell_i=\Omega_i$ for $i=1,\ldots,M$, i.e.
\[
\begin{pmatrix}
1 & \kappa_1 &\kappa_1^2 &\cdots & \kappa_1^{M-1}\\
1&\kappa_2 &\kappa_2^2 &\cdots &\kappa_2^{M-1}\\
1&\kappa_3&\kappa_3^2&\cdots&\kappa_3^{M-1}\\
\vdots & \vdots &\vdots &\ddots &\vdots\\
1&\kappa_M&\kappa_M^2 &\cdots &\kappa_M^{M-1}
\end{pmatrix} 
\begin{pmatrix}
t_0 \\ x\\  y \\ \vdots \\ t_{M-1}
\end{pmatrix}=
\begin{pmatrix}
\Omega_1 \\ \Omega_2 \\ \Omega_3 \\ \vdots \\ \Omega_{M}
\end{pmatrix}.
\]
Since the coefficient matrix is the Vandermonde matrix with distinct $\kappa_j$'s,
the system has a unique solution, which we denote by
$(a_0,x_0,y_0,{\bf a})$ with ${\bf a}\in\R^{M-3}$.   Next, we consider the plane defined by
\[
z=\ell_i(a_0,x,y,{\bf a})=a_0+\kappa_ix+\kappa_i^2y+\omega_i({\bf a}),
\]
which can be rewritten in the form
\[
z=\kappa_i(x-x_0)+\kappa_i^2(y-y_0)+\Omega_i,\qquad i=1,\ldots,M.
\]
Setting $x_0 = y_0 = 0$ translates the contour plot in the $xy$-plane, but does not change its combinatorial structure.
In other words, setting $\mathbf{t} = \mathbf{a}$ gives a choice of multi-time parameters corresponding to the soliton triangulation with the $\k$-parameters and weight vector $\Omega$. Hence
giving the weight vector $\Omega$ is equivalent to choosing a particular direction in the time-space.  Our aim is to identify the polyhedral structure in the time-space
with $\t=(t_3,\ldots, t_{M-1})$ variables.  

The plane $\ell_i=\Omega_i$ in the $M$-dimensional space with $(t_0,x,y,\t)$
has the normal vector $(1,\kappa_i,\ldots,\kappa_i^{M-1})$, and
the hyperplane arrangements,
\[
\ell_i-\ell_j=\theta_i-\theta_j=0\qquad 1\le i<j\le M,
\]
divides $\mathbb{R}^{M}$ into $M!$ regions.
The dual to the set of those regions gives a permutohedron for the symmetric group
$S_M$.  

We now give a few definitions.
\begin{definition}\label{cone}\index{soliton graph!polyhedral cone}
A \emph{polyhedral cone} (or \emph{cone} for short) generated by a finite set of vectors 
\[\B_J:=\{{\bf b}_j:j\in J\subset [M]\}\] is defined by
\[
{\rm cone}(\B_J):=\left\{\sum_{j\in J}\lambda_j{\bf b}_j:~{\bf b}_j\in \B_J,~ \lambda_j\ge0\right\}.
\]
If the dimension of  ${\rm cone}(\B_J)$ is $k$, we define the \emph{relative interior} of ${\rm cone}(\B_J)$, denoted by 
${\rm relint}({\rm cone}(\B_J))$, 
as the collection of all points $\p$ in ${\rm cone}(\B_J)$ such that
 there exists a small ball $B_\p$ of dimension $k$ centered at $\p$ such that $B_\p \subset {\rm cone}(\B_J)$. That is, the relative interior is the interior with in the topology of the 
 subspace spanned by the cone.
 \end{definition}

\begin{example}
Consider $e_1 = (1,0,0)$, $e_2=(0,1,0)$ in $\R^3$. Then 
$\text{cone}(\{e_1,e_2 \})$ is the region $\{ (x,y,0) : \,x \ge 0, y \ge 0 \}$, and has dimension two; while $\text{relint(cone}(\{e_1,e_2 \}))$ is the region $\{ (x,y,0) :\;x > 0, y > 0 \}$, which is not an interior of the topology of $\R^3$.
\end{example}

%%%%%%%%%%%%%%%%%%%%%%%%%%%%%%%%%%

\subsection{Polyhedral cones for $\Gr(1,M)_{>0}$ and $\Gr(2,M)_{>0}$}
\label{subsec: fan}
Let us first consider a non-generic subdivision of $\A_{1,M}^{\omega}$ induced from the weight $\omega = (-1,0,\cdots, 0)$. With this weight vector, the lifted points $\{\hat \p_2, \hat \p_3,\cdots, \hat \p_M\}$ have the same height, while $\hat \p_1$ is below the plane containing these points. Thus we have a non-generic subdivision with only one diagonal $\{2, M\}$. There exists a vector, say $\r_1^-$, in the $\t$-space $\R^{M-3}$, such that each point on the cone$\{\r_1^-\}$ generates this non-generic subdivision. The negative $-$ sign in the notation $\r_1^-$ means the $-1$ in the weight $\omega=(-1,0,\cdots, 0)$. Similarly, we define $\r_i^-$ for all other $i=2,\cdots, M$. We also define $\r_i^+ = -\r_i^-$, which will be explained in more detail below.  We call these $\r_i^\pm$ the \emph{main rays} in $\t$-space for non-generic subdivisions.

Before giving explicit coordinates for the vectors $\r_i^-$, we note the following lemma:
\begin{lemma}\label{XY-equivalence}
Two weights $\omega=(\omega_1,\cdots, \omega_M)$ and $\omega'=(\omega'_1,\ldots,\omega_M')$
with $\omega_i'=t_0+\k_ix_0+\k_i^2y_0+\omega_i$ for arbitrary $(t_0,x_0,y_0)$  give the same subdivision, that is, we have $\Q_{N,M}(\omega)=\Q_{N,M}(\omega')$.
\end{lemma}
\begin{proof}
Translating the coordinates $(x,y,z)$ by $(x+x_0,y+y_0,z-t_0)$, each plane 
\[z=\theta_i(x,y,{\bf t})=\k_ix+\k_i^2y+\omega_i({\bf t})\] becomes $z=\k_ix+\k_i^2y+\omega'_i({\bf t})$.
It is obvious that the dominance relation among the planes does not change under the translations of coordinates.
\end{proof}

\begin{remark}\label{rem:Gale}
Recall that for given weights $(\Omega_1,\ldots,\Omega_M)\in\mathbb{R}^M$,
one can find a unique point $(t_0,x,y,{\bf t})\in \R^M$ such that the planes $\ell_i$ are given by $t_0+\k_ix + \k_i^2y + \omega_i(\t) =\Omega_i$ for $i=1,\ldots,M$. 
Then Lemma \ref{XY-equivalence} implies that the subdivision can be determined by only
the time variable ${\bf t}=(t_3,\ldots,t_{M-1})\in \R^{M-3}$.
 \end{remark}

To find the vector $\r_i^-\in\R^{M-3}$, we consider the following system of equations
for $(t_0,x,y,t_3,\ldots,t_{M-1})$,
\begin{equation}\label{e:GaleV}
t_0+\k_ix+\k_i^2y+\sum_{k=3}^{M-1}\k_i^kt_k=\Omega_i\qquad\text{for}\quad i=1,\ldots,M,
\end{equation}
where we assign the weights $(\Omega_1,\ldots,\Omega_M)$ with $\Omega_j=-\delta_{i,j}$
(Kronecker delta).  Then, by Lemma \ref{XY-equivalence} and Remark \ref{rem:Gale}, the (column) vector $\r_i^-$ is given by
the last $M-3$ components in the solution of this system, i.e.
\[
\r^-_i=(t_3,t_4,\ldots,t_{M-1})^T \in\R^{M-3}.
\] 
Equation \eqref{e:GaleV} can be written in the $M\times M$ matrix form,
\[
R V=-Id\qquad\text{or}\qquad \sum_{k=1}^Mr_{i,k}\k^{k-1}_j=-\delta_{i,j},
\]
where $V = (\kappa_j^{i-1})_{1\le i,j\le M}$ is the Vandermonde matrix, and $Id$ is the identity matrix.  The solution matrix $R=-V^{-1}$ can be obtained by the Lagrange interpolation.  Consider the polynomial,
\[
p_i(\k)=\sum_{k=1}^Mr_{i,k}\k^{k-1}\qquad \text{with}\qquad p_i(\k_j)=-\delta_{i,j}.
\]
The Lagrange interpolation formula then gives
\[
p_i(\k)=-\prod_{l\ne i}\frac{\k-\k_l}{\k_i-\k_l}=\frac{-1}{\prod_{l\ne i}(\k_i-\k_l)}\sum_{k=1}^{M}(-1)^{M-k}e^{(i)}_{M-k}\k^{k-1},
\]
where $e^{(i)}_k$ is the $k$-th elementary symmetric polynomial of $(\k_1,\ldots,\hat\k_i,\ldots,\k_M)$
(missing  the $\k_i$ variable). Explicitly, we have

\[e^{(i)}_k = \sum_{\substack{1 \leq s_1 < s_2 \cdots < s_{k} \leq M \\ i \not\in \{s_1,\ldots,s_k\}}} \kappa_{s_1}\kappa_{s_2}\cdots \kappa_{s_k}.\]
 
Thus, we have
\begin{equation}\label{e:GaleVector}
r_{i,k}=\frac{-1}{\prod_{l\ne i}(\k_i-\k_l)}(-1)^{M-k}e^{(i)}_{M-k}.
\end{equation}
The vector $\r_i^-$ is then given by
\[
\r_i^-=(r_{i,4},r_{i,5},\ldots,r_{i,M})^T\in\R^{M-3}.
\]

\begin{definition}\label{def:GaleT}\index{Gale transform}\index{Gale transform!Gale vector}
The set of vectors $\B:=\{\r_1^-,\ldots,\r_M^-\}$ is called the \emph{Gale transform} of the point configuration
$\A_{1,M}=\{(\k_i,\k_i^2)\in\R^2:i=1,\ldots,M\}$ and the vectors $\r_i^{-}$ are referred to as the \emph{Gale vectors}. The Gale transform is defined formally using the following procedure, which applies more generally to point configurations (see e.g. \cite{DRS:10, T:06}).  Consider the $3\times M$ matrix representing $\A_{1,M}$,
\[A:=
\begin{pmatrix}
1&  1& \cdots & 1 \\
\kappa_1 &  \kappa_2 & \cdots & \kappa_{M} \\
\kappa_1^2 &  \kappa_2^2 & \cdots & \kappa_{M}^2 \\
\end{pmatrix},
\]
and consider the kernel of $A$ 
\[
{\rm ker}_{\R}(A):=\{\u \in \R^M \,:\, A\, \u=\mathbf{0}\}. 
\]
Let $\{\u_1, \cdots, \u_{M-3}\}$ be a basis for the vector space ${\rm ker}_{\R}(A)$. We organize these vectors as the columns of an $M \times (M-3)$ matrix $B$, so $AB=O_{3\times (M-3)}$, the
$3\times (M-3)$ zero matrix.
Then
\[
B := [\, \u_1, \, \u_2, \, \cdots \,, \u_{M-3}].
\] 
The $M$ ordered rows of $B$ give $B^T=[\r_1^-,\ldots,\r_M^-]$, which is the Gale transform $\B$.
\end{definition}

The Gale transform $\mathcal{B}$ is a useful tool to read off the polygons in the regular subdivision $\Q_{N,M}(\omega)$, and the faces of the polytope ${\sf P}^\omega=\text{conv}(\{\p_1, \cdots, \p_M\})$. 
The following theorem gives the method to check the regularity of a subdivision using the Gale transform.

\begin{theorem}[\cite{Le:91}]\label{Regular-checking}
 Let $\Q = \{\sigma_1, \cdots, \sigma_m \}$, $\sigma_i \subset [M]$ for $i=1,\cdots,m$,  be a subdivision of  a point configuration $\A$ and let $\B$ be a Gale transform of $\A$. Then $\Q$ is regular if and only if 
\[
\bigcap_{i=1}^m ~{\rm {relint(cone}}(\B_{\bar{\sigma_i}})) \neq \emptyset,
\]
where $\B_{\bar{\sigma}_i}:= \{\r_j^-:j\in\bar{\sigma}_i=[M]\setminus \sigma_i\}$.
\end{theorem}

Note that if $\t\in{\rm cone}\{\r_i^-\}$, then the subdivision $\Q_{1,M}(\omega(\t))$ consists of
the triangle $\{i-1,i,i+1\}$ and the $(M-1)$-gon with the vertices $\{\p_j:j\in[M]\setminus\{i\}\}$.
One can then easily find that the subdivision having just one diagonal, say $\{i,j\}$, can be
constructed by choosing the time variable $\t$ in the following cone,
\[
{\rm cone}\{\B_{\sigma_{i,j}}\}~\cap~{\rm cone}\{\B_{\tau_{i,j}}\},
\]
where $\sigma_{i,j}$ and $\tau_{i,j}$ are defined by
\[
\sigma_{i,j}=\{i+1,\ldots,j-1\},\quad \tau_{i,j}=\{j+1,\ldots, i-1\}\quad (\text{in the cyclic order}),
\]
that is, $\sigma_{i,j}\cup\tau_{i,j}=[M]\setminus\{i,j\}$. Note that the dimension of the intersection is one:
Writing $j=i+k+1$ (mod $M$), then we have
\[
\text{dim}({\rm cone}\{\B_{\sigma_{i,j}}\})=k,\qquad \text{dim}({\rm cone}\{\B_{\tau_{i,j}}\})=M-k-2.
\]
Then we define a vector $\r_{[i,j]}$ such that
\[
{\rm cone}\{\r_{[i,j]}\}={\rm cone}\{\B_{\sigma_{i,j}}\}~\cap~{\rm cone}\{\B_{\tau_{i,j}}\}.
\]
Notice that $\r_i^-=\r_{[i-1,i+1]}$.  Then it is immediate that we have the following propositions:

\begin{proposition}\label{prop:Q1M-checking}
A white polygon $\P_{\sigma}$ with vertex set $\sigma\subset [M]$ shows up in the subdivision $\Q_{1,M}(\omega(\t))$ 
if and only if the time variable $\t\in\R^{M-3}$ belongs to the relative interior
${\rm {relint(cone}}\{\B_{\bar{\sigma}}\})$ where $\bar{\sigma} = [M] \setminus \sigma$.
\end{proposition}

\begin{proposition}\label{prop:Q1Msub}
A subdivision $\Q_{1,M}(\omega(\t))$ has the diagonals $\{i_l,j_l\}$ for $l=1,\ldots,m$, if and only if
\[
\t~\in~{\rm relint}\left({\rm cone}\{\r_{[i_l,j_l]}:l=1,\ldots,m\}\right).
\]
\end{proposition}

Since any triangulation of the $M$-gon has $M-3$ diagonals, Proposition \ref{prop:Q1Msub} implies that one can realize a unique triangulation $\Q_{1,M}(\omega(\t))$ with $M-3$ diagonals $\{i_l,j_l\}$ for $l=1,\ldots,M-3$ by choosing $\t$ in the proposition.

\begin{example}\label{ex:Q15}
Consider the point configuration for $N=1$ and $M=5$:
\[
\A_{1,5}=\{(\k_i,\k_i^2):i=1,\ldots,5\},
\]
where we take $(\k_1,\ldots,\k_5)=(-2,-1,0,1,2)$.  Then, from \eqref{e:GaleVector},  the Gale transform $\B=\{\r_1^-,\ldots,\r_5^-\}$ is given by
\[
\r_1^-=\frac{1}{12}\begin{pmatrix} 2\\-1\end{pmatrix},\quad
\r_2^-=\frac{1}{6}\begin{pmatrix}-1\\1\end{pmatrix},\quad
\r_3^-=\frac{1}{4}\begin{pmatrix}0\\-1\end{pmatrix},\quad
\r_4^-=\frac{1}{6}\begin{pmatrix}1\\1\end{pmatrix},\quad
\r_5^-=\frac{1}{12}\begin{pmatrix}-2\\1\end{pmatrix}.
\]
The left figure in Fig.~\ref{fig:triangulationP15} illustrates the triangulations of the pentagon $\P_{1,5}$.
Each triangulation for $\P_{1,5}$ can be obtained by using Proposition \ref{prop:Q1M-checking}.
For example, the triangulation containing the triangle $\sigma=\{1,3,5\}$ can be obtained by choosing 
a point $(t_3,t_4)$ in the cone spanned by $\r_2^-$ and $\r_4^-$ (note $\{2,4\}=[5]\setminus\{1,3,5\}$), i.e.
\[
\t=(t_3,t_4)~\in~{\rm relint}\left({\rm cone}\{\r^-_2,\r^-_4\}\right).
\]
One should note that the triangulation $\Q_{1,5}=\{\sigma_1,\sigma_2,\sigma_3\}$ with
$\sigma_1=\{1,2,3\},\, \sigma_2=\{1,3,5\},\, \sigma_3=\{3,4,5\}$ can be realized with $\t=(t_3,t_4)$
in relint(cone$\{\B_{\bar\sigma_2}\}$), since
\[
{\rm cone}\{\B_{\bar\sigma_1}\}~\cap~{\rm cone}\{\B_{\bar\sigma_2}\}~\cap~{\rm cone}\{\B_{\bar\sigma_3}\}~=~{\rm cone}\{\B_{\bar\sigma_2}\}.
\] 

We also remark that in terms of the determinant $D_{i_1,i_2,i_3,i_4}$ in \eqref{e:4-gonD},
each direction $\r_i^-$ can be described by $D_{[5]\setminus i}=0$, i.e. the vertices
$\{\p_j:j\in[5]\setminus i\}$ are coplanar.
\end{example}
%%%%%%%%%%%%%%%%%%%%%%%%%%%
\begin{figure}[t!]
\begin{center}
\includegraphics[height=6cm]{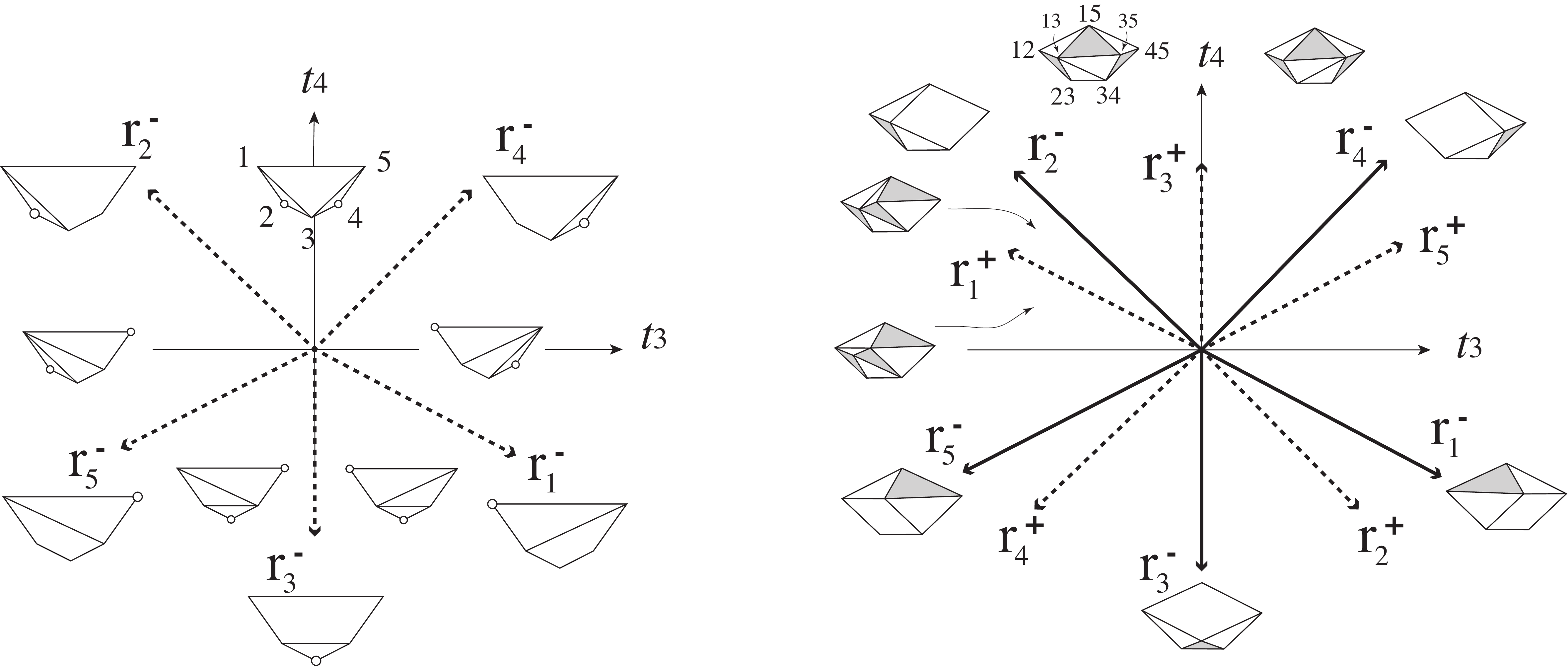}
\end{center}
\caption{Triangulations $\Q_{1,5}(\omega(\t))$ (left) and $\Q_{2,5}(\omega(\t))$ (right) for the weights $\omega_i(\t)=\k_i^3t_3+\k_i^4t_4$.
The set  $\B=\{\r_i^-:i=1,\ldots,5\}$ is the Gale transform and $\r_i^+=-\r_i^-$.
Here the $\k$-parameters are given by $(\k_1,\ldots,\k_5)=(-2,-1,0,1,2)$.
}
\label{fig:triangulationP15}
\end{figure}
%%%%%%%%%%%%%%%%%%%%%%%%%%%%%

Now we consider the subdivision for the configuration $\A_{2,M}^{\omega(\t)}$. We use Algorithm \ref{InductiveAlgorithm} 
to construct $\Q_{2,M}(\omega(\t))$. The right figure in Fig.~\ref{fig:triangulationP15} shows
the subdivisions obtained from the triangulations $\Q_{1,M}(\omega(\t))$ in the left figure for $M=5$. Each black triangle  in $\Q_{2,M}(\omega(\t))$ corresponds to a white triangle in $\Q_{1,M}(\omega(\t))$
for the same $\t$.  Notice that the dotted lines in the left figure become 
the solid lines which are the boundaries corresponding to the \emph{black-white flips}. For example, the solid line of $\r_2^-$ is the boundary corresponding to the black-white flip of the parallelogram $\{13, 34, 45, 15\}$.

The subdivision obtained by the algorithm contains some white $k$-gons where $k$ is given by the degree of the corresponding vertex. More precisely, such a $k$-gon has the index set $\{i_0i_1,i_0i_2,\ldots,i_0i_k\}\subset \binom{[M]}{2}$ for a common index $i_0\in[M]$ where $i_0$ is the index of the vertex $\p_{i_0}=(\k_{i_0},\k^2_{i_0})$ in the $M$-gon $\P_{1,M}$. Then one can triangulate this white polygon using the algorithm for the configuration $\A_{1,k}=\{\p_{i_1},\ldots,\p_{i_k}\}$.  
For example,
consider the subdivision $\Q_{2,5}$ in cone$\{\r_2^-,\r_4^-\}$ which has the white quadrilateral with vertex set
$\{13,23,34,35\}$.  Then 3 is the common index, and we triangulate the $\A_{1,4}=\{\p_1,\p_2,\p_4,\p_5\}$.  A triangulation is given by choosing the triangle $\{1,2,4\}$ as shown in Fig.~\ref{fig:triangulationP15}, and it is obtained by choosing the time variable 
\[
\t=(t_3,t_4)~\in~{\rm relint}\left({\rm cone}\{\r_5^-,\r_3^+\}~\cap~{\rm cone}\{\r_2^-,\r_4^-\}\right).
\]
The $\r_3^+$ in the first cone indicates the dominant (or common) index, and the $\r_5^-$ indicates
the missing index in $\A_{1,4}$ for the white 4-gon. Note that  the intersection is also given by
${\rm cone}\{\r_2^-,\r_3^+\}$ which is obtained by taking the other triangle $\{1,4,5\}$ in the 4-gon.
In the next section, we discuss the general case.

%%%%%%%%%%%%%%%%%%%%%%%%%%%%%%
\subsection{A realizability theorem for $\Q_{N,M}(\omega(\t))$}
\label{subsec: main theorem}
We first define the following pair of indices $(K_\sigma^+, K_\sigma^-)$ for the triangles $\sigma=\{I,J,L\}$ with $I,J,L\in \binom{[M]}{N}$ in the subdivision $\Q_{N,M}(\omega)$:
\begin{itemize}
\item[(a)] If $\sigma$ is a {white} triangle, the vertices 
of the triangle $\sigma$ are expressed by
$\{K_{\sigma}a,K_{\sigma}b,K_{\sigma}c\}$ for some $K_{\sigma}\in\binom{[M]}{N-1}$.
We define a pair of indices $(K_\sigma^+,K_{\sigma}^-)$ as
\begin{equation}\label{def:whiteK}
K_{\sigma}^+:=K_{\sigma},\quad K_{\sigma}^-=[M]\setminus (K_{\sigma} \cup\{a,b,c\}).
\end{equation}
\item[(b)] If $\sigma$ is a {black} triangle,  the vertices are expressed by
$\{K_{\sigma}\setminus a,K_{\sigma}\setminus b,K_{\sigma}\setminus c\}$ for some $ K_{\sigma}\in\binom{[M]}{N+1}$. We then define $(K_\sigma^+,K_\sigma^-)$ as
\begin{equation}\label{def:blackK}
K_{\sigma}^+:=K_{\sigma}\setminus\{a,b,c\},\quad K_{\sigma}^-=[M]\setminus K_{\sigma}.
\end{equation}
\end{itemize}
That is,  $K_\sigma^+$ represents the common indices, and $K^-_\sigma$ represents the missing indices
for the triangle $\sigma$.  Also notice that $|K_\sigma^+|+|K_\sigma^-|=M-3$.
Then we have:
\begin{theorem}\label{cone2tile-NM}
A subdivision $\Q_{N,M}(\omega(\t))$ contains a triangle $\sigma$
if and only if 
\begin{align*}
\t~\in~{\rm relint}\left({\rm cone}\left\{{\bf r}_\alpha^+,{\bf r}_\beta^-:\alpha\in K_\sigma^+,
\beta\in K^-_{\sigma}\right\}\right).
\end{align*}
The dimension of the cone is $M-3$, i.e. the full dimension of the $\bf{t}$-space.
\end{theorem}
\begin{proof}
We here consider only the white triangle case with $\{K_\sigma a,K_\sigma b,K_\sigma c\}$ 
and $K_\sigma=K^+_\sigma$ (the other case is similar). Recall that the time $\t\in{\rm cone}\{\r_i^-\}$ corresponds to a subdivision $\Q_{N,M}(\omega(\t))$
having the weight $\omega_k=-\delta_{i,k}$ for $k=1,\ldots,M$.
Similarly, the time $\t\in{\rm cone}\{\r_i^+\}$ with $\r_i^+=-\r_i^-$ implies the weight  $\omega_k=+\delta_{i,k}$ for $k=1,\ldots,M$.

First assume that $\t\in{\rm relint(cone}\{{\bf r}_\alpha^+,{\bf r}_\beta^-:\alpha\in K_\sigma^+,\,
\beta\in K^-_{\sigma}\})$.
This means that the vertices $\{\hat \p_{K_{\sigma}a},\hat \p_{K_{\sigma}b},\hat \p_{K_{\sigma}c}\}$ have the same positive weight, and all other vertices in $\A_{N,M}^\omega$ have smaller weights. Thus the subdivision contains the triangle $\{K_\sigma a,K_\sigma b,K_\sigma c\}$. 

Now assume that the points $\{\hat \p_{K\sigma a},\hat \p_{K_\sigma b},\hat \p_{K_\sigma c}\}$ form an upper $2$-face of the polytope $\P_{N,M}^\omega$.   Then, from Lemma \ref{XY-equivalence}, one can find  $(c,x_0,y_0)$ such that 
each point in $\{\hat \p_{K_\sigma a},\hat \p_{K_\sigma b},\hat \p_{K_\sigma c}\}$ has the same positive weight while all other vertices have smaller weights. This means that $\t$ is a point in

\[{\rm relint(cone}\{{\bf r}_\alpha^+,{\bf r}_\beta^-:\alpha\in K_\sigma^+,\,
\beta\in K^-_{\sigma}\}).\]
\end{proof}

Using Theorem \ref{cone2tile-NM} we now obtain the main theorem, which can be proven in the same way as in the proof of Theorem \ref{cone2tile-NM}.
\begin{theorem}\label{realizable-checking-NM}
A triangulation $\Q_{N,M}(\omega(\t))$ having the set of triangles $\triangle=\{\sigma_1,\cdots,\sigma_m\}$ is realizable if and only if the following set is not empty, i.e.
\begin{align*}
&\bigcap_{\sigma\in\triangle}{\rm relint}\left({\rm cone}\left\{{\bf r}_\alpha^+,{\bf r}_\beta^-:\alpha\in K_\sigma^+,
\beta\in K^-_{\sigma}\right\}\right)~\ne~\emptyset.
\end{align*}
If the set is empty, then the subdivision is not {realizable}.
\end{theorem}

For $N > 1$, it is sometimes useful to consider subdivisions of $\A_{N,M}^{\omega(\t)}$ modulo the triangulation of each black or white region. For a soliton subdivision $\Q_{N,M}(\omega(\t))$, we let $\Q_{N,M}(\omega(\t))/\sim w$ be the coarser subdivision obtained by merging any two white polygons that share an edge; let $Q_{N,M}(\omega)/\sim b$ be the subdivision which results from merging any two black polygons that share an edge; and let $\Q_{N,M}(\omega(\t))/\sim wb$ denote the case where we merge both black and white polygons.  We note that Algorithm \ref{InductiveAlgorithm} gives a purely combinatorial recipe for constructing $\Q_{N+1,M}(\omega(\mathbf{t}))/\sim {w}$, and hence $\Q_{N+1,M}(\omega(\mathbf{t}))/\sim wb$, from $\Q_{N,M}(\omega(\mathbf{t}))/\sim b$

\begin{definition}
For a given choice of the $\k$-parameters, we let $\widetilde{\mathcal{F}}_{N,M}$ denote the polyhedral fan whose maximal cones correspond to triangulations $\Q_{N,M}(\omega(\t))$.
We let $\mathcal{F}_{N,M}$ denote the polyhedral fan whose maximal cones correspond to subdivisions $\Q_{N,M}(\omega(\t))/\sim wb$, where $\Q_{N,M}(\omega(\t))$ is a triangulation.  Hence $\widetilde{\mathcal{F}}_{N,M}$ is a refinement of $\mathcal{F}_{N,M}$.
\end{definition}

%%%%%%%%%%%%%%%%%%%%%%%%
\section{Triangulations $\Q_{N,6}(\omega(\t))$ for $N=1,2,3$}
\label{sec: Gr36}
In this section, we construct the triangulations $\Q_{N,6}(\omega(\t))$ for $N=1,2$ and $3$ by
giving the detailed structure of the corresponding polyhedral fans in multi-time space of the KP hierarchy.  We show that all subdivision $\Q_{N,6}$ are realizable, up to triangulation of the black and white tiles.  However, some subdivisions $\Q_{3,6}$ are only realizable for certain choices of $\kappa$-parameters.  There is no fixed choice of the $\kappa$-parameters for which all subdivisions $\Q_{3,6}$ are realizable.

\subsection{Subdivisions $\Q_{1,6}(\omega(\t))$ and $\Q_{2,6}(\omega(\t))$}
We first construct $\widetilde{\mathcal{F}}_{1,6}$, which has {six} main rays 
\[\{\r_i^{-} : 1 \leq i \leq 6\}\]
in $\t$-space. Figure \ref{fig:Gr16-T} shows a schematic drawing of this fan, projected onto a region in two-dimensional space.  Each white vertex represents one of the main rays $\r_i^{-}$.  
Any two of the rays $\r_i^{-}$ span a two-dimensional cone, shown in the figure as a dashed line segment, and any three of the rays $\r_i^{-}$ span a three-dimensional cone.  Taking the common refinement of this collection of cones gives a polyhedral fan, where each full-dimensional cone can be labeled by a unique triangulation of the hexagon (Theorem \ref{realizable-checking-NM}). For example, the cone associated to the unique triangulation having triangles $\{1,2,3\}, \{1,3,6\},\{3,5,6\}$ and $\{3,4,5\}$ corresponds to the intersection of two cones,
\[
 {\rm cone}\{\r_2^-,\r_4^-,\r_5^-\}\qquad \text{and}\qquad {\rm cone}\{\r_1^-,\r_2^-,\r_4^-\}.
 \]
 The first cone corresponds to the triangle $\{1,3,6\}$ and the second one to $\{3,5,6\}$.
Note here that we only need to use a minimal number of triangles which determine the triangulation.
In particular, the cone for the case with the triangle $\{2,4,6\}$ is just ${\rm cone}\{\r_1^-,\r_3^-,\r_5^-\}$ (the middle triangular cone in Fig.~\ref{fig:Gr16-T}).
%%%%%%%%%%%%%%%%%%%%%%
\begin{figure}[h]
\begin{center}
\includegraphics[height=5cm]{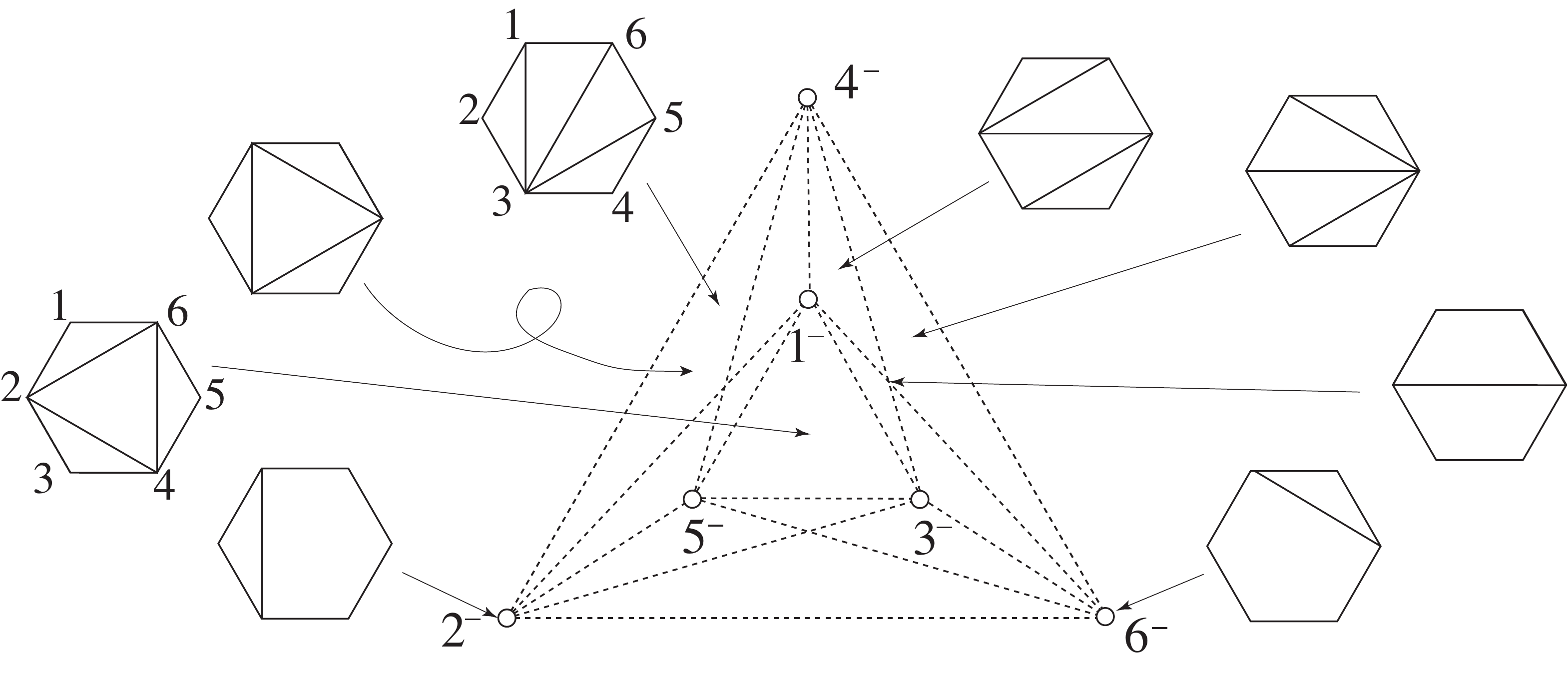}
\caption{Polyhedral cones in the time space $\R^3$. The Gale vectors $\r_i^-$
are marked by the white vertices with $i^-$ for $i=1,\ldots,6$.
In particular, the triangulation with triangle $\{2,4,6\}$ is generated by $\t\in{\rm relint}({\rm cone}\{\r_1^-,\r_3^-,\r_5^-\})$.}
\label{fig:Gr16-T}
\end{center}
\end{figure}
%%%%%%%%%%%%%%%%%%%%%%%%%%%%%
\begin{remark}
It is well-known that the number of triangulations of the $M$-gon is given by the Catalan number,
\[
C_{M-2}=\frac{1}{M-1}\binom{2(M-2)}{M-2}.
\]
Notice that the number of polyhedral cones in Fig.~\ref{fig:Gr16-T} is $C_4=14$.
The secondary polytope which is dual to the polyhedral structure of $\R^{M-3}$ is known as
the associahedron, whose vertices are labeled by the triangulations of the $M$-gon.
\end{remark}

It follows from Algorithm \ref{InductiveAlgorithm} that $\mathcal{F}_{2,6}$ has precisely the same cones as $\widetilde{\mathcal{F}}_{1,6}$.  This is illustrated in Fig.~\ref{fig:Gr26-T}.  We now refine $\mathcal{F}_{2,6}$ to produce $\widetilde{F}_{2,6}$.   First, we 
construct the rays 
\[\{\r_i^{+} : 1 \leq i \leq 6\},\]
represented by black dots in the middle panel of Fig.~\ref{fig:GrN6-Tspace}.  (The black dot inside a white circle does not represent one of the main rays, and will be explained below.)
We then construct the two-dimensional cones ${\rm cone}\{\mathbf{r}_i^-, \mathbf{r}_j^+\}$ for all $i \neq j$, represented by dashed segments in the figure. Taking the common refinement of the resulting collection of cones with the cones of $\widetilde{\mathcal{F}}_{1,6}$, we obtain $\widetilde{\mathcal{F}}_{2,6}$. As explained below, the structure of $\widetilde{\mathcal{F}}_{2,6}$ (and hence, the collection of the subdivisions for $\Gr(2,6)_{>0}$ that are realizable) depends on our choice of $\k$-parameters. 
% Figure \ref{fig:GrN6-Tspace} shows the case \[(\k_1,\ldots,\k_6) =  (-4,-2,-1,1,2,4).\]

By construction, maximal cones of $\widetilde{\mathcal{F}}_{2,6}$ correspond to the triangulations of $\mathcal{A}_{2,6}$.  For example, the unique such triangulation having two black triangles $\{\{1,2\},\{1,5\},\{2,5\}\}$ and $\{\{2,3\},\{2,5\},\{3,5\}\}$ and
a white triangle $\{ \{1,5\},\{3,5\},\{5,6\}\}$ can be realized by taking a point $\t=(t_3,t_4,t_5)$ in the intersection
of three cones spanned by $\{\r_3^-,\r_4^-,\r_6^-\}, \{\r_1^-,\r_4^-,\r_6^-\}$ and $\{\r_5^+,\r_1^-,\r_4^-\}$ (Theorem \ref{realizable-checking-NM}).

%%%%%%%%%%%%%%%%%%%%%%
\begin{figure}[h]
\begin{center}
\includegraphics[height=5cm]{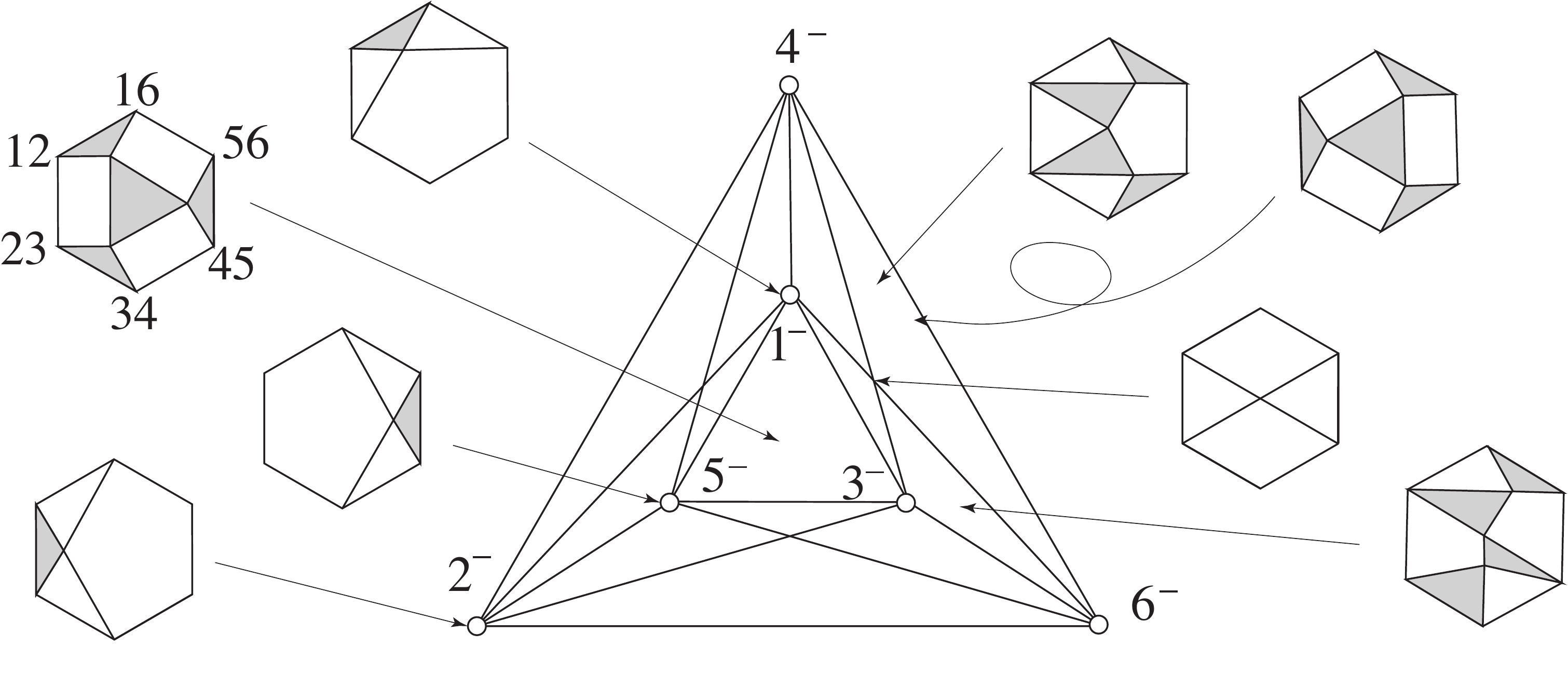}
\caption{Subdivisions $\Q_{2,6}(\omega(\t))$ via Algorithm \ref{InductiveAlgorithm} from 
the triangulations $\Q_{1,6}(\omega(\t))$. Each black triangle is induced by a white triangle in $\Q_{1,6}(\omega(\t))$.}
\label{fig:Gr26-T}
\end{center}
\end{figure}
%%%%%%%%%%%%%%%%%%%%%%%%%%%%%

%%%%%%%%%%%%%
\subsection{Subdivision $\Q_{3,6}(\omega(\t))$}
We now construct $\mathcal{F}_{3,6}$.   
Using Algorithm \ref{InductiveAlgorithm}, this is obtained by taking $\widetilde{\mathcal{F}}_{3,6}$, and merging any top-dimensional cones that represent the same triangulations, 
up to triangulation of the black polygons. For this, we claim it suffices to merge each pair of three-dimensional cones in $\widetilde{\mathcal{F}}_{2,6}$ separated by a face of the form ${\rm cone}\{\r_i^+, \r_j^-\}$.  In our example, this yields the fan shown at right in Fig.~\ref{fig:GrN6-Tspace}.  To prove the claim, note that two full-dimensional cones are separated by a face ${\rm cone}\{\r_i^+,\r_j^-\}$ if and only if the corresponding triangulations differ by a black-white flip, which occurs if and only if their images at $N=3$ are identical, up to flipping a diagonal in one of the triangulated black polygons.

%%%%%%%%%%%%%%%%%%%%%%%%%%%%
\begin{figure}[h]
\begin{center}
\includegraphics[height=5cm]{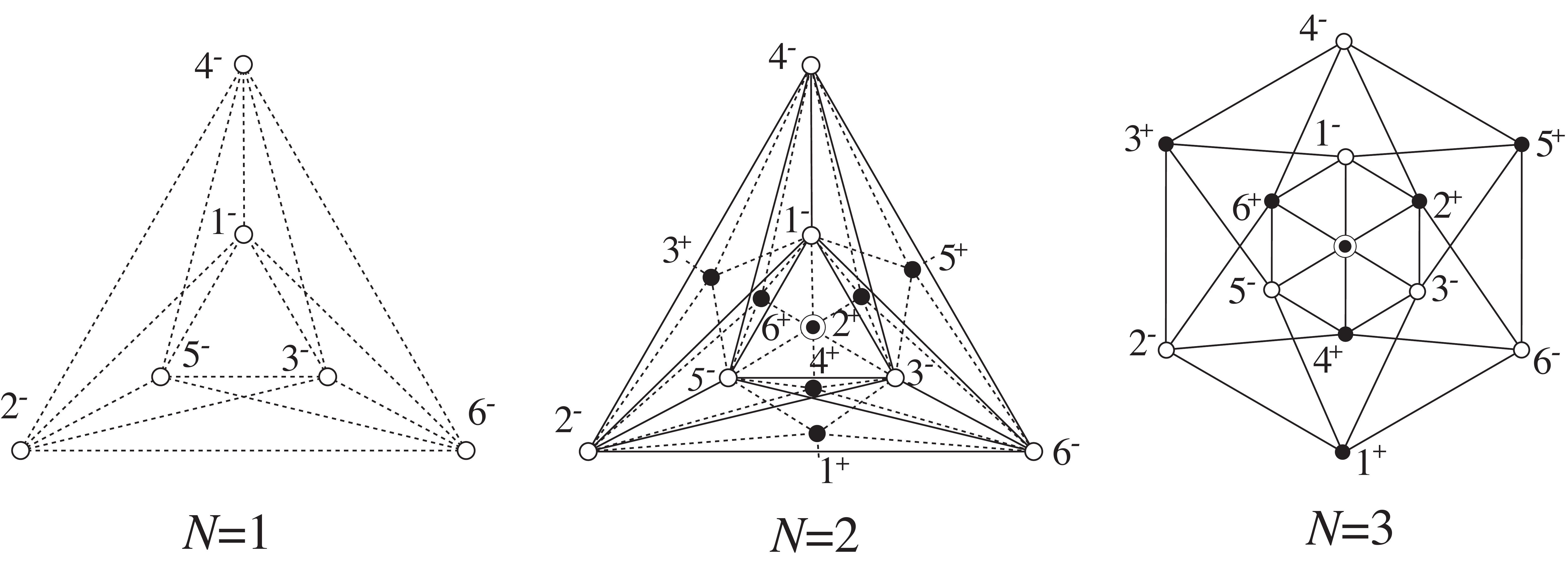}
\end{center}
\caption{The polyhedral structures $\mathcal{F}_{N,6}$ in the time space $\R^3$ for $N=1,2,3$.
The dashed lines correspond to the flips of diagonals in the white polygons, and the solid lines 
correspond to the black-white flips in the parallelograms (see Fig.~\ref{fig:Gr24}).  Note that the dashed lines become
solid lines then the solid lines disappear when $N$ increases.  In the case $N=3$, we omit the dashed lines which correspond to the triangulations of the white polygons in the subdivisions
$\Q_{3,6}(\omega(\t))$ obtained by the blow-up of $\Q_{2,6}(\omega(\t))$.}
\label{fig:GrN6-Tspace}
\end{figure}
%%%%%%%%%%%%%%%%%%%%%%%%%%%%%

%\begin{remark}\label{rem:DoubleP}
We note that there are four triangulations $\Q_{2,6}(\omega(\t))$ which can only be realized for certain $\kappa$-parameters.  Two of these are refinements of the subdivision
obtained inside the ${\rm cone}\{\r_1^-,\r_3^-,\r_5^-\}$, which has four white 4-gons as shown in Fig.~\ref{fig:Gr26-T}; two are refinements of the analogous subdivision which occurs within ${\rm cone}\{\r_2^-, \r_4^-, \r_6^-\}$.  We triangulate the subdivision in the following two cases, shown in Fig.~\ref{bad6}:

%%%%%%%%%%%%%%%%%%%%
\begin{figure}[h]
\centering
\includegraphics[trim = {2in 8.5in 2in 1in}, clip]{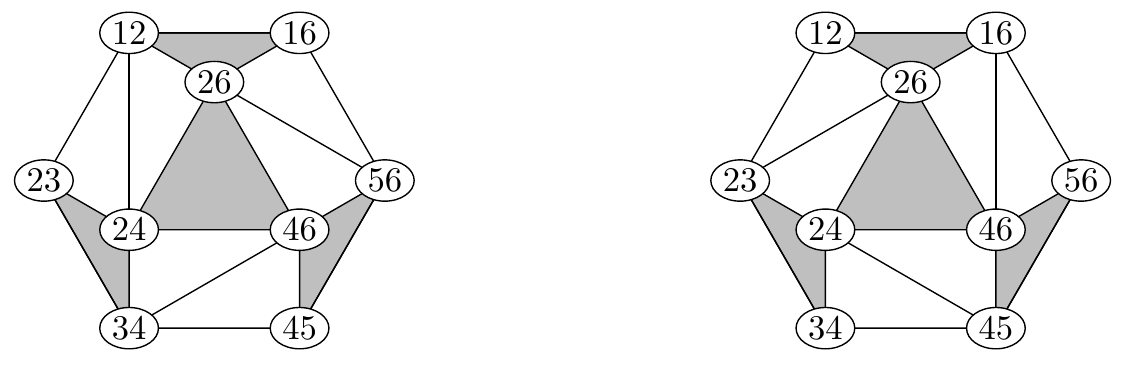}
\caption{Triangulations $\Q_{2,6}(\omega(\t))$ which are only realizable for some choices of the $\kappa$-parameters. Each of these triangulations is realized by choosing a point in the middle triangular cone in the time space $\R^3$ shown in Fig.~\ref{fig:Gr26-T}.}
\label{bad6}
\end{figure}
%%%%%%%%%%%%%%%%%%%

\begin{itemize}
\item[(a)] 
The triangulation with the three white triangles,
\[
\{ 12,24,26\},\qquad \{24, 34,46\},\qquad \{26,46,56\}.
\]
They can be realized  from the following cones, respectively,
\[
{\rm cone}\{\r_2^+,\r_3^-,\r_5^-\},\qquad {\rm cone}\{\r_4^+,\r_1^-,\r_5^-\},\qquad
{\rm cone}\{\r_6^+,\r_1^-,\r_3^-\}.
\]

\item[(b)] 
The triangulation with three white triangles,
\[
\{ 23,24,26\},\qquad \{24, 45, 46\},\qquad \{16,26,46\}.
\]
They are realized from the following cones:

\[
{\rm cone}\{\r_2^+,\r_1^-,\r_5^-\},\qquad {\rm cone}\{\r_4^+,\r_1^-,\r_3^-\},\qquad
{\rm cone}\{\r_6^+,\r_3^-,\r_5^-\}.
\]
\end{itemize}

By Theorem \ref{realizable-checking-NM}, these triangulations can be realized by a point $\t$ in the intersection of the given cones (if it is not empty).  For a choice of $\kappa$-parameters, however, both intersections can be empty. 
This occurs when the segments representing ${\rm cone}\{\r_1^-, \r_4^+\},$ ${\rm cone}\{\r_3^-, \r_4^+\}$ and ${\rm cone}\{\r_5^-, \r_2^+\}$ intersect in a single point (representing a ray in $\t$-space). We denote this \emph{double point} of $\widetilde{\FF}_{3,6}$ by a black dot inside a white circle
as shown in Fig.~\ref{fig:GrN6-Tspace}.

We now determine for which choices of the $\kappa$-parameters such a double point occurs. First, we claim that the two-dimensional fan spanned by ${\rm cone}\{\r_1^-,\r_4^+\}$ is contained within the plane defined by $D_{2,3,5,6}= 0$. To see this, note that the plane defined by $D_{2,3,5,6}$ is the region in $\mathbf{t}$-space corresponding to the point configurations where $\pp_2$, $\pp_3$, $\pp_5$ and $\pp_6$ are coplanar.  Certainly, any point in $\t$-space which is a linear combination of $\r_1^-$ and $\r_4^+$ satisfies this condition; for such a point $\p_2$, $\p_3$, $\p_5$ and $\p_6$ all have weight $0$.  

We may rewrite the equation for the plane $D_{2,3,5,6} = 0$ by plugging in the coordinates
\[p_i = \kappa_j,\qquad q_i = \kappa_i^2,\qquad \omega_i (\t)= \sum_{j = 3}^{5}\kappa_i^jt_j\]
into the determinant formula for $D_{2,3,5,6}$.
Factoring the resulting equation, and dividing by terms that cannot equal zero when $\kappa_1 < \kappa_2 < \cdots < \kappa_6,$
we obtain
$$t_3 + h_1(2,3,5,6) t_4 + h_2(2,3,5,6) t_5 = 0$$
where $h_k$ is the \emph{homogeneous symmetric polynomial of degree k} defined by
\[h_k(i_1,i_2,i_3,i_4) = \sum_{1 \leq s_1 \leq s_2 \leq \cdots \leq  s_k \leq 4} \kappa_{i_{s_1}}\cdots \kappa_{i_{s_k}}.\]

Similarly, the plane defined by $D_{1,2,4,5} = 0$ contains ${\rm cone}\{\r_3^-, \r_6^+\}$, and the plane defined by $D_{1,3,4,6} = 0$ contains ${\rm cone}\{\r_5^-,\r_2^+\}$. 
Hence the three two-dimensional cones intersect in a ray precisely when the three planes intersect in a line, that is, when we have the following determinant condition,
\begin{equation}
\label{zero}
\left|
\begin{matrix}
1  &  h_1(2,3,5,6)  & h_2(2,3,5,6)  \\
1  &  h_1(1,2,4,5)  & h_2(1,2,4,5) \\
1  &  h_1(1,3,4,6)  & h_2(1,3,4,6) 
\end{matrix} \right|
=0.
\end{equation}

%Expanding the determinant, we obtain 
%\[-\kappa_1\kappa_2\kappa_4 + \kappa_1\kappa_3\kappa_4 + \kappa_1\kappa_2\kappa_5 - \kappa_2\kappa_3\kappa_5 - \kappa_1\kappa_4\kappa_5 + \kappa_2\kappa_4\kappa_5 - \kappa_1\kappa_3\kappa_6 + \kappa_2\kappa_3\kappa_6 + \kappa_1\kappa_4\kappa_6 - \kappa_3\kappa_4\kappa_6 -  \kappa_2\kappa_5\kappa_6 + \kappa_3\kappa_5\kappa_6.\]

To obtain a simpler formula, we may specialize to the case where the $\k$-parameters satisfy a symmetric condition,
\begin{align*}
\kappa_1 =-\kappa_6,\qquad
\kappa_2 =-\kappa_5,\qquad
\kappa_3 =-\kappa_4.
\end{align*}
With this choice of parameters, the determinant \eqref{zero} becomes
\[-2(\kappa_2^2 - \kappa_1\kappa_3)(\kappa_1 - \kappa_3).\]
Since $\k_1 < \k_3$, 
the determinant is positive if $\kappa_2^2 >  \kappa_1\kappa_3$, and negative if $\kappa_2^2 < \kappa_1\kappa_3$.

We now investigate what happens when the determinant $\eqref{zero}$ is nonzero.  
By the quadrilateral-checking formula (Lemma \ref{lem:QCL}), $D_{2,3,5,6} > 0$ when the diagonal $\{\hat{3},\hat{6}\}$ passes over the diagonal $\{\hat{2},\hat{5}\}$, while $D_{2,3,5,6} < 0$ when $\{\hat{2},\hat{5}\}$ passes over $\{\hat{3},\hat{6}\}$.  
Hence, the normal vector 
\[\langle 1, h_1(2,3,5,6), h_2(2,3,5,6)\rangle\] to the plane containing ${\rm cone}\{\r_1^-, \r_4^+\}$ points toward the half-space containing $\r_5^-$ and $\r_6^+$.
Similarly, the normal vector 
\[\langle 1 ,  h_1(1,2,4,5) , h_2(1,2,4,5)\rangle\]
to the plane defined by $D_{1,2,4,5} = 0$ points toward the half-space containing $\r_1^{-}$ and $\r_2^+$, while the
normal vector 
\[\langle 1,  h_1(1,3,4,6), h_2(1,3,4,6)\rangle\] to the plane defined by $D_{1,3,4,6} = 0$ points toward the half-space containing $\r_1^{-}$ and $\r_6^{+}$.  See Fig.~\ref{directions}.

%%%%%%%%%%%%%%%%%%%%%%%%%%%%
\begin{figure}[h]
\begin{center}
\includegraphics[scale = 0.7, trim = {1in 7.25in 1in 1in}, clip]{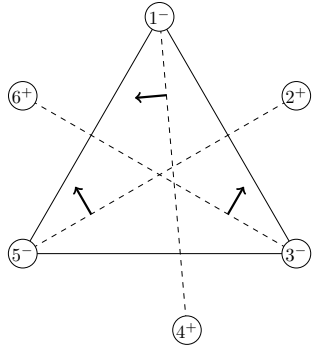}
\end{center}
\caption{One possibility for the arrangement of two-dimensional cones inside ${\rm cone}\{\r_1^-, \r_3^-, \r_5^-\}$.  Each Gale vector $\r_i^{\pm}$ is marked by $i^{\pm}$. The arrows show the normal vectors corresponding to rows of the matrix in the determinant \eqref{zero}.}
\label{directions}
\end{figure}
%%%%%%%%%%%%%%%%%%%%%%%%%%%%%

Note that the ray where the planes $D_{2,3,5,6} = 0$ and $D_{1,2, 4,5} = 0$ intersect inside ${\rm cone}\{\r_1^-, \r_3^-, \r_5^-\}$ is the cross product
\[\langle 1, h_1(1,2,4,5), h_1(1,2,4,5)\rangle \times \langle1, h_1(2,3,5,6), h_1(2,3,5,6)\rangle.\]
The triple scalar product of this ray with 
\[\langle1, h_1(1,3,4,6), h_2(1,3,4,6)\rangle\]
is positive if the ray lies on the same side of ${\rm cone}\{\r_5^-, \r_2^+\}$ as $\r_1^-$ and $\r_6^+$, as shown on the right in Fig.~\ref{fig:DoubleP}; and negative if the ray lies on the opposite side of the cone as shown at left in Fig.~\ref{fig:DoubleP}.
By properties of the triple scalar product, the determinant \eqref{zero}, is negative in the first case, and positive in the second.

We now use Algorithm \ref{InductiveAlgorithm} to construct a subdivision $\Q_{3,6}(\omega(\t))$ from a triangulation $\Q_{2,6}(\omega(\t))$. Note here that two triangulations $\Q_{2,6}$ adjacent to a common solid line in Fig.~\ref{fig:GrN6-Tspace}
lead to the same subdivision $\Q_{3,6}$ by the blow-up process. That is, the solid lines in $N=2$ case disappear in $N=3$, and each subdivision $\Q_{3,6}(\omega(\t))$ can be generated by
choosing a point $\t$ in a cone illustrated in the figure $N=3$ where the solid lines are the dashed lines in the case $N=2$.
Figure \ref{fig:Gr36-T} shows the subdivisions
$\Q_{3,6}$ obtained from the triangulations $\Q_{2,6}$ through
Algorithm \ref{InductiveAlgorithm}.  The total number of the subdivisions is given by the number of
polyhedral cones in the time space. Then recall that there are two triangulations
of $\Q_{2,6}$ which cannot be realized for fixed $\k$-parameters. 
Using a different set of $\k$-parameters, we obtained those missing triangulations $\Q_{2,6}$
as shown in Fig.~\ref{bad6}. Then it is immediate to see that the subdivisions $\Q_{3,6}$ obtained from these triangulations through the blow-up process are triangulations which cannot be obtained from the original set of the $\k$-parameters. 
Figure \ref{fig:DoubleP} shows these triangulations for the middle section in Fig.~\ref{fig:Gr36-T}.

%%%%%%%%%%%%%%%%%%%%%%%%%%%%
\begin{figure}[h]
\begin{center}
\includegraphics[height=5.5cm]{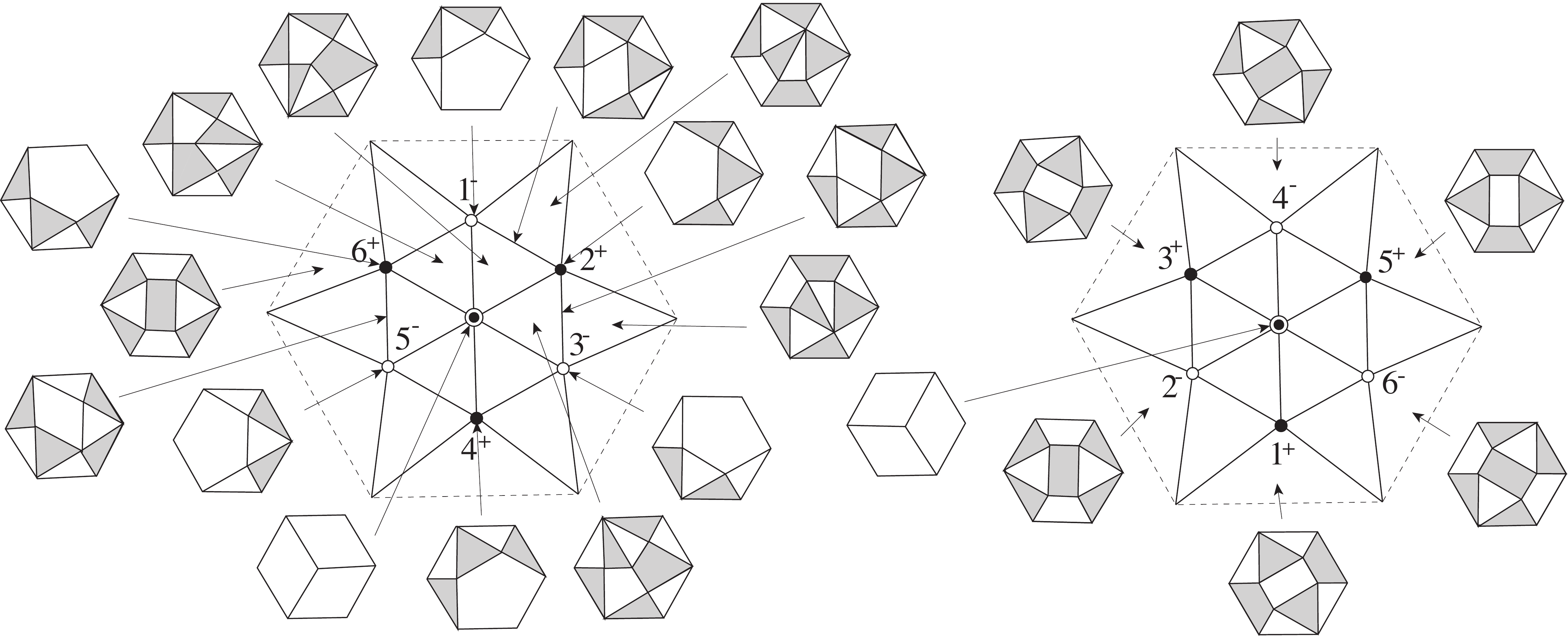}
\end{center}
\caption{Subdivisions $\Q_{3,6}(\omega(\t))$ and the corresponding polyhedral cones in the time space $\R^3$ shown in Fig.~\ref{fig:GrN6-Tspace}.}
\label{fig:Gr36-T}
\end{figure}
%%%%%%%%%%%%%%%%%%%%%%%%%%%%%
%%%%%%%%%%%%%%%%%%%%%%%%%%%%
\begin{figure}[h]
\begin{center}
\includegraphics[height=4.5cm]{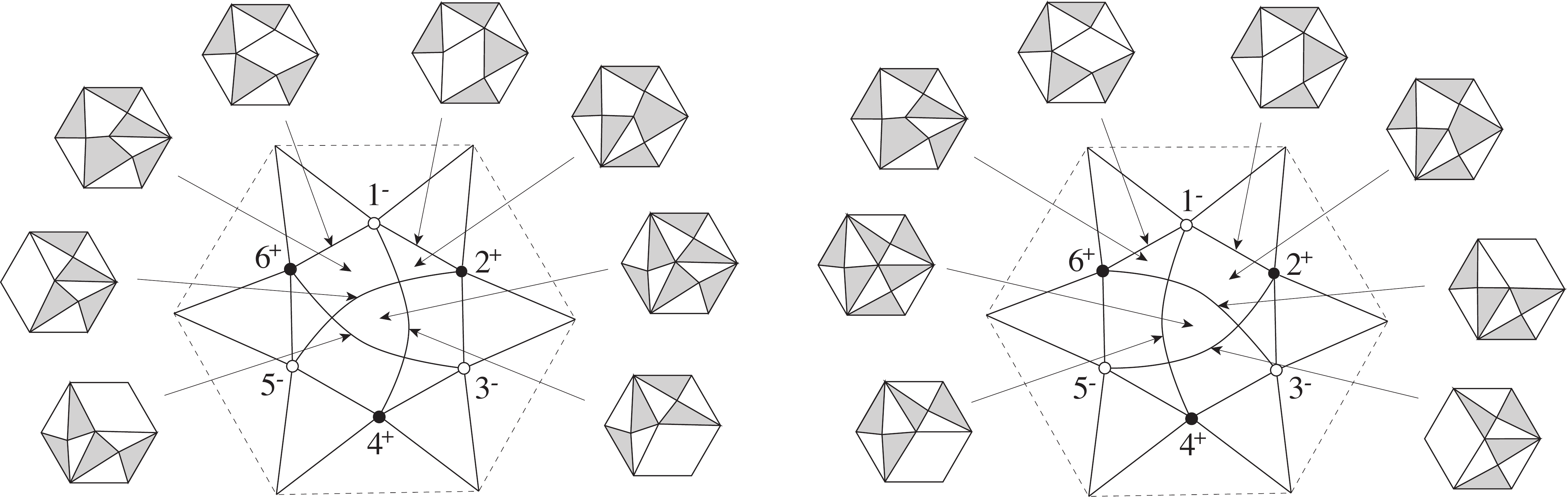}
\end{center}
\caption{Two different triangulations for the middle section of the left panel in Fig.~\ref{fig:Gr36-T}.
They are corresponding to the different choices of the $\kappa$-parameters.
For the symmetric parameters $(\k_1,\k_2,\k_3,-\k_3,-\k_2,-\k_1)$, the left figure corresponds to 
the case $\k_2^2>\k_1\k_2$, and the right one to $\k_2^2<\k_1\k_3$.}
\label{fig:DoubleP}
\end{figure}
%%%%%%%%%%%%%%%%%%%%%%%%%%%%%

%%%%%%%%%%%%%%%%%%
\begin{figure}[h]
\centering
\includegraphics[scale = 0.8, trim = {1in 6in 1in 1in}, clip]{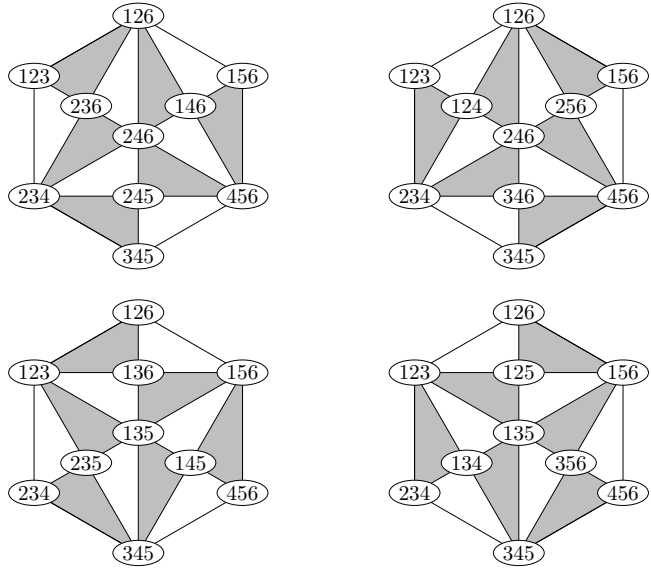}
\caption{Triangulated plabic tilings $\Q_{3,6}(\omega(\t))$ which are only realizable for some choices of the $\kappa$-parameters. There triangulations are realized by choosing points in the central polyhedral cones in Fig.~\ref{fig:DoubleP}.}
\label{bad36}
\end{figure}
%%%%%%%%%%%%%%

We now summarize this discussion in the following theorem, which states when each maximal weakly separated collection for $\Gr(3,6)_{>0}$ is realizable. See also Figs.~\ref{fig:Gr36-T} and \ref{fig:DoubleP}.
\begin{theorem}\label{classify36}
There are 34 maximally weakly separated collections for $\text{Gr}(3,6)_{>0}$.  Of these, $30$ are realizable for every choice of $\kappa$-parameters.  
For a generic choice of $\kappa$-parameters, $32$ of the $34$ are realizable.  We can realize the weakly separated collections shown at left in Fig.~\ref{bad36} if and only if the determinant \eqref{zero} is positive.  We can realize the collections shown at right in Fig.~\ref{bad36} if and only if the determinant \eqref{zero} is negative.
\end{theorem}

%\end{remark}

\begin{example}
We demonstrate the case for $\Gr(3,6)_{>0}$ by considering an explicit example where we take
 the $\kappa$-parameter as $(\k_1,\ldots,\k_6)=(-3,-2,-1,1,2,3)$.  
Then the the Gale vectors are calculated as 
\begin{align*}
\r_1^-&=\begin{pmatrix} -5\\-3\\1\end{pmatrix},\quad
\r_2^-=\begin{pmatrix} 10\\2\\-1\end{pmatrix},\quad
\r_3^-=\begin{pmatrix} -13\\-1\\1\end{pmatrix},\\
\r_4^-&=\begin{pmatrix} 13\\-1\\-1\end{pmatrix},\quad
\r_5^-=\begin{pmatrix} -10\\2\\1\end{pmatrix},\quad
\r_6^-=\begin{pmatrix} 5\\-3\\-1\end{pmatrix}.
\end{align*}
Here these vectors are normalized to be $\pm1$ in the third component.
Figure \ref{fig:Gr36-Tspace} illustrates the polyhedral cones in the time space $\t=(t_3,t_4,t_5)$.
Note that the vectors $\r_1^-, \r_3^-,\r_5^-$ appear at the plane $t_5=1$, and other vectors are at $t_5=-1$.
%%%%%%%%%%%%%%%%%%%%%%%%%%%%
\begin{figure}[h]
\begin{center}
\includegraphics[height=5cm]{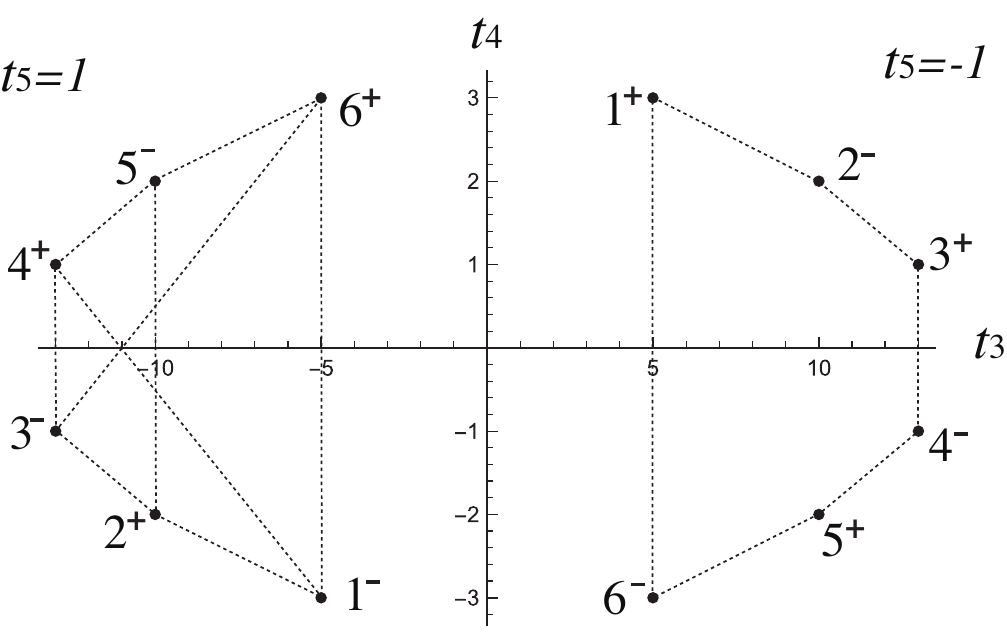}
\end{center}
\caption{The Gale vectors for $(\k_1,\ldots,\k_6)=(-3,-2,-1,1,2,3)$ in the time space.
The vectors in the left are shown at the plane $t_5=1$, and the right ones are at $t_5=-1$.}
\label{fig:Gr36-Tspace}
\end{figure}
%%%%%%%%%%%%%%%%%%%%%%%%%%%%%

%%%%%%%%%%%%%%%%%%%%%%%%%%%%
\begin{figure}[h]
\begin{center}
\includegraphics[height=5cm]{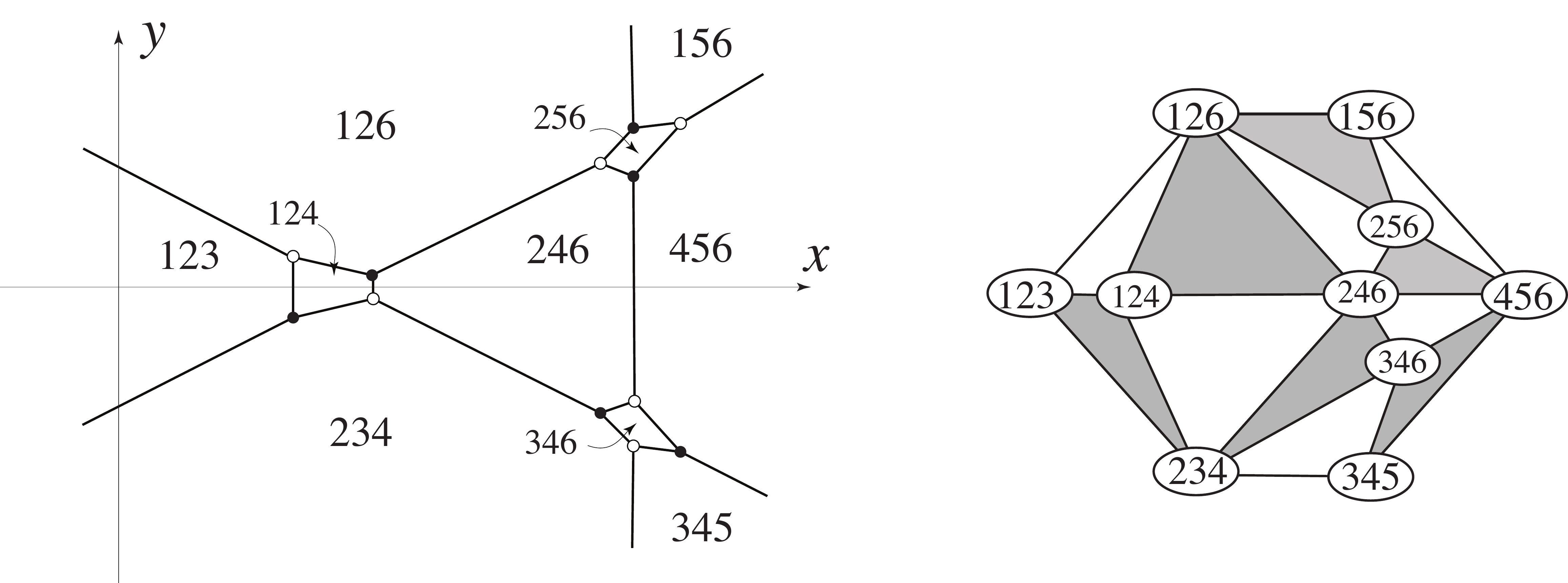}
\end{center}
\caption{The soliton graph and the corresponding triangulation $\Q_{3,6}(\t)$ for the case with
$(\k_1,\ldots,\k_6)=(-3,-2,-1,1,2,3)$ and $\t=(-10.5,0,1)$. }
\label{fig:SGDoubleP}
\end{figure}
%%%%%%%%%%%%%%%%%%%%%%%%%%%%%
As a summary of this section, we show how to find values of $\t$ that realize the subdivision 
 $\Q_{2,6}(\omega(\t))$ shown in the center of Fig.~\ref{fig:Gr26-T} and the triangulation $\Q_{3,6}(\omega(\t))$ shown in Fig.~\ref{fig:SGDoubleP}:
\begin{itemize}
\item[(a)] Consider the triangulation $\Q_{1,6}(\omega(\t))$ with the triangle $\{2,4,6\}$, which
 is realized by taking 
\[
\t~\in ~{\rm relint}\left(\mathcal{C}_{1,3,5}^{---}\right)\qquad\text{with}\qquad \mathcal{C}_{1,3,5}^{---}:={\rm cone}\left\{\r_1^-,\r_3^-,\r_5^-\right\}.
\]
This leads to the subdivision $\Q_{2,6}(\omega(\t))$ shown in Fig.~\ref{fig:Gr26-T}.
\item[(b)] Triangulate three white 4-gons in $\Q_{2,6}(\omega(\t))$ by taking the intersection of three cones,
$\mathcal{C}_{2,3,5}^{+--}, \mathcal{C}_{1,4,5}^{-+-}$ and $\mathcal{C}_{1,3,6}^{--+}$. Then the triangulation $\Q_{3,6}(\omega(\t))$ shown in Fig.~\ref{fig:SGDoubleP} is obtained by taking a point
\[
\t~\in~{\rm relint}\left(\mathcal{C}_{2,3,5}^{+--}\cap\mathcal{C}_{1,4,5}^{-+-}\cap\mathcal{C}_{1,3,6}^{--+}\right).
\]
The triangulation $\Q_{3,6}(\omega(\t))$ in Fig.~\ref{fig:SGDoubleP} is obtained by taking $\t=(-10.5,0,1)$.
\end{itemize}
\end{example}

%%%%%%%%%%%%%%%%%%%%%%%%%%%%%%%%%%%%%%%

\section{Realizability of $\Q_{3,7}(\omega(\t))$ and $\Q_{3,8}(\omega(\t))$}
\label{sec: realizability}
\subsection{Results for $\Gr(3,7)_{>0}$}
\label{subsec: 37}
We now extend our results from $\Gr(3,6)_{>0}$ to $\Gr(3,7)_{>0}$.  We show that every maximal weakly separated collection for $\Gr(3,7)_{>0}$ is realizable for \emph{some} choice of the $\k$-parameters $(\kappa_1,\ldots,\kappa_7)$, and determine which of these collections are realizable for any \emph{given} choice of the $\kappa$-parameters.

\begin{theorem}
\label{realize37}
Every maximal weakly separated collection for $\Gr(3,7)_{>0}$ is realizable for some choice of the $\k$-parameters.
\end{theorem}

\begin{proof}
Since the $\t$-space has dimension $4$,
visualizing the polyhedral fan for $\mathcal{A}_{1,7}$ is rather difficult.
We instead reason directly about the placement of the lifted points
\[\{\pp_1,\pp_2,\ldots,\pp_7\}~\subset~\R^3.\]

By Lemma \ref{tilings}, it suffices to prove that every \emph{triangulated} plabic tiling for $\Gr(2,7)_{>0}$ is realizable for some choice of the $\k$-parameters.  We have already shown the analogous result for $\Gr(3,6)_{>0}$.  Hence our approach is to start with a weight function on $\{\p_1,\ldots,\p_6\}$, and show that we can add a seventh point with an appropriate weight to realize the desired triangulation.

Consider a realizable triangulation $\Q_{2,7}(\omega(\t))$ of $\A_{2,7}$.  By Algorithm \ref{InductiveAlgorithm}, $\Q_{2,N}(\omega(\t))$ uniquely determines $\Q_{1,N}(\omega(\t))$.  Moreover, the triangulation of the white polygon whose vertices have common index $i$ in $\Q_{2,7}(\omega(\t))$ is determined by restricting $\omega(\t)$ to the set of neighbors of $\p_i$ in $\Q_{1,7}(\omega(\t))$. Hence to show that a given subdivision $\Q_{2,7}$ of $\A_{2,7}$ is realizable, it suffices to find a weight function $\omega(\t)$ such that the following hold:
\begin{enumerate}
\item $\Q_{1,7}(\omega(\t))$ is the triangulation determined by $\Q_{2,7}(\omega(\t))$.
\item For each $1 \leq i \leq 7$, restricting $\omega(\t)$ to the neighbors of $\p_i$ in the triangulation $\Q_{1,7}(\omega(\t))$ yields the appropriate triangulation.\end{enumerate}

%%%%%%%%%%%%%%%%% 
\begin{figure}[ht]
\centering
\includegraphics[trim = {1in 8.75in 1in 1in}, clip]{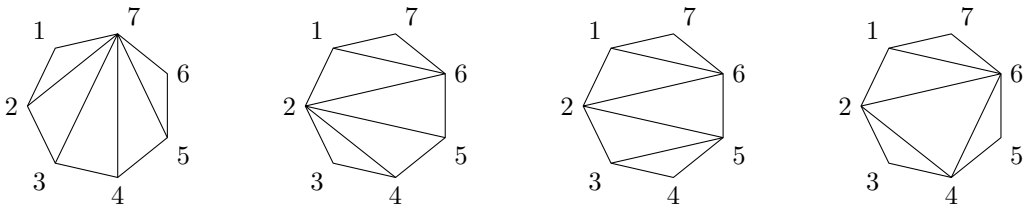}
\caption{The four triangulations $\Q_{1,7}$ of the heptagon, up to rotation and reflection.}
\label{triangulations7}
\end{figure}  
%%%%%%%%%%%%%%%%

There are four soliton triangulations $\Q_{1,7}$ of the heptagon, up to rotation and reflection.
First, consider the leftmost triangulation in Fig.~\ref{triangulations7}. 
We may assign weights $\{\omega_1,\ldots,\omega_7\}$
to produce any triangulation of $\A_{1,6}$ we desire.  Placing $\pp_7$ high enough then guarantees that all diagonals incident at $\p_7$ appear in the triangulation of $\mathcal{A}_{1,7}$.  Hence, any triangulation $\Q_{2,7}$ which is obtained by blowing up this triangulation of the heptagon is realizable.

For the middle two triangulations in Fig.~\ref{triangulations7}, we assume we have an appropriate weight function on $\p_1,\ldots,\p_6$, and then show that we can add a point $\pp_7$ to obtain the desired $\Q_{2,7}$. For this, note that our point $\pp_7$ must satisfy two constraints.  
\begin{enumerate}
\item The point $\pp_7$ lies below the plane through $\pp_1$, $\pp_2$ and $\pp_6$.
\item The line segment $\{\hat{2}, \hat{7}\}$ passes either below or above the line segment $\{\hat{1},\hat{5}\}$, depending on the desired triangulation.
\end{enumerate}
If $\{\hat{2},\hat{7}\}$ must pass below $\{\hat{1},\hat{5}\}$, this is easily achieved by placing $\pp_7$ low enough.  Otherwise, note that $\pp_5$ lies below the plane spanned by $\pp_1,\pp_2$ and $\pp_6$.  Hence we can achieve the desired configuration by taking $\pp_7$ just slightly below this plane.

Finally, we consider the rightmost triangulation in Fig.~\ref{triangulations7}.  Here, there are five possible cases for the triangulation of the white polygon with common index $6$, corresponding to possible regular triangulations of the pentagon with vertices $\p_1,\p_2,\p_4,\p_5$ and $\p_7$.   Assume we have an appropriate weight function for these points.  We must show that we can place the lifted point $\pp_7$ as needed in each case.

Suppose no diagonal incident at $\p_7$ appears in the desired triangulation of the pentagon.  Then it suffices to simply place $\pp_7$ low enough.  This covers the case where either $\p_1$ or $\p_5$ is incident at both diagonals which appear in the pentagon.

For the remaining cases, note that by our choice of weight function for $\p_1,\ldots,\p_6$, the plane $\mathcal{P}$ through $\pp_1,\pp_2$ and $\pp_4$ must pass below $\pp_6$.
Hence we obtain the desired triangulation $\Q_{1,7}(\omega(\t))$ as long as $\pp_7$ is sufficiently close to $\mathcal{P}$.

Suppose we weight the points $\p_1,\ldots,\p_6$ in such a way that $\{\hat{1},\hat{4}\}$ passes over $\{\hat{2},\hat{5}\}$.
Then placing $\pp_7$ slightly above $\mathcal{P}$ yields the case where $\p_7$ is incident at both diagonals of the pentagon which appear; and placing $\pp_7$ slightly below $\mathcal{P}$ yields the case where $\p_4$ is incident at both diagonals which appear. Finally, suppose $\{\hat{1},\hat{4}\}$ passes below $\{\hat{2},\hat{5}\}$.  Then placing $\pp_7$ slightly above $\mathcal{P}$ yields the case where $\p_2$ is incident to both diagonals.
This completes the proof.  \end{proof}

\begin{theorem}\label{classify37}
There are 203 maximal weakly separated collections in ${{[7]}\choose{3}}$ which are realizable for any choice of the $\k$-parameters.
For each generic choice of the parameters, a total of $231$ collections are realizable.
\end{theorem}

\begin{proof}
We first show that all triangulated plabic tilings $\Q_{2,7}$ are realizable for \emph{any} choice of the $\k$-parameters, except those which can obtained from the triangulations in Fig.~\ref{badtilings7} by applying symmetries of the heptagon. 

In the proof of Theorem \ref{realize37}, we realize each plabic tiling $\Q_{2,7}$ by first choosing an appropriate weight function on the points $\{\p_1,\ldots,\p_6\}$, and then adding a lifted point $\pp_7$.  None of the arguments require any restriction on the location of the (non-lifted) point $\p_7$ in the $pq$-plane, or equivalently on the value of $\kappa_7$.

Recall that any $\Q_{2,7}$ can be obtained by blowing up a unique triangulation $\Q_{1,7}$.  If $\Q_{1,7}$ is the leftmost triangulation in Fig.~\ref{triangulations7}, up to symmetries of the heptagon, the proof of Theorem \ref{realize37} simply requires that we find a weight function on $\{\p_1,\ldots,\p_6\}$ which induces the desired subdivision on $\Q_{1,6}$.  This is possible for any choice of the parameters.  

Otherwise, we must find a weight function on $\{\p_1,\ldots,\p_6\}$ which gives some specified $\Q_{2,6}$.  This is possible for any choice of $\{\kappa_1,\ldots,\kappa_6\}$, unless the desired triangulation is one of those shown in 
Fig.~\ref{bad6}, up to symmetry.  This occurs precisely when $\Q_{2,6}$ is one of those shown in Fig.~\ref{badtilings7}, up to symmetry, so the first part of the claim is proved.

%%%%%%%%%%%%%%%%%%%%
\begin{figure}[ht]
\centering
\includegraphics[scale = 0.8, trim = {1in 6in 1in 1in}, clip]{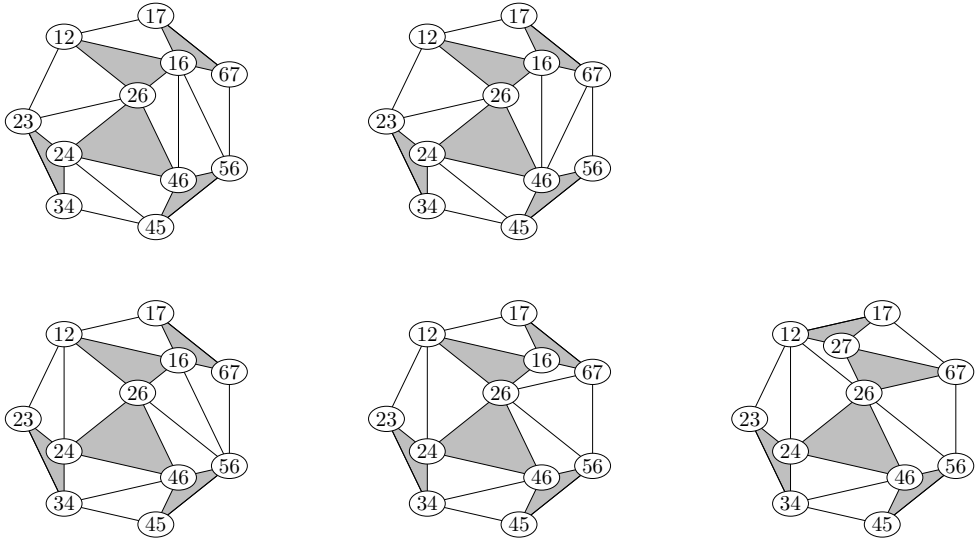}
\caption{Triangulated plabic tilings $\Q_{2,7}$ which are only realizable for some choice of parameters.}
\label{badtilings7}
\end{figure} 
%%%%%%%%%%%%%%%%%%%

We now determine when the triangulated plabic tilings in Fig.~\ref{badtilings7} are realizable.
By Theorem \ref{classify36},
the two plabic tilings on the top row of 
Fig.~\ref{badtilings7}
are not realizable unless
the determinant \eqref{zero} is negative.
Conversely, if this condition holds, then both plabic tilings are realizable, by Theorem \ref{classify36} and the proof of Theorem \ref{realize37}. Similarly, the first two plabic tilings on the bottom row of 
Fig.~\ref{badtilings7} are realizable if and only if the determinant in \eqref{zero} is positive.

The case of the tiling at bottom right in 
Fig.~\ref{badtilings7}
is more complicated.  By Theorem \ref{classify36}, this tiling cannot be realizable unless the determinant in
\eqref{zero} is positive, which suffices for our purposes. It can be shown, however, that this triangulation is realizable if and only if the $\k$-parameters satisfy the stronger condition
\[
\left|
\begin{matrix}
1  &  h_1(2,3,5,7)  & h_2(2,3,5,7)  \\
1  &  h_1(1,2,4,5)  & h_2(1,2,4,5) \\
1  &  h_1(1,3,4,7)  & h_2(1,3,4,7) 
\end{matrix} \right|> 0.\]

%%%%%%%%%%%%%%%%
\begin{figure}[ht]
\centering
\includegraphics[scale = 0.8, trim = {1in 5.75in 1in 1in}, clip]{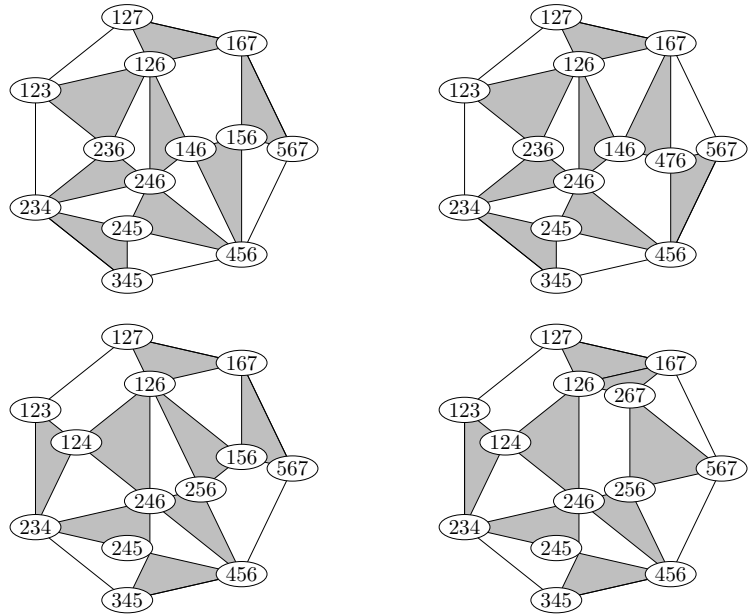}
\caption{Plabic tilings $\Q_{3,7}$ which are only realizable for some choice of the $\k$-parameters.}
\label{bad37}
\end{figure} 
%%%%%%%%%%%%%%%%%%%

Blowing up the tilings in
Fig.~\ref{badtilings7},
we obtain the (non-triangulated) plabic tiling $\Q_{3,7}$ in Fig.~\ref{bad37}.
Moreover, the plabic tilings on the top row of Fig.~\ref{bad37}
can \emph{only} be obtained by blowing up the corresponding tilings on the top row of Fig.~\ref{badtilings7};
the tiling at lower left in Fig.~\ref{bad37}
can \emph{only} be obtained by blowing up the tiling at lower left in Fig.~
 \ref{badtilings7};
and the tiling at lower right in Fig.~
\ref{bad37}
can \emph{only} be obtained by blowing up one of the tilings shown respectively at lower middle and lower right in Fig.~
 \ref{badtilings7}.
 Hence the tilings on the top row of Fig.~\ref{bad37} are realizable if and only if 
\eqref{zero} is negative
 and the tilings on the bottom row of Fig.~\ref{bad37} are realizable if and only if 
\eqref{zero} is positive.
 
In sum, for each choice of the $\k$-parameters, exactly two of the four tilings in Fig.~\ref{bad37} are realizable.  Applying the $14$ symmetries of the heptagon to the plabic tilings shown in Fig.~\ref{bad37}, we obtain a total of $56$ triangulations. Half of these, or 28 total, are realizable for any given generic choice of parameters.  There are 259 maximal weakly separated collections in ${{[7]}\choose{3}}$, so this leaves $203$ tilingss which must be realizable for any choice of the $\k$-parameters.
\end{proof}

\subsection{Results for $\Gr(3,8)_{>0}$}
\label{subsec: 38}

For $\Gr(3,8)_{>0}$, we do not yet have a classification of the possible soliton triangulations for each choice of parameters.  However, we can prove the following analog of Theorem \ref{realize37}.

\begin{theorem}\label{realize38} 
Every maximal weakly separated collection for $\Gr(3,8)_{>0}$ is realizable.
\end{theorem}

\begin{proof}
We show that every triangulated plabic tiling $\Q_{2,8}$ is realizable.  The result then follows by Lemma \ref{tilings}.
Each $\Q_{2,8}$ corresponds to a unique triangulation of the octagon $\A_{1,8}$.   The arguments used in the proof of Theorem \ref{realize37} show that $\Q_{2,8}$ is realizable in the case where the corresponding triangulation of $\A_{1,8}$ has a triangle with one vertex of degree 2, one vertex of degree 3, and one vertex of degree at most 5.  Similarly, any triangulation arising from a triangulation of $\A_{1,8}$ which has one vertex that is adjacent to all the others is realizable, by the arguments used in the proof of Theorem \ref{realize37}.

Up to symmetry, this leaves three triangulations of $\A_{1,8}$, which are shown in Fig.~\ref{triangulations8}.  We first consider triangulations $\Q_{2,8}$ arising from the leftmost triangulation in Fig.~\ref{triangulations8}. These are precisely the triangulations $\Q_{2,8}$ which refine the subdivision shown at left in Fig.~\ref{MediumTilings8}. 

%%%%%%%%%%%%%%%%%%%%%%
\begin{figure}[ht]
\centering
\includegraphics[scale = 0.9, trim = {1in 8.75in 1in 1in}, clip]{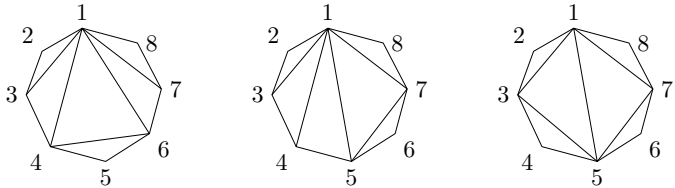}
\caption{Triangulations $\Q_{1,8}$ of the octagon.}
\label{triangulations8}
\end{figure} 
%%%%%%%%%%%%%%%%%%%
 
Fix such a $\Q_{2,8}$, and suppose the diagonal corresponding to $\{\hat{1},\hat{5}\}$ does not appear in either of the two white quadrilaterals.  Then we can realize $\Q_{2,8}$ by arranging the points $\{\pp_i : i \neq 5\}$  appropriately, and then assigning $\pp_5$ a low-enough weight.  Similarly, if the diagonal corresponding to $\{\hat{1},\hat{5}\}$ appears in both white quadrilaterals, it suffices to place $\pp_5$ high enough.

Next, suppose the diagonal corresponding to $\{\hat{1},\hat{5}\}$ appears in exactly one of the two quadrilaterals.  We consider the case where $\Q_{2,8}$ refines the subdivision shown at right in Fig.~\ref{MediumTilings8}; the other case is analogous.  For this, we choose the parameters $\kappa_i$ so that in the $pq$-plane, the segment 
$\{1,5\}$ intersects $\{3,6\}$, and $\{4,7\}$ to the \emph{right} of the point where the latter two segments intersect. 

We assign all points 
$\{\pp_i : i \neq 1,5\}$
the same weight, and assign a higher weight to $\pp_1$.  To obtain the desired subdivision, we then assign a weight to $\pp_5$ so that the segment $\{\hat{1},\hat{5}\}$ passes just slightly above $\{\hat{4},\hat{7}\}$.  We then adjust the weights of the the points
$\{\pp_i : i \neq 1, 5\}$
to obtain the desired subdivision of the white hexagon with common index $1$.  Since we can triangulate the hexagon using arbitrarily small adjustments of the weights, there is no danger of disturbing the rest of the configuration, and this case is complete.  The argument for a tiling $\Q_{2,8}$ corresponding to the middle triangulation in Fig.~\ref{triangulations8} is analogous.

%%%%%%%%%%%%%%%%%%%%
\begin{figure}[ht]
\centering
\includegraphics[scale = 0.9, trim = {1in 8.25in 1in 1in}, clip]{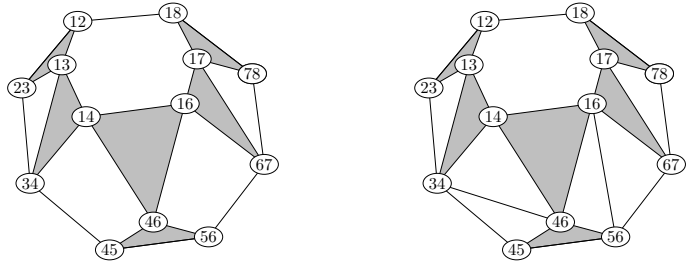}
\caption{Plabic tilings $\Q_{2,8}$ which arise from the first triangulation in Fig.~\ref{triangulations8}.}
\label{MediumTilings8}
\end{figure} 
%%%%%%%%%%%%%%%%%%%%%

It remains to show that we can realize all triangulated plabic tilings $\Q_{2,8}$ which arise from the rightmost triangulation in Fig.~\ref{triangulations8}, up to rotation and reflection.  
First, note that every such $\Q_{2,8}$ refines one of the six plabic tilings shown in Fig.~\ref{BadTilings8}, up to rotation and reflection.  (This is not immediately obvious, but follows by a simple case check.) 

If $\Q_{2,8}$ refines the tiling shown at upper left, it suffices to arrange the points $\{\pp_i : i \neq 2\}$, and then place $\pp_2$ low enough.  

%%%%%%%%%%%%%%%%%%%%%
\begin{figure}[ht]
\centering
\includegraphics[scale = 0.9, trim = {1in 6.5in 1in 1in}, clip]{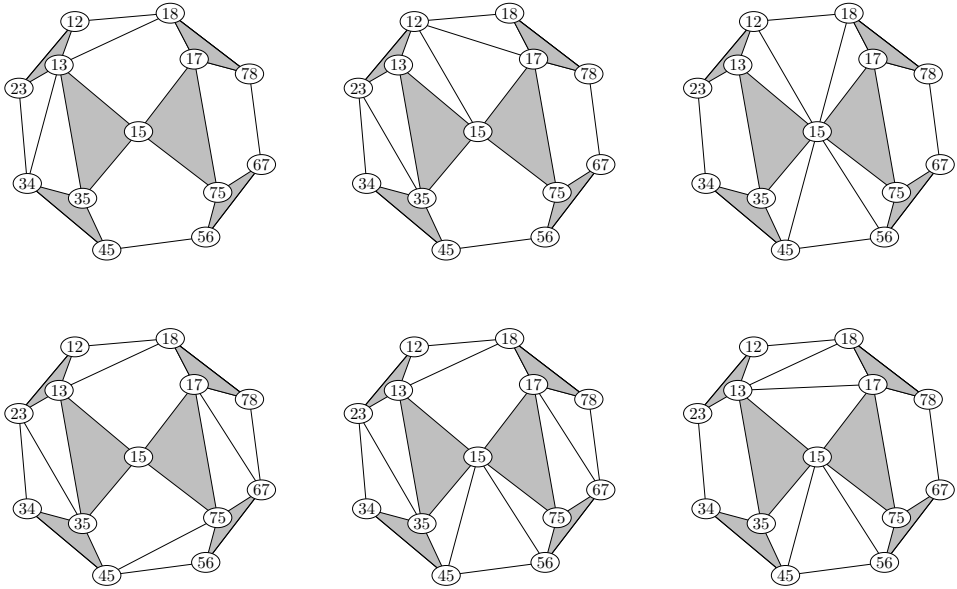}
\caption{Plabic tilings $\Q_{2,8}$ which arise from the third triangulation in Fig.~\ref{triangulations8}.}
\label{BadTilings8}
\end{figure} 
%%%%%%%%%%%%%%%%%%%%%%%

For several of the remaining cases, we start by arranging the $\pp_i$ to give a subdivision at $N=2$ with the correct black triangles, and with all the white polygons planar.  This is possible if the $\kappa_i$ are chosen so that in the $pq$-plane, the segments $\{1,4\},$ $\{2,5\},$ and $\{3,7\}$  intersect at a single point; and the same holds for $\{1,6\}$, $\{5,8\}$, and $\{3,7\}$.  We will call this degenerate subdivision $\Q_{2,8}^*$.

To realize the upper-middle tiling in Fig.~\ref{BadTilings8}, start with $\Q_{2,8}^*$, and lift $\pp_2$.  To refine the resulting tiling, we may first raise or lower $\pp_8$ to achieve the desired triangulation of the quadrilateral with common index $7$, and then adjust the heights of $\{\pp_1,\pp_3,\pp_4,\pp_6, \pp_7\}$ to triangulate the pentagon with common index $5$.  Note that at each step, we can make the height adjustments arbitrarily small, so there is no danger of disturbing the rest of the configuration.

For the tiling at upper right, we start with $\Q_{2,8}^*$, and lower both $\pp_3$ and $\pp_7$. To triangulate the white quadrilaterals, we then adjust $\pp_4$ and $\pp_8$.

The case of the tiling at lower left is slightly more complicated.  To realize this tiling, we first arrange the points $\{\pp_i : i \neq 2, 6\}$
appropriately, such that all the white polygons are planar.  Choose $\kappa_2$ so that $\{2,5\}$ crosses $\{1,4\}$ and $\{3,7\}$ to the left of the point where the latter two segments intersect, where the octagon is oriented as in Fig.~\ref{triangulations8}.  In other words, the segment $\{2,5\}$ crosses $\{3,7\}$ between the vertex $\p_3$ and the intersection of $\{3,7\}$ and $\{1,4\}$.  Then we can assign an appropriate weight to $\pp_2$ so that $\{\hat{2},\hat{5}\}$ passes just above $\{\hat{1},\hat{4}\}$, and hence below $\{\hat{3},\hat{7}\}$ as desired.  By a similar argument, we can add the point $\pp_6$, for an appropriate choice of $\kappa_6$, to produce the desired configuration.  Raising and lowering $\pp_4$ and $\pp_8$, we can refine the tiling as needed.

The case of the lower middle is similar, but simpler; we place all points $\{\pp_i : i \neq 2\}$ as desired, so that all white polygons of the resulting tiling are planar.  We then add $\pp_2$ as in the previous case, with $\kappa_2$ chosen appropriately.  Raising $\pp_1$ slightly gives the desired tiling, which we may then refine by adjusting $\pp_8$.

Finally, for the tiling at lower right, we start with $\Q_{2,8}^*$, and assume that the segments $\{\hat{1},\hat{5}\}$ and $\{\hat{3},\hat{7}\}$ are both parallel to the $pq$-plane.  We adjust $\p_3$ by decreasing $\kappa_3$ slightly, so that $\p_3$ moves toward $\p_2$, without changing the weight of $\pp_3$.  This gives the desired triangulation of the pentagon with common index $1$, and ensures that the diagonal corresponding to $\{1,4\}$ appears in the pentagon with common index $5$.  Lowering $\pp_7$ slightly then gives the desired tiling, and again we can triangulate the quadrilaterals as needed.  This completes the final case, and with it the proof.
\end{proof}

%%%%%%%%%%%%%%%%%%%%%%%%%%%

\section{Non-realizable soliton graphs}
\label{sec: non-realizable}

In this section we show that not all weakly separated collections are realizable.  We are grateful to Hugh Thomas for suggesting a counterexample, which we describe in the proof of Theorem \ref{NotRealizable}.  

\subsection{Combinatorial background}
\label{subsec: combo background}
Before proceeding to the proof, we give some background on pseudoline arrangements, and some additional details about plabic graphs.  For a reference on pseudoline arrangements, see for example \cite{richter}. 

A \emph{pseudoline} is a simple closed curve in the real projective plane $\mathbb{P}^2$ which is topologically equivalent to a line; in particular, a pseudoline has no self-intersections.  An \emph{arrangement} of pseudolines is a collection of pseudolines $\mathcal{L} = (L_1,\ldots,L_n)$ such that for any $1 \leq i < j \leq n$, the pseudolines $L_i$ and $L_j$ intersect exactly once.  A pseudoline arrangement is \emph{simple} if no three pseudolines meet in a common point.  Two pseudoline arrangements are \emph{equivalent} if they generate isomorphic cell decompositions of $\mathbb{P}^2$.  An arrangement of pseudolines is \emph{stretchable} if it is equivalent to an arrangement of projective lines.  Every arrangement of eight pseudolines or fewer is stretchable \cite{goodman}.  However, there is a non-stretchable arrangement of $9$ pseudolines, and hence of $n$ pseudolines for any $n > 9$ \cite{ringel}.

As described in \cite{richter}, we may visualize the real projective plane $\mathbb{P}^2$ as a sphere in $\R^3$ with antipodal points identified, and visualize pseudolines as \emph{great pseudocircles} on the sphere.  Assuming without loss of generality that each pseudoline crosses the equator exactly once, and that no crossing of pseudolines occurs on the equator, we may then restrict ourselves to the upper hemisphere.  Projecting to $\R^2$, we obtain an arrangement of \emph{affine pseudolines}.  We define simplicity, equivalence, and stretchability for arrangements of affine psuedolines in the obvious way.  A non-stretchable arrangement of pseudolines in $\mathbb{P}^2$ gives non-stretchable arrangement of affine pseudolines in $\mathbb{R}^2$.

For the proof of Theorem \ref{NotRealizable}, we need a bit more information about plabic graphs. 

\begin{definition}
A \emph{plabic graph} is a planar graph embedding in a disk, with vertices colored black or white. A plabic graph has $M$ \emph{boundary vertices} located on the boundary of the disk, numberered $1,2,\ldots,M$ in counter-clockwise order.  All boundary vertices have degree one.
\end{definition}

Previously, we did not give a precise definition of a \emph{reduced plabic graph}.  Postnikov originally defined reducedness in terms of certain local transformations of graphs \cite{postnikov}.  He then proved a criterion for being reduced in terms of \emph{trips}. 
 
A \emph{trip} in a plabic graph $G$ is a directed path which turns (maximally) left at each white internal vertex, and (maximally) right at each black internal vertex.  Let $T_i$ denote the trip which starts at boundary vertex $i$, and continues until it reaches the boundary again. The \emph{trip permutation} $\pi$ of $G$ is the permutation defined by setting $i \mapsto j$ if the trip $T_i$ ends at boundary vertex $j$.  Note that a trip in a plabic graph may either be a closed cycle containing no boundary vertices (called a \emph{round trip}), or it may connect two boundary vertices $i$ and $j$.   

We label each face of a plabic graph with an $i$ if it is to the left of the trip that beings at vertex $i$.  For soliton graphs, this recovers the usual face labels \cite{kodama}.

\begin{definition}\cite[Theorem 13.2]{postnikov}
\label{reduced}
A plabic graph is \emph{reduced} if and only if it satisfies the following conditions:
\begin{enumerate}
\item $G$ has no round trips.
\item No trip in $G$ uses the same edge twice (unless that edge connects a boundary vertex to an adjacent leaf).
\item No two in $G$ trips have a pair of common edges $(e_1,e_2)$, where both trips are directed from $e_1$ to $e_2$.
\end{enumerate}
\end{definition}

The trips $T_i$ in a reduced plabic graph $G$ induce a permutation $\pi$ on the boundary vertices, defined by setting $\pi(i) = j$ if the trip $T_i$ ends at boundary vertex $j$.  For soliton graphs, this gives the usual permutation. With these conventions, $G$ is a reduced plabic graph for $\Gr(N,M)_{>0}$ if and only if the trip permutation of $G$ is $i \mapsto i - N$, where indices are taken modulo $M$ \cite[Lemma 17.6]{postnikov}.

\subsection{A non-realizable soliton graph}
\label{subsec: counterexample}

\begin{theorem}{\label{NotRealizable}}
For every $N \geq 9$, there exists a plabic graph for $\Gr(N, 2N)_{>0}$ which is not a soliton graph, even up to contraction equivalence.  Equivalently, there exists a weakly separated collection for $\Gr(N,2N)_{>0}$ which is not realizable.
\end{theorem}

\begin{proof}
Consider a simple, non-stretchable arrangement $\mathcal{L}$ of $N$ affine pseudolines in the Euclidean plane.  Without loss of generality, assume we can construct a circle $C$ such that:
\begin{enumerate}
\item $C$ encloses all intersections of pseudolines in $\mathcal{L}$,
\item Each pseudoline in $\mathcal{L}$ intersects $C$ exactly twice, and
\item All intersections of the pseudolines in $\mathcal{L}$ with the circle $C$ are transversal.
\end{enumerate}

We erase the part of each pseudoline outside of $C$, place a boundary vertex at each intersection of a pseudoline with $C$, and label the boundary vertices $1,2,\ldots,2N$ in counterclockwise order.  Next, we replace each intersection of pseudolines with a bicolored square, as shown in Fig.~\ref{addsquare}. Let $G$ be the resulting graph, which is embedded in a disk with boundary $C$.  

%%%%%%%%%%%%%%%%%%%%%%%
\begin{figure}[ht]
\centering
\includegraphics[scale = 0.75, trim = {1in 9in 1in 1in}, clip]{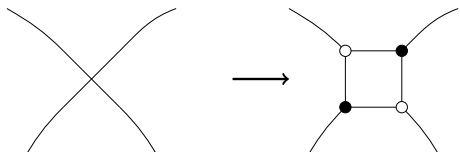}
\caption{Replacing a crossing between two Pseudolines with a black-white square.}
\label{addsquare}
\end{figure} 
%%%%%%%%%%%%%%%%%%%%%

We claim that $G$ is a reduced plabic graph.  First, note that each pseudoline in $\mathcal{L}$ connects some boundary vertex $k$ to the boundary vertex $k - N$, where indices are taken modulo $2N$. Label the pseudolines in $\mathcal{L}$ as $L_1,\ldots,L_N$, with indices taken modulo $N$, so that $L_k$ contains boundary vertex $k$.  The trip $T_k$ in $G$ follows $L_k$, taking a detour around two sides of each added square which intersects the pseudoline.  Hence no trip in $G$ crosses itself.  The common edges of the trips $T_k$ and $T_{k - N}$ are precisely those which correspond to segments of $L_k$, and $T_k$ and $T_{k - N}$ pass through those edges in opposite order.  

If $k \not\equiv \ell \pmod{2N}$, then $T_k$ and $T_{\ell}$ have a single common edge; this edge occurs in the square corresponding to the intersection of $L_k$ and $L_{\ell}$.  If follows from \cite[Theorem 13.2]{postnikov} that $G$ is a reduced plabic graph.  The trip permutation of $G$ is defined by $k \mapsto k - N$, so by \cite[Lemma 17.6]{postnikov}, the plabic graph $G$ corresponds to $\Gr(N,2N)_{>0}$.

Assume for the sake of contradiction that $G$ is a soliton graph, up to contraction equivalence.  In the corresponding contour plot, each edge in $T_k$ which represents a segment of $L_k$ separates a region where $\Theta_k$ is dominant from one where $\Theta_{k -N}$ is dominant.  Hence each such edge is a segment of the line defined by $\Theta_k = \Theta_{k - N}$.  Moreover, if $k \not\equiv \ell \pmod{2N}$, then the line $\Theta_k = \Theta_{k - N}$ must intersect the line $\Theta_{\ell} = \Theta_{\ell - N}$ inside the square corresponding to the intersection of $L_k$ and $L_{\ell}$. Hence, replacing $L_k$ with the line $\Theta_k = \Theta_{k-N}$ gives a stretching of the affine pseudoline arrangement $\mathcal{L}$.  (Note that we may contract any unicolor edges without affecting the substance of the argument, since each trip must still pass through the corresponding vertices of each black-white square after an edge-contraction move.) This is a contradiction, and the proof is complete.
\end{proof}

The smallest counterexample given by the proof of Theorem \ref{NotRealizable} is a plabic graph for $\Gr(9,18)_{>0}$.  We conjecture that much smaller non-realizable plabic graphs exist, but have yet to find them.

\raggedright

\end{document}